\documentclass[11pt]{article}

\include{epic.sty}
\usepackage{amsthm}
\usepackage{amssymb,amsmath,hyperref,srcltx,accents}
\usepackage{amsfonts,stmaryrd,mathrsfs}
\usepackage{tikz}
\usepackage{subfigure}
\usepackage{algorithm,algorithmicx}
\usepackage{algpseudocode}

\usetikzlibrary{patterns}

\newlength{\originalbase}
\newcommand{\remove}[1]{}

\newcommand{\dist}[2]{d({#1},{#2})}

\newcommand{\vertline}[1]{Y({#1})}
\newcommand{\horizline}[1]{X({#1})}
\newcommand{\horizlineneg}[1]{X_{-}({#1})}
\newcommand{\sweepline}[1]{S({#1})}

\newcommand{\spokeline}[1]{S_{#1}}
\newcommand{\spokevertline}[1]{Y_{#1}}
\newcommand{\transline}[1]{T_{#1}}

\def\frontier{X}

\def\cP{{\mathcal P}}
\def\cE{{\mathcal E}}
\def\cA{{\mathcal A}}
\def\cB{{\mathcal B}}

\def\R{{\mathbb R}}

\def\dx{\Delta x}
\def\dy{\Delta y}

\def\sp{\Pi}

\def\proj{\mathrm{proj}}
\def\Area{\mathrm{area}}

\newtheorem{theorem}{Theorem}[section]

\newtheorem{lemma}[theorem]{Lemma}

\newtheorem{definition}{Definition}[theorem]

\begin{document}

\def\proofend{\hfill$\Box$\medskip}
\def\Proof{\noindent{\bf Proof. }}
\newcommand{\ProofOf}[1]{\noindent{\bf Proof of {#1}. }}
\newenvironment{proofof}[1]{\ProofOf{#1}}{\hfill $\Box$ \medskip}

\newcommand{\style}[1]{{\bf{{#1}}}}

\def\proj{\mathrm{proj}}
\def\diam{\mathrm{diam}}




\title{Line-of-Sight Pursuit in Monotone and Scallop Polygons}

\date{}

\author{Lindsay Berry\thanks{Statistical Science Department, Duke University Durham, NC 27708, USA} \and Andrew Beveridge\thanks{Department of Mathematics, Statistics, and Computer Science, Macalester College, St. Paul MN 55105} \and Jane Butterfield\thanks{Department of Mathematics \& Statistics, University of Victoria, Victoria, B.C. V8X 2X6} \and Volkan Isler\thanks{Department of Computer Science \& Engineering, University of Minnesota, Minneapolis MN 55455} \and Zachary Keller\thanks{Department of Mathematics, University of Minnesota, Minneapolis MN 55455} \and Alana Shine\thanks{Department of Computer Science, University of Southern California, Los Angeles, CA 90089, USA} \and Junyi Wang$^\dagger$}

\maketitle

\begin{abstract}
We study a turn-based  game in a simply connected  polygonal environment $Q$ between  a pursuer $\cP$ and an adversarial evader $\cE$. 
Both players can move in a straight line to any point within unit distance during their turn. The pursuer $\cP$ wins by capturing the evader, meaning that their distance satisfies $d(\cP, \cE) \leq 1$, while the evader wins by eluding capture forever. 
Both players have a map of the environment, but they have different sensing capabilities. The evader $\cE$ always knows the location of $\cP$. Meanwhile, $\cP$ only has line-of-sight visibility:  $\cP$ observes the evader's position only when the line segment connecting them lies entirely within the polygon. Therefore $\cP$ must search for $\cE$ when the evader is hidden from view.

We provide a winning strategy for $\cP$ in two families of polygons: \style{monotone polygons} and \style{scallop polygons}. In both families, a straight line $L$ can be moved continuously over $Q$ so that (1) $L \cap Q$ is a line segment and (2) every point on the boundary $\partial Q$ is swept exactly once. These are both subfamilies of \style{strictly sweepable  polygons}.  The sweeping motion for a monotone polygon  is a single translation, and the sweeping motion for a scallop polygon is a single rotation.  
Our algorithms use \style{rook's strategy} during its  pursuit phase, rather than the well-known \style{lion's strategy}. The rook's strategy is crucial for obtaining a capture time that is linear in the area of $Q$. For both monotone and scallop polygons, our algorithm has a capture time of $O(n(Q) + \Area(Q))$, where  $n(Q)$ is the number of polygon vertices.

\end{abstract}

%

\section{Introduction}

We study a turn based pursuit-evasion game in which a  pursuer $\cP$ chases an evader $\cE$ in a simply connected polygon $Q$ with boundary $\partial Q$. 
The players have unit speed, so they can move to any point within unit distance of their current location, provided that the line segment connecting these points is contained in $Q$.  The pursuer's goal is to catch the evader in finite time  using a deterministic strategy, while the evader tries to avoid capture indefinitely.
Both players have full knowledge of the layout of environment, but they have  asymmetric sensing capabilities. The evader has \style{full-visibility}, meaning that $\cE$ knows the location of $\cP$ at all times. The pursuer only has \style{line-of-sight visibility}: she knows the exact location of $\cE$ only when the line segment  $\overline{\cP\cE} \subset Q$. Therefore, the pursuer must search when $\cE$ is obscured by features of the environment. The capture condition requires both proximity and visibility: the evader is captured when  the distance $d(\cP, \cE) \leq 1$ and  $\overline{\cP\cE} \subset Q$.

The line-of-sight restriction puts the pursuer at a great disadvantage. Isler et al.~provide a winning strategy for a full-visibility pursuer $\cP$ in a simply connected polygon, called  \style{lion's strategy}.\cite{isler05tro}  They also construct environments for which a pursuer $\cP$ with line-of-sight visibility does not have a deterministic winning strategy.   Klein and Suri provided an upper bound for the number of pursuers required for this setting:  $O(n^{1/2})$ line-of-sight pursuers have a winning strategy in any simply connected polygon.\cite{klein+suri}    They construct a comb-like environment (with visibility-blocking notches in its corridors) that requires $\Theta(n^{1/2})$ pursuers, so their bound is tight.

This paper addresses the other extreme: when can a single line-of-sight pursuer use a deterministic strategy to catch an evader? 
We provide a winning pursuer strategy for two subfamilies family of strictly sweepable polygons. Strictly sweepable polygons   were introduced by Bose and van Krevald.\cite{bose}
A polygon $Q$ is \style{sweepable} if a straight line $L$ can be moved continuously over $Q$ so that $Q \cap L$ is a line segment. The sweeping motion can include translations and rotations of the sweep line.   A polygon is \style{strictly sweepable} if every point of $Q$ is swept exactly once during this process. See Figure \ref{fig:sweepable} for examples of sweepable and strictly sweepable polygons. We build up to this result, first describing winning pursuit strategies in two simpler subfamilies. A \style{monotone polygon} is a strictly sweepable polygon where the sweep line moves via translation along a fixed line $L_0$, see  Figure \ref{fig:mono-scallop}(a).  Our second family is  \style{scallop polygons}, where the sweep line simply rotates through a fixed center $C$ located outside of the polygon, see Figure \ref{fig:mono-scallop}(b). In this case, the total angle swept must be less than $\pi$ to maintain the convex intersection with $Q$.  

\begin{figure}

\begin{center}
\begin{tikzpicture}[scale=.75]

\begin{scope}

\foreach \ang in {190, 215, 240, 270} {
\draw[dashed] (\ang:0) -- (\ang:3.75);

}

\foreach \ang in {190, 200, 210, 220, 230, 240, 250, 260} {
\draw[thick] (\ang:2.5) -- (\ang+5:3.5) -- (\ang+10: 2.5);
}

\draw[thick] (190:2.5) --  (270:1);

\end{scope}

\begin{scope}[shift={(270:1)}]

\draw[thick] (0,0) -- (2,0);
\draw[thick] (0,-1.5) -- (2,-1.5);

\node at (1,-3) {\scriptsize (b)};

\end{scope}

\begin{scope}[shift={(2,-3.5)}]

\foreach \ang in {10, 35, 60, 90} {
\draw[dashed] (\ang:0) -- (\ang:3.75);

}

\foreach \ang in {10, 20, 30, 40, 50, 60, 70, 80} {
\draw[thick] (\ang:2.5) -- (\ang+5:3.5) -- (\ang+10: 2.5);
}

\draw[thick] (10:2.5) --  (90:1);

\end{scope}


\begin{scope}[shift={(-10,-.5)}]

\draw[thick] (0,0) -- (5,0) -- (5,-2) -- (4,-2);

\draw[thick] (1, -2) -- (.25, -2.5) -- (1.25, -2) -- (2.25, -2.5) -- (1.5,-2);

\draw[thick] (4, -2) -- (4.75, -2.5) -- (3.75, -2) -- (2.75, -2.5) -- (3.5,-2);

\draw[thick] (3.5,-2) -- (1.5,-2);

\draw[thick] (1,-2) -- (0,-2) -- (0,0);

\node at (2.5, -3.5) {\scriptsize (a)};

\end{scope}

\end{tikzpicture}

\end{center}

\caption{(a) A sweepable polygon that is not strictly sweepable. (b) A strictly sweepable polygon. }

\label{fig:sweepable}

\end{figure}


\begin{figure}[htp]
  \begin{center}
    	\begin{tikzpicture}[scale=.75]

		\draw[thick] (0.75,0.5) -- (1.3,2.3) -- (1.7,1.8) -- (2.1,2.4) -- (2.6, -.5) -- (3.2, 2.5) -- (3.5, 1.5) -- (3.75,2.3) -- (4.5, 0) ;
		\draw[thick] (0.75,0.5) -- (1.1,-2) -- (1.4, -.75)  -- (1.9,0) -- (3,-2.25) -- (3.25, 0.0) -- (3.5,.1) -- (3.7, -0.05) -- (4.2,-2) -- (4.5, 0) ;
	  	\draw[dashed] (1,-2.4) -- (1,2.8);
	  	\draw[dashed] (2.5,-2.4) -- (2.5,2.8);	
	  	\draw[dashed] (4,-2.4) -- (4,2.8);

\draw[thick] (.25,.25) -- (5,.25);
\node at (-.1,.2) {\small $L_0$};

\node[above] at (0.55,0.5) {\small $v_1$};
\node[right] at (1.05,-2.1) {\small $v_2$};
\node[above] at (1.4,2.2) {\small $v_3$};
\node[right] at (4.5, -.25) {\small $v_{17}$};

\node at (2.75,-3.5) {\scriptsize (a)};


\begin{scope}[shift={(11,-2.5)}]

\foreach \ang in {30, 45, 60, 75, 90, 105, 120, 135, 150} {
\draw[dashed] (\ang:0) -- (\ang:4);

}

\foreach \ang in {30, 50, 70, 90, 110, 130} {
\draw[thick] (\ang:2.5) -- (\ang+10:4) -- (\ang+20: 2.5);
}

\foreach \ang in {40, 65,  90,  115} {
\draw[thick] (\ang:1) -- (\ang+12.5:1.5) -- (\ang+25:1);
}

\draw[thick] (30:2.5) -- (40:1);
\draw[thick] (150:2.5) -- (140:1);

\node[below] at (150:2.5) {\small $v_{1}$};
\node[above] at (145:4.2) {\small $v_{2}$};
\node[below] at (27:2.75) {\small $v_{21}$};

\node at (0,-1) {\scriptsize (b)};

\end{scope}

\end{tikzpicture}

\end{center}

  \caption{(a) A monotone polygon with respect to horizontal axis $L_0$ with vertices indexed from left to right. (b) A scallop polygon with vertices  indexed by  decreasing polar angle.}

\label{fig:mono-scallop}

\end{figure}

\begin{theorem}
\label{thm:monotone}
A line-of-sight pursuer can catch an evader in a monotone polygon $Q$ with capture time $O(n(Q) + \Area(Q))$.
\end{theorem}

\begin{theorem}
\label{thm:scallop}
A line-of-sight pursuer can catch an evader in a scallop polygon $Q$ with capture time $O(n(Q) + \Area(Q))$.
\end{theorem}

Our pursuit algorithms for monotone and scallop polygons follow the same high-level strategy.  The algorithms  alternate between a \style{Search Mode} and an active pursuit mode, called \style{Rook Mode}. During Search Mode, the pursuer methodically clears the polygon from left-to-right, maintaining a series of \style{checkpoints} that mark her progress. When $\cP$ enters Rook Mode, her movements still guard these checkpoints. During Rook Mode, the evader can use features of the environment to hinder the pursuer. For example,  the pursuer can be blocked by the boundary. In scallop  polygons, the evader can also hide behind features. In response, $\cP$ reverts back into Search Mode, carefully ensuring that she maintains her checkpoint progress. 
Maintaining the progress is straight-forward in monotone polygons, but is trickier in scallop polygons. 
Eventually, the evader will be pushed into a subpolygon where the pursuer can capture him. Figure \ref{fig:example-monotone} shows a typical pursuit trajectory in a monotone polygon. The pursuer's Search Modes are shown in black and the Rook Modes are shown in gray. Note that searching always progresses from left-to-right. While in Rook Mode, $\cP$ guards her current horizontal position, moving left or right to track the evader's movement, and advancing towards the evader when given the opportunity.

\begin{figure}
  \begin{center}
  \begin{tikzpicture}[scale=1]

\draw (0,1) -- (0.5,2) -- (1,0) -- (1.3,1.8) -- (1.7,1.3) -- (2.1,1.9) -- (2.6, .2) -- (3.2, 2) -- (3.5, 1) -- (3.75,1.8) -- (3.9,.5) --(4.3,-.2) -- (5, 2.1) -- (5.5,-.75) -- (6.25,1.5) -- (6.5,0);
\draw (0,1) -- (0.45,-.3) -- (.6,.2) -- (1.1,-1.5) -- (1.4, -.4) -- (1.7,-1) -- (1.9,0.5) -- (3,-1.75) -- (3.5,-.75) -- (4.2,-1.5) -- (4.9, -.5) -- (5.5,-1.6) -- (5.75,-1) -- (6,-1.7) -- (6.5,0);

\draw[very thick] (0,1) -- (.75,1) -- (1,0) -- (1.5,0);

\draw[very thick, -latex] (0,1) -- (.55,1);
	
\draw[color=gray, very thick] (1.5, 0) -- (1.5, .15) -- (1.2, .15) -- (1.3, .28) -- (1.87, .28);

\draw[color=gray, very thick, -latex] (1.3, .28) -- (1.7, .28);

 \draw[very thick]  (1.87, .28) -- (1.9,0.5) -- (2.51176,0.5) -- (2.6,0.2) -- (3.3,0.2);

\draw[color=gray, very thick] (3.3,0.2) -- (4.0,0.2) -- (4.0, 0.05) -- (2.3, 0.05) -- (2.3, -0.1) -- (4.25, -0.1);

\draw[color=gray, very thick, -latex] (4.0, 0.05) -- (2.9, 0.05);
	
\draw[very thick]  (4.25, -0.1) -- (4.3,-.2) -- (5.25,-0.2);	

\draw[color=gray, very thick] (5.25,-0.2) -- (4.6, .4) -- (5, .6) -- (5.17, .6) -- (5.17, .75) -- (4.68, .75) -- (4.68, .9) -- (5.145, .9) -- (5.145, 1.05) -- (4.8, 1.05) -- (4.8, 1.2) -- (5.05, 1.2) -- (5.05, 1.35) -- (4.87, 1.35) -- (4.87, 1.5) -- (5, 1.5) -- (4.98, 1.82)
;

\draw[color=gray, very thick, -latex] (5.25,-0.2) -- (4.7, .3);

\draw[fill] (0,1) circle (1.25pt);
\draw[fill=white] (1.5,0) circle (1.25pt);
\draw[gray, fill=gray] (1.87, .28) circle (1.25pt);
\draw[fill=white] (3.3,0.2) circle (1.25pt);
\draw[gray, fill=gray] (4.25, -0.1) circle (1.25pt);
\draw[fill=white] (5.25,-0.2) circle (1.25pt);
\draw[fill] (4.98, 1.82) circle (1.25pt);

\end{tikzpicture}

\end{center}

\caption{An example pursuit trajectory in which the evader makes two blocking moves. The black segments are Search Mode and the gray segments are Rook Mode.  The white points indicate the transition from Search Mode to Rook Mode. The gray points indicate where the pursuer is blocked and transitions back to Search Mode. The evader is captured at the final black point. Note that the pursuer never visits the same point twice.}

\label{fig:example-monotone}

\end{figure}

The proofs for each polygon family can be found in the two sections that follow, with our arguments building successively.  Monotone pursuit is the simplest case, so we prove Theorem \ref{thm:monotone} in Section \ref{sec:mono-pursuit}. Previously, Noori and Isler proved that monotone polygons are pursuer-win for a line-of-sight pursuer.\cite{noori+isler}  Their algorithm guarantees a capture time of $O(n(Q)^7  \cdot \diam(Q)^{13})$ where  $\diam(Q) = \max_{x,y \in Q} d(x,y)$. Our  pursuit strategy  is much simpler, as is reflected in our significantly reduced capture time. The key improvement is a novel  chasing strategy for pursuit in polygonal environments. Rather that using lion's strategy in the endgame, we use a different tactic that is better suited to the monotone  setting. This \style{rook's strategy} chases the evader while simultaneously guarding a horizontal frontier from incursion by the evader. This allows for a seamless transition between searching for the evader and actively chasing him. More specifically, our rook phase also guards the pursuer's progress, while Noori and Isler needed a third mode to defend against certain pursuer gambits, which increased their capture time bound.

In Section \ref{sec:scallop}, we adapt the monotone pursuit strategy for scallop polygons. The guiding principle of the scallop strategy is to update the pursuer's frame of reference as she moves through the polygon. This update takes some care, as $\cP$ must reposition herself to continue to guard her progress. We conclude in Section \ref{sec:conclusion} with some observations and open problems. Before launching into the proofs, we discuss related work and some  environmental assumptions.

\subsection{Related work}

Polygonal environments with nice geometric properties (such as convexity) provide an opportunity to simplify algorithms, or even make feasible an otherwise intractable problem. Chapter 30 of the Handbook of Computational Geometry surveys the theory and applications of such polygons.\cite{orourke+suri+toth} One generalization of convexity is \style{monotonicity}:  polygon $P$ is monotone with respect to line $L$ when every line orthogonal to $L$ intersects $P$ in a line segment. A polygon $P$ is \style{star-shaped} when $P$ contains a point $x$ such that every half line emanating from $x$ intersects $P$ in a single line segment. Monotonicity and star-shapedness can each be determined in $O(n)$ time.\cite{preparata+supowit,lee+preparata} A polygon $P$ is \style{radially monotone} (or equivalently, \style{pseudo-star shaped}) when there exists a point $z$ in the plane such that every half line emanating from $z$ intersects $P$ in a line segment.\cite{elgindy+toussiant} Radial monotonicity can be determined in  $O(n^2)$ time.\cite{dean+lingas+sack}

We  focus on the generalization from monotone to strictly sweepable, as defined in Bose and and van Kreveld.\cite{bose} A polygon $Q$ is strictly sweepable when a straight line $L$ can be moved continuously (via translation and rotation) so that every point of $Q$ is swept exactly once. Bose and van Kreveld show that strict sweepability can be determined in $O(n^2)$ time. Note that if the sweep line only moves by translation then $Q$ is a monotone polygon. Meanwhile, when the sweep line only moves by rotation, the polygon is rotationally monotone, with rotation angle smaller than $\pi$ radians (because we use a rotating line instead of a half line). We refer to this subfamily of radially monotone polygons as \style{scallop polygons}.

The literature on pursuit-evasion games spans a  variety of settings and disciplines, with various motion and sensing capabilities.
For a given environment, the main questions are (1) to determine the minimum number of pursuers needed to capture the evader, and (2) to bound the \style{capture time}, which is the number of rounds needed to catch the evader. Diverse game models have been studied, considering different types of environments, motion constraints, sensing capabilities and capture conditions. We list a sampling of these variations, providing references below. Researchers have explored  speed differentials between the players and constraints on acceleration. As for sensing models, the players may have full information about the positions of other players, or they may have incomplete or imperfect information. Typically, the capture condition requires achieving colocation,  a proximity  threshold, or sensory visibility (such as a non-obstructed view of the evader).

 The original pursuit-evasion setting is Rado's lion-and-man game from the 1930's. In this game, a lion chases a man in a circular arena, with each player moving with unit speed. The lion wins if it becomes colocated with man in finite time, while the man tries to evade the lion forever. Intuition suggests that the lion should be victorious. However, Besicovich showed that if the game is played in continuous time, then the man can gently spiral away from the center, so that the lion becomes arbitrarily close, but never achieves colocation.\cite{littlewood} Meanwhile, the turn-based version avoids this pathology:   Sgall showed that the lion is victorious in finite time.\cite{sgall}

Pursuit-evasion games have enjoyed significant recent attention in computational geometry, combinatorics and topology. The computational geometry viewpoint, which considers the turn-based pursuit game between autonomous agents in a polygonal environment,  has been embraced by robotics researchers and theoretical computer scientists.  Interest in polygonal pursuit games has been spurred by the widespread availability of practical sensing technologies and robotics platforms. Applications for automated search and pursuit are on the rise, including intruder neutralization, search-and-rescue, and tracking of tagged wildlife. Analysis of games with an adversarial evader provide important worst-case bounds for these applications.\cite{chi}

When the game is played in $\R^d$, Kopparty and Ravishankar\cite{KR} showed that $d+1$ pursuers win when the evader starts in the convex hull of their initial location, extending the classic result of Jankovic\cite{Jankovic}  concerning pursuit-evasion in $\R^2$.
Isler et al.~adapted the lion's strategy of Sgall to simply connected polygons.\cite{isler05tro}    More recently, Noori and Isler employed a novel pursuer strategy called \style{rook's strategy} for  pursuit on surfaces and convex terrains.\cite{noori+isler-terrain,noori+isler-polysurface}  This strategy 
is valid whenever the capture distance is nonzero. 

In recent years, researchers have studied polygonal pursuit-evasion with sensing limitations. For example, Bopardikar, Bullo and Hespanha considered a pursuer with a limited sensing distance.\cite{bopardikar} Guibas et al.~introduced pursuit using line-of-sight visibility, where capture is equivalent to detection: the evader is captured as soon as he is visible.\cite{guibas} 
Stiffler and O'Kane  determine when $m$ line-of-sight pursuers have a winning detection strategy in a polygonal environment.\cite{stiffler+okane}
 Using the more restrictive colocation capture condition,   Noori and Isler proved that a single  line-of-sight pursuer is successful in monotone polygons.\cite{noori+isler} Herein,  we add to the study of line-of-sight pursuit by resolving co-location capture in strictly sweepable polygons.
We consider agents with equal speeds, but others researcher have studied pursuit-evasion in polygonal environment with various movement constraints. For example, Tovar and LaValle show that for any polygonal environment, there exists a speed ratio such that a line-of-sight pursuer can achieve detection capture.\cite{tovar+lavalle}

Meanwhile, combinatorics researchers have devoted intense scrutiny to pursuit-evasion on graphs over the past four decades. Parsons introduced the graph searching variant, where the evader is invisible to the pursuers.\cite{parsons} The full visibility case,  known as the game of cops and robbers, was was introduced  by Quillot and independently by Nowakowski and Winkler.\cite{quilliot,nowakowski+winkler}
For an overview of pursuit-evasion on graphs, see the monograph by Bonato and Nowakowski.\cite{bonato+nowakowski} More recently,  topologists have characterized capture strategies  for    spaces  with various curvature conditions.\cite{abg1,abg2}
There has certainly been  cross-pollination between these traditions. For example, the analog of the classic Aigner and Fromme\cite{aigner+fromme}  result that three cops can always capture a robber on a planar graph was recently proven to hold for the polygonal setting by Bhaudauria et al.\cite{bhadauria+klein+isler+suri} and further generalized to a class of compact two-dimensional domains by Beveridge and Cai.\cite{beveridge+cai}

\subsection{Environmental assumptions}

We conclude this section with some comments on the game model and our assumptions about the geometric environment. 

First, we assume that there is a relationship between the geometry of the environment and the speed of the players. In particular, the environment has a \style{minimum feature size}, meaning that  the minimum distance between any two polygonal vertices is at least one unit. This assumption can be achieved by a simple rescaling of the environment, if necessary. Klein and Suri pointed out that this minimum feature size has been an implicit assumption in many papers.\cite{klein+suri2}  Furthermore, minimum feature size is necessary to avoid an unexpectedly powerful evader who can use his super-speed to confound a line-of-sight pursuer. Of course, this assumption is a natural one: we view  the turn-based game as an approximation of the continuous time, so we  should choose a time scaling that appropriately partitions the dynamics of movement in the environment.

Second, recall that capture in our game requires two conditions:   $d(\cP,\cE) \leq 1$ and  $\cP$ sees $\cE$. We employ unit capture distance to simplify the presentation of our proofs. Our results also hold when we replace the first condition with $d(\cP,\cE) \leq \epsilon$ for any positive capture radius $\epsilon >0$.  The only impact of a smaller capture radius $\epsilon$ is a multiplicative factor of $1/\epsilon$ in the capture time.  We note that the polygonal literature uses proximity capture  ($d(\cP,\cE) \leq \epsilon$) and colocation capture ($d(\cP,\cE)=0$) with roughly equal frequency. Herein, we take additional advantage of the former capture condition in developing the rook's strategy for chasing the evader. In particular,  the buffer distance between the players plays an essential role in the rook's strategy. 
In practice, the players would have non-zero radius so this assumption always holds.

Finally, make two modest geometric assumptions about our environments. We assume that our vertices are in \style{general position}, so that no three are colinear. For scallop polygons, this includes the rotational sweep center $C$. We also assume that $|\Area(Q)| \geq |\diam(Q)|$  to simplify our time bound notation. 
This will generally be the case, except for long, skinny environments. Note that $\Area(Q) \leq \diam(Q)^2$ whenever $\diam(Q) \geq 1$, so we could use the latter value for our capture time bounds, if desired. We state our bounds in terms of area since our analysis will naturally partition the polygon into disjoint regions.

%

\section{Pursuit in a Monotone Polygon}
\label{sec:mono-pursuit}

In this section, we prove Theorem \ref{thm:monotone}: $\cP$ can capture $\cE$ in a monotone polygon. 
We start by setting some notation. We fix our monotone axis $L_0$ and then label our vertices from left to right as $v_1, v_2, \ldots , v_n$.
The monotonicity of $Q$ has a simple characterization with respect to $\partial Q$: the boundary can be partitioned into two piecewise linear functions defined between $v_1$ and $v_n$. These functions only intersect at $v_1$ and $v_n$, so we refer to them as the \style{upper chain} $\Pi_U$ and the \style{lower chain} $\Pi_L$. 

Given a point $Z$, let $x(Z)$ and $y(Z)$ denote its $x$- and $y$-coordinates respectively, so that 
$Z=(x(Z), y(Z))$. The horizontal $x$-axis aligns with the monotone axis $L_0$. 

\begin{definition}
Let $Q$ be a monotone  polygon. Given $Z \in Q$, the \style{frame} $(X(Z), Y(Z))$ is the pair of maximal horizontal and vertical line segments through $Z$, that is
\begin{align*}
\horizline{Z} &= \{ V \in Q : y(V) = y(Z) \mbox{ and } \overline{VZ} \subset Q \}, \\
\vertline{Z} &= \{ V \in Q : x(V) = x(Z) \mbox{ and } \overline{VZ} \subset Q \}. 
\end{align*}

We define $\horizlineneg{Z} = \{ W  \in X(Z) :  x(W)  \leq x(Z) \}$ and
$X_{+}(Z) = \{ W  \in X(Z) :  x(W)  \geq x(Z) \}$ to be  the points on  $X(Z)$ to the left and  right of $Z$  (including point $Z$). Analogously we define $Y_{-}(Z)$ and $Y_{+}(Z)$ to be the points below and  above $Z$, respectively. 
\end{definition}

We view the line segments $X(Z)$ and $Y(Z)$ as the $x$-axis and $y$-axis  for the frame of reference centered at $Z$. The segments $\horizlineneg{Z}$ and $X_{+}(Z)$ (resp.~$Y_{-}(Z)$ and $Y_{+}(Z)$) are the negative $x$-axis and positive $x$-axis (resp.~negative $y$-axis and positive $y$-axis).

Note that since $Q$ is a closed set, $\horizline{Z}$ extends past any local extrema. Likewise, $\vertline{Z}$ extends past any reflex vertex $v$ for which the local boundary near $v$ lies completely to the left or right of $Y(v)$.

The symbols $\cP$ and $\cE$ denote the pursuer and the evader respectively.  We are often concerned with these positions at some time $t \geq 0$, and will use $\cP_t = (x(\cP_t), y(\cP_t))$ and $\cE_t =  (x(\cE_t), y(\cE_t))$ to denote these points in the polygon.  We define 
\begin{align*}
\dx(t) &= |x(\cP_t) - x(\cE_t)|, \\ 
\dy(t) &= |y(\cP_t) - y(\cE_t)|.
\end{align*}  
Occasionally,  we will drop the subscript and use $\cP, \cE$ to denote the player positions in order to ease the exposition.

\subsection{Overview}

The high-level pursuer strategy is given in Algorithm \ref{alg:monotonePursuit}. We start with a qualitative description of the pursuit, tackling the details of this algorithm in the subsections that follow.  Our pursuer algorithm in a monotone polygon has much in common with the algorithm presented by Noori and Isler.\cite{noori+isler} Both alternate between a searching mode and a chasing mode. In the \style{Search Mode}, $\cP$ traverses left-to-right along a specified search path. At certain times during Search Mode, $\cP$ is able to mark a region to her left as \style{guarded}, meaning $\cE$ cannot enter this region without being captured. These checkpoints are necessary for showing that the pursuit terminates in finite time. 

This methodical advancement continues until $\cP$ establishes the criteria to enter our chasing mode, called  \style{Rook Mode}. This requires (1) visibility of $\cE$ and (2) being positioned in a manner that protects the guarded region. Once $\cP$ has transitioned into Rook Mode, she chases the evader until $\cE$ makes a rightward move that obstructs her pursuit (either by hiding behind a vertex or by using the boundary to block the pursuer's responding move). This forces $\cP$ to re-enter Search Mode. Crucially, this is done in a way that protects (and perhaps updates) the guarded territory.   The algorithm continues, and the evader territory shrinks over time. Eventually, $\cE$ is trapped in a subregion where he cannot foil Rook Mode any longer, and $\cP$ captures him.

\begin{algorithm}
\caption{  Pursuit Strategy}
\label{alg:monotonePursuit}
\begin{algorithmic}[1]
\Require $\cP$ starts out at left endpoint $v_1$
\While {$\cE$ has not been captured}
\State Create  Search Path from the  current location of $\cP$
\While {$\cP$ has not attained rook position}
\State Follow Search Strategy
\EndWhile
\While {$\cE$ does not make an escape move}
\State Follow Rook's Strategy
\EndWhile
\EndWhile
\end{algorithmic}
\end{algorithm}

The fundamental difference between our algorithm and that of Noori and Isler  is the choice of chasing mode. The method of chasing informs the trajectory for searching, so the search paths are also different. Noori and Isler use lion's strategy to chase the evader.  In lion's strategy, the pursuer stays on the shortest path between the leftmost point $v_1$ and the evader. Noori and Isler need an additional guarding mode to establish this position after spying the evader. This guarding mode is quite involved (and contributes to their worst case $O(n(Q)^7 \cdot \diam(Q)^{13})$ capture time), as the evader has multiple gambits to threaten the pursuer's guarded region. However, once $\cP$ has established lion's position, she can reliably push the evader to the right, ending in capture or an evader hiding move (which triggers a transition back to Search Mode).

Instead of lion's strategy, we use a chasing tactic called \style{rook's strategy} that is better suited for our monotone setting. 
This strategy has been used successfully on polyhedral surfaces for the full visibility case,\cite{noori+isler-terrain,noori+isler-polysurface} but this marks its first appearance for polygonal pursuit and with visibility constraints. Rook's strategy draws inspiration from the chess endgame of a black rook and black king versus a solitary white king. During this endgame, the black rook reduces the area available to the white king, one row at a time (with occasional support of the black king). Eventually, the white king is confined to a single row, where black can  checkmate. 
At any given time, the rook guards a horizontal row on the chessboard. Likewise, a pursuer in rook position will also guard a horizontal frontier. This is a natural complement to the pursuer's horizontal search trajectory. Rook's strategy eliminates the challenging transition that  Noori and Isler mitigated with their guarding phase, which dealt with the mismatch of their linear search path and the curved boundary of their guarded region.  Our improved capture time of $O(n(Q) + \Area(Q) )$ directly reflects the seamless transition between our Search Mode and our Rook Mode. 

\begin{definition}
\label{def:rook}
Let $Q$ be a monotone polygon with a vertical sweep line.
The pursuer $\cP$ is in \style{rook position} when $|x(\cP) - x(\cE)| \leq 1/2$ and the line segment $\overline{\cP \cE} \subset Q$, so that $\cP$ sees $\cE$ . The difference $x(\cP) - x(\cE)$ is the pursuer \style{offset} and  the horizontal line  segment $\horizline{\cP}$ is the \style{rook frontier}. 
\end{definition}

Figure \ref{fig:rook-position}(a) shows an example of a pursuer in rook position.  In a convex polygon (where visibility is not an issue), rook position has three important characteristics. First, a pursuer in rook position can re-establish this position after every evader move. Second, the evader cannot get too close to the rook frontier (or cross it) without being captured. Third, the pursuer can reliably advance her rook frontier: the worst case scenario is when $\cE$ moves full speed horizontally across the region. At some point, he must reverse direction, and the pursuer changes her offset (say, from positive to negative), and moves her rook frontier towards the evader, see Figure \ref{fig:rook-position}(b). The situation is more complicated in a general polygon, but these three basic characteristics remain the driving force of rook's strategy. 

\begin{figure}[t]

\begin{center}
\begin{tikzpicture}

\begin{scope}

\draw (2,0) -- (1,2.25) -- (-1,2) -- (-2,0) -- cycle;

\draw[gray, very thick] (-1.7, .6) -- (1.75, .6);

\draw[-latex] (.35, 1.6) -- (.8, 1.6);

\draw[dashed] (.35, 1.6) -- (.35, .6);
\draw[fill] (.35,  1.6) circle (1.5pt);
\node[above] at (.4, 1.6) {\scriptsize $\cE$};

\draw[dashed] (.85, 1.6) -- (.85, .6);

\draw[fill] (.85, 1.6) circle (1.5pt);
\node[above] at (.85, 1.6) {\scriptsize $\cE'$};

\draw[fill] (.2, .6) circle (1.5pt);

\node[below] at (.2, .6) {\scriptsize $\cP$};

\draw[-latex] (.2, .6) -- (.65, .6);

\draw[fill] (.7, .6) circle (1.5pt); 

\node[below] at (.7, .6) {\scriptsize $\cP'$};

\node at (0,-.5) {\scriptsize (a)};

\end{scope}

\begin{scope}[shift={(6,0)}]

\draw (2,0) -- (1,2.25) -- (-1,2) -- (-2,0) -- cycle;

\draw[gray, very thick] (-1.7, .6) -- (1.75, .6);

\draw[-latex] (.85, 1.6) -- (.43, 1.8);

\draw[dashed] (.4, 1.8) -- (.4, .6);
\draw[fill] (.4,  1.8) circle (1.5pt);
\node[above] at (.4, 1.8) {\scriptsize $\cE''$};

\draw[dashed] (.85, 1.6) -- (.85, .6);

\draw[fill] (.85, 1.6) circle (1.5pt);
\node[above] at (.85, 1.6) {\scriptsize $\cE'$};

\draw[fill] (.55, 1) circle (1.5pt); 

\node[above] at (.65, 1) {\scriptsize $\cP''$};

\draw[-latex]  (.7, .6) -- (.57, .95);

\draw[fill] (.7, .6) circle (1.5pt);

\node[below] at (.7, .6) {\scriptsize $\cP'$};

\node at (0,-.5) {\scriptsize (b)};
\end{scope}

\end{tikzpicture}
\end{center}

\caption{Rook's strategy in a convex polygon. (a) $\cP$ is in rook position with negative offset $x(\cP) - x(\cE) \leq 1/2$. $\cP$ can maintain this offset when $\cE$ moves to the right. (b) Eventually, $\cP$ will make progress, perhaps by switching to a positive offset when $\cE$ finally moves to the left.}

\label{fig:rook-position}

\end{figure}

Figure \ref{fig:rook-types} shows the  different  configurations for a pursuer in rook position. There are two main types: pocket position and chute position. Each of these has four subtypes, according to where the rook frontier intersects the boundary chains. The last boundary point visited by the pursuer is her checkpoint, see Definition \ref{def:checkpoint-mono} below. We use $c$ to denote the last checkpoint on the lower chain and $a$ to denote the last checkpoint on the upper chain.

\begin{definition}
Suppose that $\cP$ is in rook position.
The pursuer $\cP$ is in \style{upper (lower) pocket position}   when (1)  the evader is above (below) the pursuer,  and (2) both of the endpoints of the rook frontier are on the upper chain $\Pi_U$ (lower chain $\Pi_L$). 
The pursuer  $\cP$ is in \style{upper (lower) chute position}   when (1)  the evader is above (below) the pursuer,  and (2)   the right endpoint of the rook frontier is on the lower (upper) chain.  
\end{definition}

%

\begin{figure}[t]

\begin{center}

\begin{tikzpicture}[scale=.6]


\begin{scope}[shift={(8,-10)}]

\begin{scope}[shift={(115:4.5)}]
\path (3.15,0) coordinate (A1);
\path (-.85,0) coordinate (A2);
\end{scope}

\fill[gray!40] (115:4.5) -- (90:3) -- (82:3.5) -- (80:3) -- (60:4.5) -- (68:5.5) -- (70:4) -- (A1);

\draw[thick]  (129:4.75) -- (128:4.5) -- (115:6.5) --  (105:5) -- (100:6) -- (98:4.5) -- (85:6) -- (70:4) -- (68:5.5);
\draw[thick]  (135: 4.4) -- (130:3.4) -- (115:4.5) --  (90:3) -- (82:3.5) -- (80:3) -- (60:4.5);

\draw (A1) -- (A2);

\draw[fill] (115:4.5) circle (1.5 pt);
\draw[fill] (-.5,4.075) circle (1.5 pt);
\draw[fill] (-.25,3.7) circle (1.5 pt);

\node[above] at (115:4.5) {\small $c$};
\node[above right] at (-.5,4.075) {\small $\cP$};
\node[right] at (-.25,3.7) {\small $\cE$};

\node at (4,3.5) {\small lower chute 1};

\end{scope}


\begin{scope}[shift={(8,-15)}]

\begin{scope}[shift={(115:4.5)}]
\path (3.15,0) coordinate (A1);
\path (-.7,0) coordinate (A2);
\end{scope}

\fill[gray!40] (115:4.5) --  (110:2.5) -- (95:3.5) -- (90:3) -- (82:3.5) -- (80:3) -- (60:4.5) -- (68:5.5) -- (70:4) -- (A1);

\draw[thick]  (129:4.75) -- (128:4.5) -- (115:6.5) --  (105:4.23) -- (100:6) -- (98:4.5) -- (85:6) -- (70:4) -- (68:5.5);
\draw[thick]  (135: 4.4) -- (130:3.4) -- (115.5:5) --  (110:2.5) -- (95:3.5) -- (90:3) -- (82:3.5) -- (80:3) --  (60:4.5);

\draw (115:4.5) -- (A1);

\draw[fill] (115.5:5.025) circle (1.5 pt);
\draw[fill]  (105:4.23) circle (1.5 pt);. The
\draw[fill] (-.5,4.075) circle (1.5 pt);
\draw[fill] (-.25,3.7) circle (1.5 pt);

\node[above] at (115.5:5.025) {\small $c$};
\node[below] at (105:4.23) {\small $a$};
\node[above right] at (-.5,4.075) {\small $\cP$};
\node[right] at (-.25,3.7) {\small $\cE$};

\node at (4,3.5) {\small lower chute 2};

\end{scope}


\begin{scope}[shift={(8,0)}]

\begin{scope}[shift={(115:4.5)}]
\path (3.15,0) coordinate (A1);
\path (-.85,0) coordinate (A2);
\end{scope}

\fill[gray!40]   (115:4.5) --   (90:3) -- (82:3.5) -- (80:3) -- (A1);

\draw[thick]  (129:4.75) -- (128:4.5) -- (115:6.5) --  (105:5) -- (100:6) -- (98:4.5) -- (85:6) -- (80:5) -- (75:5.5);
\draw[thick]  (135: 4.4) -- (130:3.4) -- (115:4.5) --   (90:3) -- (82:3.5) -- (80:3) -- (71:4.8);

\draw (A1) -- (A2);

\draw[fill] (115:4.5) circle (1.5 pt);
\draw[fill] (.25,4.075) circle (1.5 pt);
\draw[fill] (0,3.5) circle (1.5 pt);

\node[above] at (115:4.5) {\small $c$};
\node[above right] at (.25,4.075) {\small $\cP$};
\node[left] at (0,3.6) {\small $\cE$};

\node at (3,3.5) {\small lower pocket 1};

\end{scope}


\begin{scope}[shift={(8,-5)}]

\begin{scope}[shift={(115:4.5)}]
\path (3.15,0) coordinate (A1);
\path (-.7,0) coordinate (A2);
\end{scope}

\fill[gray!40]  (115.5:4.5) -- (110:2.5) -- (95:3.5) -- (90:3) -- (82:3.5) -- (80:3) -- (A1);

\draw[thick]  (129:4.75) -- (128:4.5) -- (115:6.5) --  (105:4.23) -- (100:6) -- (98:4.5) -- (85:6) -- (80:5) -- (75:5.5);
\draw[thick]  (135: 4.4) -- (130:3.4) -- (115.5:5) --  (110:2.5) -- (95:3.5) -- (90:3) -- (82:3.5) -- (80:3) -- (71:4.8);

\draw (115:4.5) -- (A1);

\draw[fill] (115.5:5.025) circle (1.5 pt);
\draw[fill]  (105:4.23) circle (1.5 pt);
\draw[fill] (.1,4.075) circle (1.5 pt);
\draw[fill] (.3,3.7) circle (1.5 pt);

\node[above] at (115.5:5.025) {\small $c$};
\node[below] at (105:4.23) {\small $a$};
\node[above] at (.1,4.075) {\small $\cP$};
\node[right] at (.3,3.8) {\small $\cE$};

\node at (3,3.5) {\small lower pocket 2};

\end{scope}


\begin{scope}[shift={(0,-5)}]

\begin{scope}[shift={(115:4.5)}]
\path (3.15,0) coordinate (A1);
\path (-.85,0) coordinate (A2);
\end{scope}

\fill[gray!40] (115:6.5) --  (105:5) -- (100:6) -- (98:4.5) -- (85:6) -- (A1) -- (A2);

\draw[thick]  (129:4.75) -- (128:4.5)-- (115:6.5) --  (105:5) -- (100:6) -- (98:4.5) -- (85:6) -- (70:4) -- (68:5.5);
\draw[thick]  (135: 4.4) -- (130:3.4) -- (115:4.5) --  (90:3) -- (82:3.5) -- (80:3) -- (60:4.5);

\draw (A1) -- (A2);

\draw[fill] (115:4.5) circle (1.5 pt);
\draw[fill] (-.5,4.075) circle (1.5 pt);
\draw[fill] (-.25,4.5) circle (1.5 pt);

\node[above] at (115:4.5) {\small $c$};
\node[below right] at (-.6,4.075) {\small $\cP$};
\node[right] at (-.25,4.5) {\small $\cE$};

\node at (-4.7,4.5) {\small upper pocket 2};

\end{scope}


\begin{scope}[shift={(0,0)}]

\begin{scope}[shift={(115:4.5)}]
\path (3.15,0) coordinate (A1);
\path (-.7,0) coordinate (A2);
\end{scope}

\fill[gray!40]  (105:4.23) -- (100:6) -- (98:4.5) -- (85:6) -- (A1);

\draw[thick]  (129:4.75) -- (128:4.5) -- (115:6.5) --  (105:4.23) -- (100:6) -- (98:4.5) -- (85:6) -- (70:4) -- (68:5.5);
\draw[thick]  (135: 4.4) -- (130:3.4) -- (115.5:5) --  (110:2.5) -- (95:3.5) -- (90:3) -- (82:3.5) -- (80:3) --  (60:4.5);

\draw (115:4.5) -- (A1);

\draw[fill] (115.5:5.025) circle (1.5 pt);
\draw[fill]  (105:4.23) circle (1.5 pt);
\draw[fill] (0,4.075) circle (1.5 pt);
\draw[fill] (.2,5.15) circle (1.5 pt);

\node[above] at (115.5:5.025) {\small $c$};
\node[below] at (105:4.23) {\small $a$};
\node[below right] at (-.35,4.075) {\small $\cP$};
\node[below right] at (.2,5.25) {\small $\cE$};

\node at (-4.7,4.5) {\small upper pocket 1};

\end{scope}


\begin{scope}[shift={(0,-15)}]

\begin{scope}[shift={(115:4.5)}]
\path (3.15,0) coordinate (A1);
\path (-.85,0) coordinate (A2);
\end{scope}

\fill[gray!40]  (115:6.5) --  (105:5) -- (100:6) -- (98:4.5) -- (85:6) -- (80:5) -- (75:5.5) -- (71:4.8) -- (A1) -- (A2);

\draw[thick]  (129:4.75) -- (128:4.5)-- (115:6.5) --  (105:5) -- (100:6) -- (98:4.5) -- (85:6) -- (80:5) -- (75:5.5);
\draw[thick]  (135: 4.4) -- (130:3.4) -- (115:4.5) --   (90:3) -- (82:3.5) -- (80:3) -- (71:4.8);

\draw (A1) -- (A2);

\draw[fill] (115:4.5) circle (1.5 pt);
\draw[fill] (.25,4.075) circle (1.5 pt);
\draw[fill] (0,4.75) circle (1.5 pt);

\node[above] at (115:4.5) {\small $c$};
\node[below left] at (.5,4.075) {\small $\cP$};
\node[above right] at (0,4.75) {\small $\cE$};

\node at (-4.7,5) {\small upper chute 2};

\end{scope}


\begin{scope}[shift={(0,-10)}]

\begin{scope}[shift={(115:4.5)}]
\path (3.15,0) coordinate (A1);
\path (-.7,0) coordinate (A2);
\end{scope}

\fill[gray!40]  (105:4.23) -- (100:6) -- (98:4.5) -- (85:6) -- (80:5) -- (75:5.5) -- (71:4.8) -- (A1) -- (A2);

\draw[thick]  (129:4.75) -- (128:4.5) -- (115:6.5) --  (105:4.23) -- (100:6) -- (98:4.5) -- (85:6) -- (80:5) -- (75:5.5);
\draw[thick]  (135: 4.4) -- (130:3.4) -- (115.5:5) --  (110:2.5) -- (95:3.5) -- (90:3) -- (82:3.5) -- (80:3) -- (71:4.8);

\draw (115:4.5) -- (A1);

\draw[fill] (115.5:5.025) circle (1.5 pt);
\draw[fill]  (105:4.23) circle (1.5 pt);
\draw[fill] (.25,4.075) circle (1.5 pt);
\draw[fill] (0,4.75) circle (1.5 pt);

\node[above] at (115.5:5.025) {\small $c$};
\node[below] at (105:4.23) {\small $a$};
\node[below left] at (.5,4.075) {\small $\cP$};
\node[above right] at (0,4.75) {\small $\cE$};

\node at (-4.7,4.5) {\small upper chute 1};

\end{scope}

\end{tikzpicture}

\end{center}

\caption{During Rook Mode, the pursuer is  either in pocket position or in chute position, with four subtypes depending on where the rook frontier meets the boundary.  The evader territories are shown in gray. The vertices labeled $c$ and $a$ are checkpoints. Once $\cP$ has visited this point, the horizontal line segment to the left of the checkpoint is guarded, meaning that if $\cE$ crosses this segment then $\cP$ captures the evader on her next turn. }
\label{fig:rook-types}

\end{figure}

When the pursuer establishes pocket position, we have reached the endgame of the pursuit. The evader can try to use the features in the pocket to his advantage, but the pursuer always has a counter move. Eventually, the evader will be caught, just as in a convex polygon. On the other hand, when the pursuer is in chute position, the evader may be able to employ an escape move to prolong the game. 

\begin{definition}
\label{def:escape}
Suppose that $\cP$ and $\cE$ are in upper (lower) chute position. An \style{escape move} by the evader is a rightward move so that the lower (upper) boundary prevents the pursuer from re-establishing rook position. 
\end{definition}

Figure \ref{fig:escapemove} shows an example of an escape move in upper chute position. In response to an escape move, $\cP$ transitions back into Search Mode, using a new search path drawn from her current location. Pursuit continues as before, alternating between Search Mode and Rook Mode until the evader is trapped by a pocket position, which leads to capture.

This concludes the overview of our pursuer strategy. We work through the details in the subsections that follow.

\begin{figure}

\begin{center}
\begin{tikzpicture}[scale=2]

\draw[very thick] (.6, .6) -- (.9,.2) -- (1.3,1)  -- (2.3,.9) -- (2.8, .4);
\draw[very thick]  (.6,0)  -- (1.5,-.3) -- (1.9,0.5) -- (2.4,0) -- (2.8,0);

\fill[gray!20]  (.9,.2) -- (1.3,1) -- (2.3,.9) -- (2.8, .4) -- (2.8,0) -- (2.4,0) -- (1.9,0.5) -- (1.75,.2);

\draw (.6,.2) -- (1.75,.2);

\draw[fill] (1.75, .2) circle (.75pt);

\draw[fill] (1.6, .2) circle (1pt);

\draw[fill] (1.7, .65) circle (1pt);

\draw[fill] (2.3, .45) circle (1pt);

\draw[-latex] (1.7, .65) -- (2.25, .465);

\node[below left] at (1.6, .2) {$\cP$};
\node[below left] at (1.7, .65) {$\cE$};
\node[above right] at (2.3, .45) {$\cE'$};
\node[right] at (1.75, .2) {$b$};

\end{tikzpicture}

\caption{An escape move from upper chute position. When $\cE$ moves to $\cE'$, the lower boundary point $b$ prevents $\cP$ from re-establishing rook position.}

\label{fig:escapemove}

\end{center}
\end{figure}


\subsection{Search Strategy} 
\label{sec:search}

The key to the pursuer's Search Mode is the choice of \style{search path} $\sp$. Algorithm \ref{alg:monotone-path} explains how to  construct $\sp$ from any point $p \in Q$.

\begin{algorithm}
\caption{Create Monotone Search Path}
\label{alg:monotone-path}
\begin{algorithmic}[1]
\Require starting point  $p \in Q$
\While {$p \neq v_n$}
\State Move $p$ horizontally until reaching the boundary $\partial Q$
\If{$p$ is on lower boundary $\Pi_L$}
\State Traverse  upwards along $\Pi_L$ until reaching local maximum vertex $v_k \in \Pi_L$
\Else
\State Traverse downward along $\Pi_U$ until reaching  local minimum vertex $v_k \in \Pi_U$ 
\EndIf 
\EndWhile

\end{algorithmic}
\end{algorithm}

\begin{definition}
A \style{search path} in monotone polygon $Q$ is a path $\sp$ constructed by Algorithm \ref{alg:monotone-path}  Create Monotone Search Path.
\end{definition}

Intuitively, the pursuer follows a search path $\Pi$ that remains as horizontal as possible.
Our original search path is generated from the leftmost vertex, that is $p=v_1$.  A search path always terminates at the rightmost vertex $v_n$.  An example of the initial search path is shown in Figure \ref{fig:search-path} (a). 

To search for the evader, the pursuer travels along the search path according to Algorithm \ref{alg:monotone-search}.  While searching, we say that the pursuer is in \style{Search Mode}. Typically, $\cP$ traverses  the search path at unit speed. However, she stops at the $x$-coordinate of each vertex of $Q$ in order to (1) avoid stepping past an evader hiding behind a feature and (2)  allow for a change in direction when encountering a vertex of $\partial Q$ that is on the search path $\Pi$ (recall that  the pursuer can only move in a straight line on each turn). The pursuer also ensures that  $x(\cP)\leq x(\cE)$ at all times.


\begin{figure}
  \begin{center}
  \begin{tikzpicture}[scale=.95]

	\draw (0,1) -- (0.5,2) -- (1,0) -- (1.3,1.8) -- (1.7,1.3) -- (2.1,1.9) -- (2.6, .2) -- (3.2, 2) -- (3.5, 1) -- (3.75,1.8) -- (3.9,.5) --(4.3,-.2) 
	-- (5,1.5) -- (5.25,.5);
	\draw (0,1) -- (0.45,-.3) -- (.6,.2) -- (1.1,-1.5) -- (1.4, -.4) -- (1.7,-1) -- (1.9,0.5) -- (3,-1.75) -- (3.5,-.75) -- (4.2,-1.5) -- (5, -.5) -- (5.25,.5);
  	\draw[thick] (0,1) -- (.75,1) -- (1,0) -- (1.8333,0) -- (1.9,0.5) -- (2.51176,0.5) -- (2.6,0.2) -- (4.07143,0.2) -- (4.3,-.2) -- (5.075,-0.2)  --(5.25,.5);
	\node[left] at (0,1) {$v_1$};
	\node[right] at (5.25,.5) {$v_n$};

\draw[thick, fill] (0,1) circle (1pt);
\draw[thick, fill] (5.25,.5) circle (1pt);

\node at (4.1,1.8) {$\Pi_U$};
\node at (4.2,-1.75) {$\Pi_L$};
\node at (3.3, .5) {$\Pi$};

\draw[fill] (.6,1) circle (2pt);
\node at (-.1, 1.75) {$\cP_i$};
\draw[-latex] (0, 1.5) -- (.5, 1.1);

\draw[fill] (2.3,0.5) circle (2pt);
\node at (2.65, 1.5) {$\cP_j$};
\draw[-latex] (2.65, 1.4) -- (2.3, 0.6);

\draw[fill] (4.18,0.0) circle (2pt);
\node at (3.55, -.4) {$\cP_k$};
\draw[-latex] (3.7, -.25) -- (4.10, -0.05);

\draw[fill] (5.17,0.18) circle (2pt);
\node at (5.35, -1) {$\cP_{\ell}$};
\draw[-latex] (5.35, -.8) -- (5.23, 0.12);

\node at (3.25, -2.33) {\scriptsize (a)};

\begin{scope}[shift={(6.5,0)}]

	\draw (0,1) -- (0.5,2) -- (1,0) -- (1.3,1.8) -- (1.7,1.3) -- (2.1,1.9) -- (2.6, .2) -- (3.2, 2) -- (3.5, 1) -- (3.75,1.8) -- (3.9,.5) --(4.3,-.2) 
	-- (5,1.5) -- (5.25,.5);
	\draw(0,1) -- (0.45,-.3) -- (.6,.2) -- (1.1,-1.5) -- (1.4, -.4) -- (1.7,-1) -- (1.9,0.5) -- (3,-1.75) -- (3.5,-.75) -- (4.2,-1.5) -- (5, -.5) -- (5.25,.5);


\draw[fill] (.6,1) circle (2pt);
\draw[thick] (0,1) -- (.75,1);
\node[left] at (0,1) {$c_i$};
\node[below] at (.4,1) {$\frontier_i$};
\node at (-.1, 1.75) {$\cP_i$};
\draw[-latex] (0, 1.5) -- (.5, 1.1);
\draw[thick, fill] (0,1) circle (1.5pt);


\draw[fill] (2.3,0.5) circle (2pt);
\draw[thick] (1.1,0.5) -- (2.5,0.5);
\node[above] at (1.9,0.5) {$c_j$};
\node at (1.4,0.75) {$\frontier_j$};
\node at (2.65, 1.5) {$\cP_j$};
\draw[-latex] (2.65, 1.4) -- (2.3, 0.6);
\draw[thick,fill] (1.9,0.5) circle (1.5pt);


\draw[very thick] (2.15,0.0) -- (4.18,0.0) ;
\draw[fill] (4.18,0.0) circle (2pt);
\node at (4.0, -.4) {$\cP_k = c_k$};
\node at (2.8, -0.35) {$\frontier_k$};
\draw[-latex] (3.7, -.25) -- (4.10, -0.05);
\draw[very thick,fill] (4.18,0.0) circle (2pt);


\node[above] at (4.9,0.18) {$\frontier_{\ell}$};
\draw[very thick] (4.45,0.18) -- (5.17,0.18) ;

\draw[fill] (5.17,0.18) circle (2pt);
\node at (5.3, -1) {$c_{\ell} = \cP_{\ell} $};
\draw[-latex] (5.55, -.8) -- (5.23, 0.12);


\node at (3.25, -2.33) {\scriptsize (b)};

\end{scope}

  \end{tikzpicture}

  \end{center}
  \caption{(a) Search path through a monotone polygon with four pursuer positions shown, $i < j < k < \ell$. (b) The corresponding search frontier $\frontier_t$ and  checkpoint $c_t$ for those four pursuer positions. 
 }
  \label{fig:search-path}
\end{figure}

\begin{definition}
\label{def:checkpoint-mono}
Suppose that the pursuer is in Search Mode at time $t$.  The pursuer's
\style{checkpoint} $c_t$ is the rightmost point in 
$\frontier_t^- \cap \partial Q$.
\end{definition}

When the pursuer is in Search Mode in a monotone polygon, her checkpoint is always the last boundary point that she visited. We note here that scallop polygons,  require a more complicated notion of checkpoints that distinguishes between the upper boundary and  the lower boundary; see Definition \ref{def:scallop-checkpoint} in Section \ref{sec:scallop-search-overview}.

\begin{algorithm} 
\caption{Monotone Search Strategy}
\label{alg:monotone-search}
\begin{algorithmic}[1]
\Require Vertices of $Q$ are $ v_1, v_2, \ldots , v_n$ where $x(v_{i}) \leq x(v_{i+1})$ for $1 \leq i < n-1$. 
\Require $\Pi$ is the search path (from Create Monotone Search Path)
\Require ${ x(\cP_{t-1}) \leq x(\cE_{t-1})}$, but $\cE_{t-1}$ might be invisible to $\cP_{t-1}$

\While { {\sc not} $\big(\cE_{t-1}$ is visible and  $0 \leq x(\cE_{t-1})-x(\cP_{t-1}) \leq  1/2\big)$ }
\State Evader moves from $\cE_{t-1}$ to $\cE_{t}$
\State Set checkpoint $c$ to be the rightmost boundary point  in $\partial Q \cap \horizlineneg{\cP_{t-1}}$  
\If{$\cE_t$ is visible and $d(\cP_{t-1}, \cE_t) \leq 1$}
\State $\cP_t \leftarrow \cE_t$, which captures the evader
\ElsIf{$ -1/2 \leq x(\cE_t) - x(\cP_{t-1}) < 0}$ 
\State Note: $\cE_t$ is in evader territory and visible by Lemma \ref{lemma:hide-left}
\State $\cP'_t  \leftarrow (x(\cE_t), y(\cP_{t-1}))$
\Else 
\State Note: $x(\cP_{t-1}) \leq x(\cE_t)$
\State $k \leftarrow$  the unique index for which $x(v_k) \leq x(\cP_{t-1}) < x(v_{k+1})$ 
\State $Z \leftarrow $  the rightmost search path point in $\{ W \in \Pi : d(\cP_{t-1}, W)  \leq 1\} $
\State $\cP_{t} \leftarrow $   the point on $\Pi$ with $x$-coordinate $\min \{ x(v_{k+1}), x(Z), x(\cE_t) -1/2 \}$
\EndIf
\State Update $t \leftarrow t+1$
\EndWhile

\end{algorithmic}
\end{algorithm}


\begin{definition}
\label{def:search-terms}
Suppose that the pursuer is in Search Mode at time $t$ with checkpoint $c_t$. The pursuer's \style{search frontier}  
$\frontier_t = \horizline{\cP_t}$ 
is the horizontal line segment through $\cP_t$.
The \style{guarded frontier}
$\frontier_t' = X^-(c_t)$
is the left part of $\frontier_t$, up to and including the checkpoint $c_t$.
The guarded frontier $\frontier_t'$ partitions $Q$ into two sub-polygons:   the \style{pursuer territory} $Q_{\cP}(t)$ to the left,  and   the \style{evader territory} $Q_{\cE}(t)$ to the right. 
If a point $Z \in Q_{\cP}(t)$ then we say that $Z$ is \style{guarded}, and that $\cP$ \style{guards} $Z$.
\end{definition}

Figure \ref{fig:search-path}(b) shows the search frontier and the checkpoint for four different pursuer positions.  Figure~\ref{fig:territory} gives two examples of guarded frontiers and their corresponding pursuer territory and evader territory.
We  prove below that $\cE$ cannot cross the guarded frontier  $\frontier_t'$ without being captured.
Observe that the checkpoint $c_t$ is the unique rightmost vertex in the pursuer territory $Q_{\cP}(t)$. Similarly, the left endpoint of the guarded frontier $\frontier_t'$ is  the unique leftmost point in the evader territory $Q_{\cE}(t)$. Both of these facts play a role in the pursuer's ability to guard $Q_{\cP}(t)$.

As the pursuer searches the polygon, the checkpoint represents our progress.
Notice that $c_t=v_1$ until $\cP$ first reaches $\partial Q$. Up until that event, the pursuer territory is simply the point $v_1$.  When the checkpoint location is updated, its $x$-coordinate increases, causing the pursuer territory to increase and the evader territory to decrease. In particular, if $t < t'$ then $Q_{\cP}(t) \subseteq Q_{\cP}(t')$ and
$Q_{\cE}(t) \supseteq Q_{\cE}(t)$ with equalities holding if and only if we have not updated our checkpoint. We  discuss this notion of progress more rigorously in  Section~\ref{sec:time}.

\begin{definition}
The searching pursuer $\cP$ \style{makes progress} in the monotone polygon every time that she updates her checkpoint from $c_t$ to $c_t'$ where $t < t'$ and $x(c_t) < x(c_{t'})$.
\end{definition}

\begin{figure}
  \begin{center}
  \begin{tikzpicture}[scale=.9]

\fill[gray!25] (2.6, .2) -- (3.2, 2) -- (3.5, 1) -- (3.75,1.8) -- (3.9,.5) -- (4.07143,0.2)  -- cycle;

\fill[gray!25] (2.05,0.2) -- (4.07143,0.2) --(4.3,-.2) -- (5, 2.1) -- (5.5,.5) -- (4.2,-1.5) -- (3.5,-.75)  -- (3,-1.75) -- cycle;  
  
	\draw (1,0) -- (1.3,1.8) -- (1.7,1.3) -- (2.1,1.9) -- (2.6, .2) -- (3.2, 2) -- (3.5, 1) -- (3.75,1.8) -- (3.9,.5) --(4.3,-.2) -- (5, 2.1) -- (5.5, .5);
	\draw (1,0) -- (1.7,-1) -- (1.9,0.5) -- (3,-1.75) -- (3.5,-.75) -- (4.2,-1.5) -- (5.5, .5);
  	\draw[dashed] (1,0) -- (1.8333,0) -- (1.9,0.5) -- (2.51176,0.5) -- (2.6,0.2) -- (4.07143,0.2) -- (4.3,-.2) -- (5,-0.2);

\draw[thick] (2.05,0.2) -- (2.6, 0.2); 

\draw[thick, fill] (2.6, 0.2) circle (2pt);

\node[above right] at (2.6, 0.2) {$c$};

\draw[fill] (3.5,0.2) circle (2pt);
\node at (3.3, -.15) {$\cP$};

\node at (2.6, -.05) {$X'$};

\node at(1.7, .9) {$Q_{\cP}$};
\node at(4.25, -.8) {$Q_{\cE}$};

\begin{scope}[shift={(5.75,0)}]

\draw[gray!50, fill=gray!25] (2.1,0) -- (4.1,0) --(4.3,-.2) -- (5, 2.1) -- (5.5,.5) -- (4.2,-1.5) -- (3.5,-.75)  -- (3,-1.75) -- cycle;  
  
	\draw (1,0) -- (1.3,1.8) -- (1.7,1.3) -- (2.1,1.9) -- (2.6, .2) -- (3.2, 2) -- (3.5, 1) -- (3.75,1.8) -- (3.9,.5) --(4.3,-.2) -- (5, 2.1) -- (5.5, .5);
	\draw (1,0) -- (1.7,-1) -- (1.9,0.5) -- (3,-1.75) -- (3.5,-.75) -- (4.2,-1.5) -- (5.5, .5);
  	\draw[dashed] (1,0) -- (1.8333,0) -- (1.9,0.5) -- (2.51176,0.5) -- (2.6,0.2) -- (4.07143,0.2) -- (4.3,-.2) -- (5,-0.2);

\draw[thick] (2.15,0) -- (4.2,0);
\draw[fill] (4.17,.0) circle (2pt);
\node at (4.25,-.3) {$\cP = c$};

\node at (3.2,-.25) {$X'$};

\node at(1.7, .9) {$Q_{\cP}$};
\node at(4.25, -.8) {$Q_{\cE}$};

\end{scope}

  \end{tikzpicture}

  \end{center}
  \caption{ Two examples of a guarded frontier $X'$. The dashed lines indicate the search path. The evader territory $Q_{\cE}$ is shaded and pursuer territory $Q_{\cP}$ is unshaded. The checkpoint $c$ is the rightmost point of $X'$, as well as the rightmost point in $Q_{\cP}$.}
  \label{fig:territory}
\end{figure}

Next, we show that if  $\cE$ moves horizontally  past $\cP$, then he must  be visible.

\begin {lemma} 
\label{lemma:hide-left}
Suppose that $\cP$ is following the search strategy and  $x(\cP_{t-1}) \leq x(\cE_{t-1})$.  If  $x(\cE_{t} ) < x(\cP_{t-1})$ then $\cP_{t-1}$ can see $\cE_{t}$. Furthermore, if $\cE$ also crossed the horizontal line $\horizline{\cP_{t-1}}$, then he is captured immediately.  
\end{lemma}

\begin{proof} 
First, suppose that both $\cE_{t-1}, \cE_{t}$ are invisible to $\cP_{t-1}$.
We have $x(\cE_{t}) < x(\cP_{t-1}) < x(\cE_{t-1})$, so each evader position is obscured by a distinct vertex. However, these two vertices would be separated by a distance strictly less than $d(\cE_{t-1}, \cE_{t}) \leq 1$, which is impossible due to the minimum feature size of $Q$, which requires that   no two vertices are within one unit of each other.
 Therefore, at least one of $\cE_{t-1}, \cE_{t}$ is visible to $\cP_{t-1}$. Without loss of generality, assume that $y(\cP_{t-1})\leq y(\cE_{t-1})$.  There are two cases, depending on the location of $\cE_{t}$.

Suppose that $y(\cP_{t-1}) > y(\cE_{t})$, so that $\cE$ has moved from the first quadrant of $\cP$ to her third quadrant, see Figure \ref{fig:hide-left}(a). Consider the triangle $\cE_{t-1} \cP_{t-1} \cE_{t}$, which might intersect the boundary of $Q$. The angle $\angle \cE_{t-1} \cP_{t-1} \cE_{t}$ is obtuse, which means that
$\overline{\cE_{t-1} \cE_{t}}$ is the unique longest side. We know that at least one of $\cE_{t-1}, \cE_{t}$ is visible to $\cP_{t-1}$. If $\cP_{t-1}$ sees $\cE_{t-1}$ then
$d(\cP_{t-1},\cE_{t-1}) < 1$,  so the game was over at time $t-1$. Therefore $\cP_{t-1}$ must see $\cE_{t}$ and  $d(\cP_{t-1},\cE_{t}) < 1$, so the game is over after this evader move.

Next, suppose that $y(\cP_{t-1}) \leq y(\cE_{t})$, so that $\cE$ has moved from the first quadrant of $\cP$ to her second quadrant, see Figure \ref{fig:hide-left}(b).
Assume that $\cP_{t-1}$ cannot see $\cE_{t}$. Since we are following the search strategy, no lower feature to the left of the pursuer can impede her horizontal movement. Therefore the pursuer's visibility must be blocked by the upper chain $\Pi_U$. In other words,   $\overline{\cP_{t-1}\cE_{t}}$ intersects $\Pi_U$ at some point $Z$; see Figure~\ref{fig:territory}(a). Let $Z'$ be the point on $\overline{\cE_{t-1} \cE_{t}}$ that is directly above $Z$. Since $Z$ is on the boundary of $Q$, the point $Z'$ lies outside $Q$. Therefore $\overline{\cE_{t-1}\cE_{t}}$ also intersects $\Pi_U$, which means that $\cE_{t}$ was not visible from $\cE_{t-1}$, and so the evader could not have moved to $\cE_{t}$. We conclude that $\cP_{t-1}$ sees $\cE_{t}$.
\end{proof}

\begin{figure}
  \begin{center}
  \begin{tikzpicture}[scale=.9]

\begin{scope}[shift={(-5,2)}]

\draw (0, -1.5) -- (0,1.5);
\draw (-1.5, 0) -- (1.5,0);

\draw[fill] (0,0) circle (2pt);

\draw[fill] (.25,1) circle (2pt);

\draw[fill] (-1,-.25) circle (2pt);

\draw[-latex] (.25,1) -- (-.95,-.225);

\draw[dashed] (.25,1) -- (0,0) -- (-1, -.25);

\node[below right] at (0,0) {\small $\cP_{t-1}$};

\node[right] at (.25,1) {\small $\cE_{t-1}$};

\node[below] at (-1,-.25) {\small $\cE_{t}$};

\node at (-.75,.6) {\small $\leq 1$};

\end{scope}  

\node at (-5, 0) {\scriptsize (a)};

\begin{scope}[shift={(0,0)}]

\fill[gray!50] (0,1) -- (.25, .5) -- (1.5, .5) -- (2,1.5) -- (.27,1.5);

\draw (0,1) -- (1, 3) -- (2,4) -- (2.35, 2) -- (3,4) -- (4, 3) -- (5, 2);

\draw (0,1) -- (.25, .5) -- (1.5, .5) -- (2,1.5) -- (2.5, .5) -- (4,1) -- (5, 2);

\draw[thick] (2,1.5) -- (.27,1.5);
\draw[dashed] (2,1.5) -- (4.5,1.5);

\draw[-latex] (3.2, 3.2) -- (1.95,3);

\draw[fill]  (3.2, 3.2) circle (2pt);

\node[below right] at (3.2, 3.2) {\small $\cE_{t-1}$};

\draw[fill]  (1.9, 3) circle (2pt);

\node[left] at (1.9, 3) {\small $\cE_{t}$};

\draw[fill]  (3, 1.5) circle (2pt);

\node[below right] at (3, 1.5) {\small $\cP_{t-1}$};

\draw[thick, dashed] (3,1.5) -- (1.9,3);

\draw[thick, dashed] (2.45,2.25) -- (2.45, 3.075);

\draw[fill]  (2.45,2.25) circle (2pt);
\node[right] at (2.45,2.25) {\small $Z$};

\draw[fill]  (2.45, 3.075) circle (2pt);

\node[above] at (2.45, 3.075) {\small $Z'$};

\node at (.9,1) {\scriptsize $Q_{\cP}(t-1)$};

\node at (2.5, 0) {\scriptsize (b)};

\end{scope}

  \end{tikzpicture}

  \end{center}
  \caption{The evader cannot hide when moving past the pursuer to the left. (a) When $\cE$ moves from the first quadrant to the third quadrant, at least one of $\cE_{t-1},\cE_t$ is visible to $\cP_t$ and within unit distance. (b) When $\cE$ starts above $\cP$, he cannot hide behind an upper feature.}
  \label{fig:hide-left}
\end{figure}

We can now prove that the pursuer achieves rook position by following the search strategy.

\begin{lemma} \label{lemma:search}
If  $\cP$ follows the search strategy then $\cE$ cannot step into the pursuer territory without being caught. Furthermore, $\cP$  either achieves rook position or captures $\cE$ within $O(n(Q) + \diam(Q))$ rounds. 
\end{lemma}

\begin{proof}
We claim that two loop invariants are maintained by the pursuer during Search Mode: $\cE_t \notin Q_{\cP}(t)$ and $x(\cP_t) \leq x(\cE_t)$. 
Initially $\cP$ is located at the leftmost vertex $\cP_0 = v_1$,  so $x(\cP_0) \leq x(\cE_0)$ and
the pursuer territory $Q_{\cP}(0) = \{ v_1 \}$. 
Now suppose that  $\cP$ has not established rook position by time $t-1$ with $x(\cP_{t-1}) < x(\cE_{t-1})$ and $\cE_{t-1} \notin Q_{\cP} (t-1)$. Since $\cP$ is not in rook position, either $x(\cE_{t-1}) > x(\cP_{t-1})  + 1/2$, or the evader is hidden and $x(\cP_{t-1}) \leq x(\cE_{t-1}) \leq x(\cP_{t-1}) + 1/2$. At this point, the evader moves from $\cE_{t-1}$ to $\cE_t$. There are four cases, three of which result in rook position (or capture). 

{\bf Case 1:} Either $x(\cP_{t-1}) + 1/2 < x(\cE_t)$, or $x(\cP_{t-1}) \leq  x(\cE_t)$ and $\cP_{t-1}$ does not see $\cE_{t-1}$. Note that $\cE_t \notin Q_{\cP}(t-1)$ since $x(c_{t-1}) \leq x(\cP_{t-1}) < x(\cE_t)$ and $c_{t-1}$ is the rightmost point in $Q_{\cP}(t-1)$. In response, the pursuer moves forward along $\Pi$. If she encounters the $x$-coordinate of a vertex, encounters the boundary,  or achieves $x(\cE_t) - x(\cP_t) \leq 1/2$ with visibility, then she stops; otherwise she moves unit distance. Either way, she maintains $x(c_t) \leq x(\cP_{t}) < x(\cE_t)$ and $\cE_t \notin Q_{\cP}(t)$. After this move, if $x(\cE_t) - x(\cP_t) \leq   1/2$ with visibility, then the pursuer has achieved rook position; otherwise we continue in Search Mode. The number of such moves is $O(\diam(Q) + n)$. 

{\bf Case 2:}  $\cP_{t-1}$ sees $\cE_t$ and 
$x(\cP_{t-1}) \leq x(\cE_t) \leq x(\cP_{t-1}) + 1/2$. This means that $\cP$ is in already in rook position, so we take $\cP_t = \cP_{t-1}$ and $c_t =c_{t-1}$.  Since $x(c_t) \leq x(\cP_{t}) \leq x(\cE_t)$, we have $\cE_t \notin Q_{\cP}(t)$ (unless $c_t = \cP_t = \cE_t$, which ends the game).

{\bf Case 3:}  $x(\cE_t) < x(\cP_{t-1})$ and $\cE_t \in Q_{\cE}(t-1)$. The  evader position $\cE_{t}$ is visible to $\cP_{t-1}$ by Lemma \ref{lemma:hide-left}. Since the leftmost point of $Q_{\cE}(t-1)$ is the left endpoint of $\frontier_t'$,  the pursuer can move horizontally to the point $(x(\cE_t), y(\cP_{t-1}))$ to achieve rook position because $x(\cP_t) - x(\cE_t) \leq x(\cE_{t-1}) - x(\cE_t) \leq 1$.

{\bf Case 4:}   $x(\cE_t) < x(\cP_{t-1})$ and $\cE_t \in Q_{\cP}(t-1)$, so that the evader crossed the guarded frontier $X'_{t-1}$. 
This means that $\cE_t$ is not visible to $\cP_{t-1}$, blocked by the feature containing the current checkpoint. As observed in the proof of Lemma \ref{lemma:hide-left}, this means that $d(\cP_{t-1}, \cE_{t-1}) \leq 1$ and $\cE_{t-1}$ was visible to $\cP_{t-1}$. Therefore the game was over at time $t-1$, and this case cannot occur.

Having established these two loop invariants, we now show that $\cP$ achieves rook position.
When $\cP$ encounters a boundary point at time $t$, she makes progress by updating her checkpoint, which also increases the pursuer territory $Q_{\cP}(t)$. Our loop invariants show that $\cE \notin Q_{\cP}(t)$ and that $\cE$ is to the right of $\cP$. 

Assume for the sake of contradiction that the searching $\cP$ reaches the rightmost vertex $v_n$ without achieving rook position.
Let $t$ be the time that $\cP$ reaches $v_n$, so that
$x(v_n) - x(P_{t-1}) \leq 1$.
The loop invariants ensure that $x(\cP_{t-1}) < x(\cE_{t-1}) \leq x(v_n)$. We claim that $\overline{\cP_{t-1} \cE_{t-1}} \in Q$. Indeed,  the searching pursuer stops at the $x$-coordinate of every boundary vertex, so $x(v_{n-1}) \leq x(\cP_{t-1}) < x(\cE_{t-1})$, which means that there are no features between $\cP_{t-1}$ and $\cE_{t-1}$.  
The pursuer  $\cP_{t-1}$ was not in rook position, so $1/2 < x(E_{t-1}) - x(P_{t-1}) \leq
x(v_n) - x(P_{t-1}) \leq 1$.

Now consider $\cE_t$. We claim that  $1/2 < x(\cE_{t}) - x(\cP_{t-1}) \leq 1$. Otherwise,  
$|x(E_{t}) - x(P_{t-1})| \leq 1/2$, and the pursuer could have remained stationary $\cP_t = \cP_{t-1}$ to achieve rook position, contradicting the pursuer's move to $v_n$. The evader is still visible, so $\cP$ can increase her $x$-coordinate to
$x(\cE_{t}) - 1/2$ and achieve rook position, rather than stepping to $v_n$, a contradiction once again.
Therefore, $\cP$ achieves rook position before reaching $v_n$.

\end{proof}

%

\subsection{Rook's Strategy} \label{rook}

After searching, $\cP$ uses the  \style{monotone rook's strategy} of Algorithm \ref{alg:monotone-rook} to either capture $\cE$ or to reduce the evader territory. A pursuer executing this algorithm is in \style{Rook Mode}.  We describe two key features  of Rook Mode before delving into the proofs.

Once $\cP$ is in rook position, the  entire frontier $\frontier_t$ is actually guarded. In other words,  the evader cannot cross $\frontier_t$ without being caught by the pursuer. Therefore, when the pursuer enters Rook Mode, we immediately update the guarded frontier to be $\frontier_t' :=\frontier_t$ and update the checkpoint $c_t$ to be the righthand endpoint of this segment.  

\begin{definition}
Suppose that the pursuer is in rook position at time $t$. The horizontal segment  $\frontier_t = \horizline{\cP_t}$ is the \style{rook frontier}.
\end{definition}

The pursuer  $\cP$  can methodically advance the rook frontier via rook's strategy.  Any leftward movement of $\cE$ can be copied by $\cP$ since the leftmost point in $Q_{\cE}(t)$ is the endpoint of $\frontier_t$. The pursuer can also mimic the evader's rightward movement, up to the checkpoint $c_t$. If the pursuer is impeded by the boundary to the right, she will start a new Search Mode. Throughout these horizontal moves, the pursuer territory is guarded from incursion.

We now give a proper description of Rook Mode. 
 Recall that there are two types of rook position: pocket position and chute position, as shown in Figure  \ref{fig:rook-types}.
 We will see that pocket position is the endgame of the pursuit: $\cP$ can maintain rook position and methodically advance her rook frontier until capture.
In an upper pocket position, the lower chain $\Pi_L$ lies beneath the rook frontier $\frontier(\cP)$, so  features on $\Pi_L$ cannot impede the pursuer's movements.  Meanwhile we will see that any attempt by $\cE$ to use upper features to his advantage can be immediately neutralized. The analogous statements hold for a lower pocket position.
On the other hand, if $\cP$ is in chute position, then $\cE$ can break out of this rook configuration via an escape move (see Definition \ref{def:escape}) in which  the evader moves rightward so that the  boundary blocks the rook response. This forces the pursuer to abandon Rook Mode and initiate a new Search Mode.  In this transition to a new Search Mode, we will have increased the pursuer territory and updated the pursuer's checkpoint, so we have made progress towards eventual capture.

\begin{algorithm}
\caption{ Monotone Rook's Strategy }
\label{alg:monotone-rook}
\begin{algorithmic}[1]
\Require  $0 \leq x(\cE_{t-1}) - x(\cP_{t-1}) \leq  1/2$ and $\cP_{t-1}$ sees $\cE_{t-1}$
\Require $Q_{\cE}(t-1)$ is bounded above by $\Pi_U$  and below by the rook frontier $\frontier(\cP_{t-1})$
\While {$\cE$ is not captured}
\State Evader moves from $\cE_{t-1}$ to $\cE_t$
\If {$\cE$ made an escape move}
\State Note: $\cP$ must have been in chute position
\State Exit  (in order to start a new Monotone Search Strategy)
\Else
\If{$\cE$ crossed below $\frontier(P_{t-1})$}
\State $\cP_t \leftarrow$  the point where the evader crossed $\frontier(\cP_{t-1})$, capturing $\cE$
\ElsIf{$\cE_t$ is within $\sqrt{3}/2$ of $\frontier(\cP_{t-1})$}
\State Note: $\cE_t$ must be visible to $\cP_{t-1}$
\State $\cP_{t} \leftarrow$  the point on $\overline{ \cP_{t-1} \cE_t}$ as close to $\cE_t$ as possible, capturing $\cE$
\ElsIf{$\cE_t$ is not in sight}
\State $\cP_t \leftarrow$   the highest reachable point below $\cE_{t-1}$
\State Note: this advances the rook frontier by at least $\sqrt{3}/2$.
\ElsIf {$|x(\cP_{t-1})-x(\cE_t)|\le 1$}
\State $\alpha \leftarrow$  the closer of $x(\cE_t) \pm 1/2$ to $x(\cP_{t-1})$
\State $\cP_t \leftarrow$  the highest reachable point on vertical line $x = \alpha$
\State Note: this advances the rook frontier by at least $7/22$.
\Else 
\State $\cP_t \leftarrow$ the highest reachable point on vertical line $x = x(\cE_t) - 1/2$
\EndIf
\EndIf
\State Update $t \leftarrow t+1$
\EndWhile
\end{algorithmic}
\end{algorithm}

We discuss the behavior of Algorithm \ref{alg:monotone-rook} in pocket position and in chute position. We will then be able to prove Theorem \ref{thm:monotone}.

\subsubsection{Pocket Position}
\label{sec:pocket}

We prove that once the pursuer has achieved pocket position, she can capture the evader using 
the Monotone Rook's Strategy.

\begin{lemma}
\label{lemma:endgame}
If   $\cP$ is in pocket position then  the monotone rook's strategy captures the evader in $O(\Area(Q_{\cE}))$ turns.
\end{lemma}

The proof requires some helper lemmas,  paying careful attention to  visibility and mobility constraints. 
We begin by showing that the evader cannot approach $\frontier(\cP)$ or leave $Q_{\cE}$ without being captured.

\begin{lemma}
\label{lemma:guard}
Suppose  $\cP$ and $\cE$ are in pocket position at time $t-1$. 
If the evader moves so that $| y(\cP_{t-1}) - y(\cE_{t})| \le \sqrt{3}/2$, then the pursuer can move to a point $\cP_{t}$ that captures $\cE$. In particular, $\cE$ cannot cross the rook frontier $\frontier(\cP_{t-1})$ without being captured.
\end{lemma}

\begin{proof}
Without loss of generality, $\cP$ is in upper pocket position and $\cE$ is in the first quadrant of $\cP$.
By the definition of rook position, $\cP$ sees $\cE$ and  $\dx(t-1) \leq 1/2$. Therefore  $\dy(t-1) > \sqrt{3}/2$,  otherwise the evader was captured at time $t-1$. 
Suppose that the evader does not cross the rook frontier and that $| y(\cP_{t-1}) - y(\cE_{t})| \le \sqrt{3}/2$.
We know that  $y(\cE_{t}) < y(\cE_{t-1})$ and  that $\cP_{t-1}$ can see $\cE_{t-1}$. Since we are in pocket position, $\cP_{t-1}$ can also see $\cE_{t}$. Indeed, an upper feature blocking $\overline{\cP_{t-1} \cE_{t}}$ would have also blocked the move $\overline{\cE_{t-1}\cE_{t}}$.
 We have $| x(\cP_{t-1}) - x(\cE_{t}) | \leq \dx(t-1)+1 = 3/2$ and therefore
$$
\dist{\cP_{t-1}}{\cE_{t}} \le \sqrt{(3/2)^2+(\sqrt{3}/2)^2} = \sqrt{3} < 2.
$$ 
The pursuer can move directly towards $E_{t}$ to achieve $\dist{\cE_{t}}{\cP_{t}} <1$, which ends the game.  

Finally, suppose that  the evader moves to a point $\cE_{t}$ below the rook frontier $\frontier(\cP_{t-1})$. We must be careful in this case since features on the lower boundary might prevent $\cP$ from moving directly towards $\cE_{t}$. Instead, the pursuer moves laterally along the rook frontier to capture the evader. Let $Z$ be the point of intersection of $\frontier(\cP_{t-1})$ and  $\overline{\cE_{t-1} \cE_{t}}$. Since $\dy(t-1) > \sqrt{3}/2$, we have 
$$ \dist{\cP_{t-1}}{Z}  \leq | x(\cP_{t-1}) - x(\cE_{t-1})|  + | x(\cE_{t-1}) - x(Z)| \leq 1/2 +| x(\cE_{t-1}) - x(Z)| < 1.$$ 
Therefore $\cP$ can move to $Z$, and clearly $\dist{Z}{ \cE_{t}} < \dist{\cE_{t-1}}{\cE_{t}} \leq 1$, so the evader is captured.
\end{proof}

In addition to closing the distance to the evader, our rook's strategy maintains (or re-establishes) visibility  with the evader at each pursuer turn. 

\begin{lemma}
\label{lemma:pocket-viz}
Suppose that $\cP_{t-1}$ is in pocket position and that  either $x(\cE_{t} ) \leq x(\cP_{t-1}) \leq x(\cE_{t-1})$, or  $x(\cE_{t-1} ) \leq x(\cP_{t-1}) \leq x(\cE_{t})$.  Then $\cP_{t-1}$ can see $\cE_{t}$. Furthermore, if $\cE$ also crossed the horizontal line $\horizline{\cP_{t-1}}$, then he is captured immediately.  
\end{lemma}

\begin{proof}
Without loss of generality, we are in upper pocket position and $\cE$ is in the first quadrant of $\cP$. The proof is identical to the proof of Lemma \ref{lemma:hide-left}, with one small change:  the endpoints of the rook frontier are on the upper boundary, which means that the lower boundary does not impede the pursuer's movement along the rook frontier. 
\end{proof}

The next lemma handles the  case in which  $\cE$ moves to a position that is invisible to $\cP$.

\begin{lemma}
\label{lemma:re-establish}
Suppose  $\cP$ and $\cE$ are in pocket position at time $t-1$. 
If the evader moves to  $\cE_{t}$ that is invisible to  $\cP_{t-1}$ then the pursuer can move to  $\cP_{t}$ to re-establish rook position.  This move advances the rook frontier by at least $\sqrt{3}/2$.

\end{lemma}

\begin{figure}
\begin{center}

\begin{tikzpicture}[scale=1]


\draw (-1,0) -- (2,0);

\node[right] at (2,0) {\small $\frontier(\cP_{t-1})$};

\draw[fill] (0,0) circle (1.5pt);

\node[above left] at (0,0) {\small $\cP_{t-1}$};

\draw[fill] (1,1.735) circle (1.5pt);

\node[below right] at (1,1.735) {\small $\cP_{t}$};

\draw[fill] (1,2.1) circle (1.5pt);

\node[right] at (1.1,2.1) {\small $\cE_{t-1}$};

\draw[fill] (1.42,3.5) circle (1.5pt);

\node[right] at (1.42, 3.5) {\small $\cE_{t}$};

\draw (0,0) -- (1,2.1);

\draw (0,0) -- (1.42,3.5);

\draw (0,0)  -- (1, 1.735) -- (1.42, 3.5);

\draw (1,0) -- (1, 2.1) -- (1.42, 3.5);

\node[below] at (.75, 0) {\small $\dx(t-1)$};
\node[right] at (1, .8) {\small $\sqrt{ 1- \dx(t-1)^2}$};

\draw[thick][pattern=north west lines, pattern color=gray] (-.5, 2.5) -- (.65,1.55) -- (.85, 2.7) -- (1.1, 2.9) --(1.3,3.75);

\node at (.5,-.85) {\scriptsize (a)};

\end{tikzpicture}
\qquad \qquad
\begin{tikzpicture}[scale=.6]

\draw[thick] (-6,0) -- (1,0);

\draw[thick] (-5,0) -- (-5,5);

\draw[thick] (0,0) -- (-5,5)--(-2.5,0);

\draw[thick] (0,0) -- (-3.5,2)--(-3.5,0);

\draw[thin] (0.5,5) -- (-5,5);

\draw[fill]  (0, 0) circle (3pt);

\node[below] at (0,0) {\small $\cP_{t-1}$};

\draw[fill]  (-5, 5) circle (3pt);

\node[above] at (-5,5) {\small  $\cE_{t}$};

\draw[fill]  (-2.5, 0) circle (3pt);

\node[above right] at (-2.5,0) {\small  $Z$};

\draw[fill] (0.5,5) circle (3pt);

\node[above right] at (0.5,5) {\small $\cE_{t-1}$};

\draw[fill] (-3.5,2) circle (3pt);

\node[left] at (-3.5,2) {\small $\cP_{t}$};

\node[above] at (-3,0) {\small $a$};

\node[left] at (-3.5,1) {\small $\sqrt{3}a$};

\draw[thick] (-2.5, -.5) -- (-5, -.5);
\draw[thick] (-2.5, -.25) -- (-2.5, -.75);
\draw[thick] (-5, -.25) -- (-5, -.75);

\draw [thick] (-3.7,-.5) node[below]{\small $1/2$};
\draw [thick] (-2.5,5) node[above]{\small $\leq 1$};

\node at (-3,-2) {\scriptsize (b)};

\end{tikzpicture}

\end{center}
\caption{Counter-moves during Rook Mode.  (a) If $\cE$ hides behind a feature on $\Pi_U$ then  $\cP$ responds by moving to a point directly below the evader's previous position. (b) The pursuer response when $\cE$ moves leftward past $\cP$ and $|x(\cE_{t}) - x(\cP_{t-1})| \leq 1$.}
\label{fig:countermove}
\end{figure}

\begin{proof}
Without loss of generality, $\cP$ is in upper pocket position and $\cE$ is in the first quadrant of $\cP$. Suppose that  $0 \leq \dx(t-1) \leq 1/2$ with  $\overline{\cP_{t-1} \cE_{t-1}} \subset Q$, while $\Pi_U$ obstructs $\overline{\cP_t \cE_{t}}$. 
By Lemma \ref{lemma:pocket-viz}, $\cE_{t}$ must  be to the right of $\cP_{t-1}$.
We show that the pursuer can move to
$\cP_{t} = \cP_{t-1} + \left(\dx(t-1), \sqrt{ 1- \dx(t-1)^2} \right)$, which has visibility to $\cE_{t}$.
By construction, $x(\cP_{t}) = x(\cE_{t-1})$. We also have $y(\cP_{t}) < y(\cE_{t-1})$, since otherwise $\cE$ was caught at time $t-1$. Since 
$\overline{\cP_{t-1} \cE_{t-1}} \subset Q$ and $\frontier(\cP_{t-1}) \subset Q$, we clearly have $\overline{\cP_{t-1} \cP_{t}} \subset Q$, see Figure \ref{fig:countermove} (a). Similarly, since $\overline{\cE_{t-1} \cE_{t}} \in Q$, we also have $\overline{\cP_{t} \cE_{t}} \in Q$, so the point $\cP_{t}$ sees the point $\cE_{t}$. The pursuer's progress is $\sqrt{ 1- \dx(t-1)^2} \geq \sqrt{3}/{2}$ since $\dx(t-1) \leq 1/2$.

Finally, we show that $\cP_{t}$ is in rook position. Since $x(\cP_{t}) = x(\cE_{t-1})$, it is sufficient to  prove that $| x(\cE_{t}) - x(\cE_{t-1}) | \leq 1/2$.
Suppose that the evader moves leftward. By Lemma \ref{lemma:hide-left}, we have $x(\cP_{t-1}) < x(\cE_{t})$ and therefore $| x(\cE_{t}) - x(\cE_{t-1}) | < | x(\cP_{t-1}) - x(\cE_{t-1}) |  \leq 1/2$. 
Suppose that the evader moves rightward. 
We claim that the slope of $\overline{\cE_{t-1}\cE_{t}}$ must be greater than $\sqrt{3}$.  
Since $\cE_{t}$ is not visible to $\cP_{t-1}$, the slope of $\overline{\cE_{t-1}\cE_{t}}$
must be strictly larger than the slope of $\overline{\cP_{t-1}\cE_{t-1}}$, as shown in Figure \ref{fig:countermove} (a). However, $\cP$ was in rook position with $\dx(t-1) \leq 1/2$ and therefore $\dy(t-1) \geq \sqrt{3}/{2}$, as otherwise the evader was caught at time $t-1$. Therefore the slope of  $\overline{\cP_{t-1}\cE_{t-1}}$  is at least $\sqrt{3}$, so $\dx(t) \leq 1/2$.

The argument for $-1/2 \leq \dx(t-1) \leq 0$ is entirely analogous.
\end{proof}

Our final  lemma shows that the pursuer can advance the rook frontier when $\cE$ does not significantly increase his horizontal distance from $\cP$.

\begin{lemma}
\label{lemma:7/22}
Suppose  $\cP$ and $\cE$ are in pocket position at time $t-1$.  If the evader moves so that $\cE_{t}$ is visible to $\cP_{t-1}$ and
$|x(\cE_{t}) - x(\cP_{t-1}) | \leq 1$ then the pursuer can re-establish rook position  while also  advancing the rook frontier by at least $7/22$.
\end{lemma}

\begin{proof}
 Without loss of generality, $\cP$ is in upper pocket position and $\cE$ is in the first quadrant of $\cP$, so that
 $0 \leq x(\cE_{t-1}) - x(\cP_{t-1}) \leq 1/2$. 
Let $\alpha$ be the closer of $x(\cE_t) \pm 1/2$ to $x(\cP_{t-1})$, and let $Z$ be the point with $x(Z)=\alpha$
that is closest to $\cP_{t-1}$, so that $d(Z, \cP_{t-1}) \leq 1/2$.  Let $\cP_{t}$ be the point on $\overline{Z \cE_{t}}$ at distance one from $\cP_{t-1}$. A leftward evader move is shown in Figure \ref{fig:countermove} (b), and the rightward move results in an equivalent situation. The visibility of
$\cE_{t}$ from $\cP_{t-1}$ ensures that both of $\overline{\cP_{t-1} \cP_{t}}$ and $\overline{\cP_{t} \cE_{t}}$ are in $Q$. By Lemma \ref{lemma:guard}, $\cE_{t}$ is at least $\sqrt{3}/2$ above $\frontier(\cP_{t-1})$, so the absolute value of the slope of $\overline{Z\cE_{t}}$ is at least $\sqrt{3}$. In addition, $|x(P_{t-1}) - x(Z)| \leq 1/2$ and
$|x(\cE_{t})-x(\cP_{t})|\le|x(\cE_{t})-x(Z)|=1/2.$ 
Setting $a = |x(P_{t}) - x(Z)|$, the distance from $\cP_{t}$ to $\frontier(\cP_{t-1})$ is at least $\sqrt{3} a$.
By the Pythagorean theorem, 
$1 \leq (1/2+a)^2 + 3a^2$, or equivalently, $a \geq (\sqrt{13} -1 )/8.$ Therefore, the pursuer has advanced the rook frontier by $\sqrt{3}a > 7/22$.
\end{proof}

We are now ready to prove {Lemma \ref{lemma:endgame}}.

\bigskip

\begin{proofof}{Lemma \ref{lemma:endgame}}
For simplicity of notation, we start at $t=0$.
Let $s$ be the length of the rook frontier $\frontier(\cP_0)$. We show that after at most $2s$ turns, either the evader is caught or the pursuer advances the rook frontier by at least $7/22$. We  repeat this process until the height of the evader territory is at most $\sqrt{3}/2$, so that the evader is caught by Lemma \ref{lemma:guard}. The total number of turns to catch the evader is $O(\Area(Q_{\cE}))$, as explained below.

Let $Q_{\cE}$ denote the initial evader territory, bounded by $\Pi_U$ and $\frontier(\cP_0)$. 
If $s \leq 1$ then $\cP$ can move to the midpoint of the frontier on his next turn and then advance the frontier by 1 on all successive turns until catching $\cE$. So we consider $s > 1$. Without loss of generality, we assume that $x(\cP_0)  \leq x(\cE_0) \leq x(\cP_0) + 1/2$ and $y(E_0) > y(P_0) + \sqrt{3}/2$. We show that after at most $2s$ moves, either the evader is caught or the pursuer advances the rook frontier by at least $7/22$.

Suppose that $\cP_{t-1}$ is in rook position with respect to $\cE_{t-1}$ and that the evader moves to $\cE_{t}$. There are three cases to consider. First, if
$\cE_{t-1}$ is not visible to $\cP_{t-1}$, then the pursuer moves according to Lemma \ref{lemma:re-establish}, advancing the frontier by $\sqrt{3}/2 > 7/22$. Second, if $|x(\cE_{t}) - x(\cP_{t-1})| \leq 1$ then the pursuer moves according to Lemma \ref{lemma:7/22}, advancing the frontier by at least $7/22$.
Third, suppose that $1 < |x(\cE_{t}) - x(\cP_{t-1})| \leq 3/2$, which only occurs when $\cE$ moves leftward by at least $1/2$. In this case, the pursuer responds by moving to the highest point reachable on the vertical line through $x = x(\cE_{t})-1/2$. (If $|x(\cE_{t}) - x(\cP_{t-1})| =3/2$ then this is a purely horizontal move, but otherwise the pursuer makes some vertical progress.) The evader can only make $2 s$ such moves before he is forced to make a move covered by the previous two cases. Therefore, after at most $2 s$ moves, the pursuer advances the frontier by at least $7/22$.

The worst-case scenario for the capture time is when the evader repeatedly runs the width of the region using horizontal steps of distance 1. This zig-zagging forces the pursuer to trace a path of length $O(\Area(Q_{\cE}))$. 
\end{proofof}


\subsubsection{Chute Position}

We require one additional lemma concerning chute position. 

\begin{lemma}
\label{lemma:pursuit}
Suppose that the pursuer is in chute position. Using Monotone Rook's Strategy, $\cP$ can either capture $\cE$, or she can update her checkpoint to the right. 
\end{lemma}

\begin{proof}
Suppose that  $\cP$ is in upper chute position (a lower chute is analogous).   If the evader never makes an escape move, we claim that the pursuer will establish an upper pocket position. Observe that  
$y(\cP_t) > y(\cP_{t-1})$ whenever $\cE$ does not move completely horizontally, or when $\cE$ switches his horizontal direction. Each time that $y(\cP_t) > y(\cP_{t-1})$,  the pursuer increases the height the rook frontier. If $\cE$ never makes an escape move, then   right endpoint of the rook frontier must eventually transition from the lower chain  to the upper  chain. We are now in upper pocket position, so  the pursuer catches the evader by Lemma \ref{lemma:endgame}.

Suppose that $\cE$ makes an escape move. That is,  $\cE$ makes a rightward move above a lower  feature that obstructs the pursuer's rook response, see Figure \ref{fig:blocked} (a). Let $Z \in \Pi_U$ be the point that obstructs the pursuer's desired movement. In response, $\cP$ re-enters Search Mode. Her new search path, which is constructed via Algorithm \ref{alg:monotone-path}, starts from her current location, and then continues rightward, as shown in Figure \ref{fig:blocked} (b). After her first move, she encounters the blocking feature, so her checkpoint $c$ updates to  the blocking point $Z$, or some point  on the blocking feature even further to the right. In other words, the entire rook frontier is now guarded. This updates the checkpoint compared to the last time that $\cP$ was in Search Mode. 
\end{proof}


\begin{figure}

\begin{center}

  \begin{tikzpicture}[scale=1,y=-1cm]

	\draw (1.7,1.3) -- (2.1,1.9) -- (2.6, .2) -- (3.2, 2) -- (3.5, 1) -- (3.75,1.8) -- (3.9,.5) --(4.3,-.2) -- (5, 2.1) -- (5.5,-.75) -- (6.25,2.2) -- (6.8,.7);
	\draw   (1.7,-1) -- (1.9,0.5) -- (3,-1.75) -- (3.5,-.75) -- (4.2,-1.5) -- (4.9, -.5) -- (5.5,-1.6) -- (5.75,-1) -- (6,-1.7) -- (6.5,0);

\draw[very thick] (2.15,0) -- (4.175, 0);
\draw[fill] (2.15,0.0)  circle (1.5pt);
\node[below] at (2.2,0.1)  {$c$};
\draw[fill] (4.175,0.0)  circle (1.5pt);
\node[below] at (4.2,0.1)  {$Z$};

\draw[fill] (3.7, 0)  circle (2pt);
\node[below] at (3.5, .05)  {\small $\cP_{t-1}$};

\draw[fill] (4.1, -.5)  circle (2pt);
\node[above] at (4.2, .-.5)  {\small $\cE_{t-1}$};

\draw[fill] (4.6, .05)  circle (2pt);
\node[below] at (4.9, .05)  {\small $\cE_{t}$};

\draw[-latex] (4.2, -.4) -- (4.55, -.05);

\node at (4,2) {\scriptsize (a)};

\begin{scope}[shift={(6.5,0)}]

	\draw  (1.7,1.3) -- (2.1,1.9) -- (2.6, .2) -- (3.2, 2) -- (3.5, 1) -- (3.75,1.8) -- (3.9,.5) --(4.3,-.2) -- (5, 2.1) -- (5.5,-.75) -- (6.25,2.2) -- (6.8,.7);
	\draw    (1.7,-1) -- (1.9,0.5) -- (3,-1.75) -- (3.5,-.75) -- (4.2,-1.5) -- (4.9, -.5) -- (5.5,-1.6) -- (5.75,-1) -- (6,-1.7) -- (6.5,0) -- (6.8,-.58);

\draw[very thick] (2.15,0) -- (4.175, 0) -- (4.3,-.2) -- (5.40351,-0.2) -- (5.5,-.75) -- (6.27941,-.75) --(6.5,0) -- (6.8,0);

\draw[fill] (3.7, 0)  circle (2pt);
\node[below] at (3.5, .05)  {\small $\cP_{t-1}$};

\draw[fill] (4.3,-.2)  circle (2pt);
\node[above] at (4,-.3)  {\small $c = \cP_{t}$};
\draw[-latex] (3.8, -.075) -- (4.2,-.2);

\draw[fill] (4.6, .05)  circle (2pt);
\node[below] at (4.9, .05)  {\small $\cE_{t}$};

\node at (4,2) {\scriptsize (b)};

\end{scope}

  \end{tikzpicture}

\caption{Transition from Rook Mode back to Search Mode. (a) $\cE$ makes an escape move, so that  $\cP$'s responding rook move is obstructed by point $Z$. (b) Instead, $\cP$ counters by re-entering Search Mode with an updated search path and updated checkpoint $c$. }

\label{fig:blocked}

\end{center}

\end{figure}

%

\subsection{Catching the Evader}\label{sec:time}

We can now prove Theorem \ref{thm:monotone}. 

\bigskip

\begin{proofof}{Theorem \ref{thm:monotone}}
The pursuer alternates between Search Mode and Rook Mode until she captures the evader: see Figure \ref {fig:example-monotone} for an example pursuer trajectory. Once the pursuer is in Search Mode, she will establish rook position by Lemma \ref{lemma:search}. The evader can hide at most $n$ times (once for each vertex of $Q$), so we switch modes $O(n)$ times. Indeed, $\cE$  cannot hide behind the same vertex twice because $\cP$ does not exit Search Mode until $0 \leq x(\cE) - x(\cP) \leq 1/2$ and $\cP$ has visibility of $\cE$.  
Once $\cP$ is in Rook Mode, $\cE$ must make an escape move in order to hide.  Therefore $\cE$ must hide behind a new vertex to the right of the previous hiding spot. 

Overall, the pursuer spends at most $O(n(Q) + \diam(Q))$ turns in Search Mode: the factor of $n$ appears since $\cP$ defensively stops every time she passes the $x$-coordinate of a vertex of $Q$, or is blocked by a boundary segment. Note that $\cP$ can only pass by a vertex in Search Mode at most once (as any future passings of that vertex must occur in Rook Mode), and each of the $n$ boundary segments can obstruct the pursuer at most once. Therefore the number of truncated steps is $O(n)$.

Each time $\cP$ switches to Rook Mode, she either captures $\cE$, or $\cE$ makes an escape move. In the latter case, the pursuit continues to the right of the current player positions. Crucially, throughout the pursuit, $\cP$ never visits the same point twice. Once in Rook Mode, $\cP$ guards her pursuer frontier by Lemma \ref{lemma:guard}, so the evader cannot cross it. Any vertical move by the pursuer advances this frontier, further limiting the pursuer territory. Furthermore, by Lemma \ref{lemma:7/22}, $\cP$ is able to advance the rook frontier by at least $7/22$ once every $\diam(Q)$ moves. Therefore, the entire pursuit path taken by $\cP$ has length $O(\Area(Q))$. 

Every time  the pursuer switches from Rook Mode back to Search Mode, she updates her checkpoint and the evader territory shrinks. Therefore, the pursuer ultimately traps the evader in a pocket position and achieves capture by Lemma \ref{lemma:endgame}.
Combining the time bounds for both Search Mode and Rook Mode, the pursuer captures the evader in $O(n(Q) + \Area(Q) )$ turns.
\end{proofof}

This completes our discussion of pursuit in a monotone polygon. In the next  section, we describe how to adapt this basic algorithm to  scallop polygons.

%

\section{Pursuit in a Scallop Polygon}
\label{sec:scallop}

A polygon $Q$ is a \style{scallop polygon} when it can be swept by rotating a line $L$ through some center point $C$ outside the polygon such that the intersection $L \cap Q$ is a line segment. An example of a scallop polygon is shown in Figure \ref{fig:mono-scallop} (c). 
We  use polar coordinates $(r:\theta)$ where $C$ is located at the origin, and the polygon is located in the upper half plane. We will sweep our polygon in the clockwise direction, so we index the vertices of our polygon by decreasing angle. More precisely, let the $n$ vertices of $Q$ be $v_k = (r_k : \theta_k)$ for $1 \leq k \leq n$ where $r_k \geq 0 $ for all $k$ and and $\pi > \theta_1 > \theta_2 > \cdots > \theta_n \geq 0$. (Recall  that our vertices and our sweep center are in general position, so all these angles are distinct.) For a point $Z \in Q$,  let $\sweepline{Z}$ denote the sweep line through $Z$, and  let $\theta(Z)$ denote the angle of $\sweepline{Z}$.

We begin with an overview of the pursuer's strategy. 
Algorithm \ref{alg:monotonePursuit} still describes the  pursuit algorithm: the pursuer alternates between Search Mode and Rook Mode until the evader is captured. 
Similar to searching in a monotone polygon, the Scallop Search Strategy guards an increasing region as the pursuer traverses clockwise through the polygon, measured by the decreasing angle of the sweep line. Due to the changing angle, search path construction is more involved and  the transition from Rook Mode back into Search Mode requires more care. Finally, the evader has two new ways to foil rook position: a hiding move and a blocking move. So the pursuer needs a recovery mode to respond to these disruptions.

\subsection{Overview of Scallop Search Strategy}
\label{sec:scallop-search-overview}

We give a high-level description  of the search strategy in a scallop polygon, and define the required terminology, focusing on the adaptations required to generalize the strategy for monotone search. 
The algorithms and proofs are in Section \ref{sec:scallop-search}.

\begin{figure}

\begin{center}

\subfigure[]
{

\begin{tikzpicture}[scale=.8]

\foreach \ang in {160, 80, 70, 60, 50} {
\draw[dotted] (\ang:0) -- (\ang:5);
}

\foreach \ang in {150, 145, 130, 120, 115, 110, 100, 90} {
\draw[dotted] (\ang:0) -- (\ang:6);
}

\draw (160:4) -- (150: 2) -- (130:3) -- (120: 2) -- (100:3.7) -- (80:1) -- (70:2) -- (60: 3.05) -- (50: 1) -- (35: 3);

\draw (160:4) -- (150:6) -- (145: 3.5) -- (130: 6) -- (115:4) -- (110:6) -- (90: 5) -- (80:3.6) -- (60:4.5) -- (35:3);

\draw[fill] (0,0) circle (3pt);

\node at (.5,0) {$C$};


\begin{scope}[shift={(160:4)}]
\draw[very thick] (0:0) -- (70:.7);
\draw[very thick] (55:.7) -- (55:.975);
\draw[very thick] (30:2.56) -- (30:2.84);
\end{scope}

\draw[very thick] (150:4.07) -- (150:3.86);
\draw[very thick] (150:4.07) -- (150:3.86);
\draw[very thick] (146:3.87) -- (145:3.5);

\begin{scope}[shift={(145:3.5)}]
\draw[very thick] (0:0) -- (55:.95);
\draw[very thick] (40:.9) -- (40:1.5);
\end{scope}

\draw[very thick] (130:3.64) -- (130:3.355);
\draw[very thick] (120:3.45) -- (120:3.05);
\draw[very thick] (115:3.1) -- (115:2.88);
\draw[very thick] (110:2.915) -- (110:2.795);

\begin{scope}[shift={(130:3)}]
\draw[very thick] (25:.77) -- (25:1.04);
\draw[very thick] (20:1.03) -- (20:1.18);
\end{scope}

\draw[very thick] (107:2.83) -- (100:3.7);

\begin{scope}[shift={(100:3.7)}]
\draw[very thick] (10:0) -- (10:.665);
\draw[very thick] (0:.63) -- (0:1.235);
\draw[very thick] (-10:1.25) -- (-10:1.9);
\draw[very thick] (-20:1.85) -- (-20:2.42);
\end{scope}

\draw[very thick] (90:3.78) -- (90:3.63);
\draw[very thick] (81:3.7) -- (80:3.6) -- (80:3.46);
\draw[very thick] (70:3.55) -- (70:3.18);
\draw[very thick] (60:3.25) -- (60:3.05);

\begin{scope}[shift={(60:3.05)}]
\draw[very thick] (-30:0) -- (-30:.55);
\draw[very thick] (-40:.53) -- (-40:1.21);
\end{scope}

\draw[very thick] (50:3.11) -- (50:2.98);
\draw[very thick] (37.35:3.085) -- (35:3);

\end{tikzpicture}
}
\subfigure[]
{
\begin{tikzpicture}[scale=1]

\draw[dotted] (160:1.5) -- (160:5);

\foreach \ang in {150, 145, 130, 120, 115, 110, 100} {
\draw[dotted] (\ang:1.5) -- (\ang:6);
}

\draw (160:4) -- (150: 2) -- (130:3) -- (120: 2) -- (100:3.7) -- (90:2.35);

\draw (160:4) -- (150:6) -- (145: 3.5) -- (130: 6) -- (115:4) -- (110:6) -- (90: 5) ;

\begin{scope}[shift={(160:4)}]
\draw[very thick] (0:0) -- (70:.7);
\draw[dashed] (0:0) -- (55:.7);
\draw[very thick] (55:.7) -- (55:.975);
\draw[dashed] (0:0) -- (30:2.56);
\draw[very thick] (30:2.56) -- (30:2.84);
\end{scope}

\draw[very thick] (150:4.07) -- (150:3.86);
\draw[very thick] (150:4.07) -- (150:3.86);
\draw[very thick] (146:3.87) -- (145:3.5);

\begin{scope}[shift={(145:3.5)}]
\draw[very thick] (0:0) -- (55:.95);
\draw[dashed] (0:0) -- (40:.9);
\draw[very thick] (40:.9) -- (40:1.5);
\end{scope}

\draw[very thick] (130:3.64) -- (130:3.355);
\draw[very thick] (120:3.45) -- (120:3.05);
\draw[very thick] (115:3.1) -- (115:2.88);
\draw[very thick] (110:2.915) -- (110:2.795);

\begin{scope}[shift={(130:3)}]
\draw[dashed] (0:0) -- (25:.77);
\draw[very thick] (25:.77) -- (25:1.04);
\draw[dashed] (0:0) -- (20:1.03);
\draw[very thick] (20:1.03) -- (20:1.18);
\end{scope}

\draw[very thick] (107:2.83) -- (100:3.7);

\begin{scope}[shift={(100:3.7)}]
\draw[very thick] (10:0) -- (10:.665);

\end{scope}

\draw[fill]  (160:4)  circle (1.5pt);
\node[below] at (160:4) {$c_1$};

\draw[fill]  (130:3) circle (1.5pt);
\node at (136:2.9) {$c_2$};

\draw[fill]  (100:3.7)  circle (1.5pt);
\node[above left] at (100:3.7)  {$c_3$};

\draw[fill]  (145: 3.5) circle (1.5pt);
\node at (140: 3.68) {$a$};

\end{tikzpicture}
}

\end{center}

\caption{The search path of a scallop polygon. (a) At each spoke line, we adjust the path to our new frame of reference  to guard the checkpoint vertex. (b)
Guarding upper boundary vertex $a$ is only temporary: later on, we revert back to guarding  $c_1$ directly. Later still, the pursuer updates her checkpoint to $c_2$, which also 
guards the previous checkpoint $c_1$.   }
\label{fig:scallop-search}

\end{figure}

\begin{definition}
A \style{search path} in scallop polygon $Q$ is a path $\sp$ constructed by Algorithm \ref{alg:scallop-path}  Create Scallop Search Path.
\end{definition}

Figure \ref{fig:scallop-search} (a) shows an example search path. The search path construction algorithm is similar to Algorithm \ref{alg:monotone-path} for monotone polygons. The primary difference is how we handle the notions of ``horizontal'' and ``vertical.'' During Search Mode, we change our frame of reference every time we pass a vertex of the polygon, matching the sweeping nature of the polygon itself. 
The secondary difference 
is that only points on the lower boundary are taken as checkpoints. The search path can include vertices on $\Pi_U$. However, we may have to relinquish these upper boundary points as our frame of reference rotates. Figure \ref{fig:scallop-search} (b) shows an example of an upper vertex  that goes from guarded to unguarded as $\cP$ traverses $\Pi$ down the sweep line. Therefore, we need to distinguish checkpoints on the lower boundary from the  temporary checkpoints  on the upper boundary. 

\begin{definition}
Let $Q$ be a scallop polygon with center point $C$.
Let $v_1, v_2, \ldots, v_n$ be the vertices of $Q$, ordered by decreasing angle. For $1 \leq k \leq n$, the \style{spoke line} $S_k = S(v_k)$ is the line through $C$ and $v_k$. The \style{transverse line} $T_k = T(v_k)$ is the line through $v_k$ that is perpendicular to $S_k$. Taken together,  $(T_k, S_k)$ define a \style{frame of reference}. A pursuer using this frame of reference has a well-defined notion of up, down, left and right.
\end{definition}

During Search Mode, the pursuer adjusts her frame of reference so that horizontal and vertical match the last vertex $v_{\ell}$ that she has passed. Namely, for point $P$ where $\theta(v_{\ell}) \geq \theta(P) > \theta(v_{\ell+1})$, the frame of reference $(X,Y)$ at $P$ will be the translate of $(T_{\ell}, S_{\ell})$, the transverse and spoke lines for vertex $v_{\ell}$. For convenience, we express this translation equivalence relation as $(X,Y) \sim (T_{\ell}, S_{\ell})$.
Before defining checkpoints on $\Pi_L$ and auxiliary  vertices on $\Pi_U$ (their temporary analogs), we must adapt our notion of guarding for a scallop polygon.
Figure \ref{fig:scallop-guard} offers simple examples of guarded vertices.

\begin{definition}
\label{def:scallop-guarded}
Consider a pursuer $\cP$  a scallop polygon $Q$ where $\theta(v_{\ell}) \geq \theta(\cP) > \theta(v_{\ell+1})$ which  is using frame of reference $(X,Y) \sim (T_{\ell}, S_{\ell})$ centered at $\cP$. Suppose that
$
x(\cE) \geq x(\cP) - 1/2
$
with respect to frame $(X,Y)$.

Let $a \in \Pi_U$ be a  vertex on the upper boundary  such that $\theta(a)  \geq \theta(\cP)$,  and consequently $x(a) \leq x(\cP)$. Then $\cP$ \style{guards}  $a \in \Pi_U$  with respect to $(X,Y)$ when  $a \in \horizlineneg{\cP}$. In other words,  $y(a) =  y(\cP)$ and the line segment between $a$ and $\cP$ is contained in $Q$.

Let $c \in \Pi_L$ be a point on the lower boundary  with $\theta(c)  \geq \theta(\cP)$,  
such that $c$ is contained in the subpolygon $Q' \subset Q$ created by $\horizlineneg{\cP}$. 
Then   $\cP$ \style{guards}  $c$ with respect to $(X,Y)$  when one of the following holds:
\begin{enumerate}
\item[(a)]  $c \in \horizlineneg{\cP}$. In other words, $y(c) =  y(\cP)$ and the line segment between $c$ and $\cP$ is contained in $Q$. 
\item[(b)] $y(c) \leq y(\cP)$  and either
\begin{enumerate}
\item[(i)] $y(\cE) > y(\cP)$ and $\cP \in  \Pi_L$ is on the lower boundary of $Q$, or
\item[(ii)] $\horizlineneg{\cP}$ intersects the highest local maximum $c' \in  \Pi_L$ with $x(c) < x(c') < x(\cP)$.
\end{enumerate}
\item[(c)] $y(c) \geq y(\cP)$ and $\horizlineneg{\cP}$ intersects the lowest local minimum $a \in \Pi_U$ with $x(c) < x(a) < x(\cP)$.
\end{enumerate}
\end{definition}

\begin{figure}[ht]

\begin{tikzpicture}[scale={.9}]

\begin{scope}

\foreach \ang in {150, 110, 90, 70, 50} {
\draw[dotted] (\ang:.25) -- (\ang:2.15);
}

\draw  (150: 2) -- (90:.5) -- (50: 1.8) ;

\draw (150: 2)  -- (110:2.1) -- (70:1.85) -- (50: 1.8);

\draw[fill] (150:2) circle (1.5pt);
\node[below left]  at (143:1.75) {{\scriptsize $Q' \!\! =$} {\small $ \! c$}};

\draw[fill] (90:.5) circle (1.5pt);
\node[below]  at (90:.5) {\small $v_{\ell}$};


\begin{scope}[shift={(150:2)}]
\draw[thick, dashed] (0:1.9) -- (0:0);
\draw[fill] (0:1.9) circle (2pt);
\node[above right]  at (0:1.65) {\small $\cP$};
\end{scope}


\node at (-.2,-.65) {\small  (a)};

\end{scope}

\begin{scope}[shift={(3.5,0)}]

\fill[gray!33] (150:2) -- (130:1) -- (0,1.25) -- (-1.65,1.25);

\foreach \ang in { 50} {
\draw[dotted] (\ang:.5) -- (\ang:1.75);
}

\foreach \ang in {150, 130, 120, 90, 80} {
\draw[dotted] (\ang:.5) -- (\ang:2.15);
}

\foreach \ang in {120, 90} {
\draw[dotted] (\ang:.5) -- (\ang:2.5);
}

\draw[fill] (90:2.5) circle (1.5pt);
\node[above]  at (90:2.5) {\small $v_{\ell}$};

\draw  (150: 2) -- (130:1) --  (80:1.5) -- (50: 1.5) ;

\draw (150: 2) -- (120:2.5) -- (90:2.5) -- (50: 1.5);

\draw[fill] (150:2) circle (1.5pt);
\node[below left]  at (150:2) {\small $c$};

\draw[fill] (0,1.25) circle (2pt);
\node[above left]  at (0,1.25) {\small $\cP$};
\draw[thick, dashed] (-1.65,1.25) -- (0,1.25) ;

\node  at (-.9,1.025) {\scriptsize $Q'$};

\node at (-.2,-.65) {\small  (b.i)};

\end{scope}

\begin{scope}[shift={(7,0)}]

\fill[gray!33] (150:2) -- (110:1.75)  -- (-1.48,1.63) --cycle;

\foreach \ang in {150, 120,110, 90, 70, 50} {
\draw[dotted] (\ang:.5) -- (\ang:2.15);
}

\foreach \ang in {120,110, 90} {
\draw[dotted] (\ang:.5) -- (\ang:2.5);
}

\draw  (150: 2) -- (110:1.75) --  (70:1) -- (50: 2) ;

\draw (150: 2) -- (120:2.5) -- (90:2.5) -- (50: 2);

\draw[fill] (90:2.5) circle (1.5pt);
\node[above]  at (90:2.5) {\small $v_{\ell}$};

\draw[fill] (150:2) circle (1.5pt);
\node[below left]  at (150:2) {\small $c$};

\draw[fill] (.2,1.63) circle (2pt);
\node[above right]  at (0,1.63) {\small $\cP$};
\draw[fill] (-.6,1.63) circle (1.5pt);
\draw[thick, dashed] (-1.48,1.63) -- (.2,1.63);

\node[above]  at (-.4,1.63) {\small $c'$};

\node  at (-1.3,1.425) {\scriptsize $Q'$};

\node at (-.2,-.65) {\small  (b.ii)};

\end{scope}

\begin{scope}[shift={(10.5,0)}]

\fill[gray!33] (150: 2) -- (130:2.1) -- (120:1) -- (-1.4,.85);

\foreach \ang in {150, 140, 130, 120, 90, 60, 50} {
\draw[dotted] (\ang:.25) -- (\ang:2.15);
}

\foreach \ang in {90} {
\draw[dotted] (\ang:.25) -- (\ang:2.5);
}

\draw  (150: 2) -- (140:.5) --  (60:.5) -- (50: 2) ;

\draw (150: 2) -- (130:2.1) -- (120:1)  -- (90:2.5) -- (50: 2);

\draw[fill] (90:2.5) circle (1.5pt);
\node[above]  at (90:2.5) {$\small v_{\ell}$};

\draw[fill] (150:2) circle (1.5pt);
\node[below left]  at (150:2) {\small $c$};
\node[above right]  at (120:1) {\small $a$};

\draw[fill] (.2,.85) circle (2pt);
\node[above] at (.2,0.85) {\small $\cP$};
\draw[fill] (-.5,.85) circle (1.5pt);
\draw[thick, dashed] (-1.4,.85) -- (.2,.85);

\node  at (-1.25,1.15) {\scriptsize $Q'$};

\node at (-.2,-.65) {\small  (c)};

\end{scope}

\end{tikzpicture}

\caption{Pursuer $\cP$ guards vertex $c \in \Pi_L$ of subpolygon $Q'$ with respect to the current frame $(X,Y) \sim (T_{\ell}, S_{\ell})$, provided that $x(\cE) \geq x(\cP) - 1/2$. 
The subfigure labels match the cases listed in Definition \ref{def:scallop-guarded}. 
Subfigure (b.i) requires that $y(\cE) > y(\cP)$ as well.
Subfigure (b.ii) also contains  guarded vertex $c' \in \Pi_L$. Subfigure (c) also contains  guarded vertex $a \in \Pi_U$.
}
\label{fig:scallop-guard}

\end{figure}

As the pursuer in Search Mode moves through the polygon, she updates her notion of ``vertical'' to match the last polygon vertex that she passed by. Therefore $\cP$ enters an adjustment phase every time that she reaches a spoke line $S_k$, where she transitions her frame of reference from $(T_{k-1},S_{k-1})$ to $(T_k, S_k)$, see Figure \ref{fig:scallop-search}(b). Whenever $\cP$ encounters a spoke line during search, the pursuer halts.  She then traverses down the spoke line until she reaches a point that guards the current checkpoint $c \in \Pi_L$  with respect to her new frame of reference. 
She then adopts frame $(T_k,S_k)$ and continues forward. Throughout this transition, she must be aware of both her old and new frames of reference, in case $\cE$ tries to invade the pursuer territory.

Having offered a high-level description of our search path, we can now articulate the update rule for the checkpoint and the auxiliary vertex.

\begin{definition}
\label{def:scallop-checkpoint}
The searching pursuer's \style{checkpoint} $c_t$ and \style{auxiliary vertex} $a_t$ are defined recursively. At $t=0$, we have $c_0 = v_1$ and $a_0 = \emptyset$. For $t>0$, the update rule 
depends on the movement from $\cP_{t-1}$ to $\cP_{t}$. There are four movement types leading to eight different update rules.  Examples are  shown in  Figure  \ref{fig:scallop-update-checkpoint}.

\begin{enumerate}
\item[(a)] The pursuer moved horizontally with respect  to her current frame. If we reach the lower boundary $\cP_t \in \Pi_L$ then update $c_t=\cP_t$ and reset $a_t=\emptyset$. Otherwise, keep $c_t=c_{t-1}$ and $a_t = a_{t-1}$.
\item[(b)] The pursuer moved along the lower boundary. Update $c_t = c_{t-1}$ and keep $a_t = \emptyset$.
\item[(c)] The pursuer moved along the upper boundary. 
\begin{enumerate}
\item[(i)] We did not reach a vertex. Keep $c_t = c_{t-1}$ and $a_t = a_{t-1}$.
\item[(ii)] We reached vertex $v_k$. Keep $c_t = c_{t-1}$. If $c_{t-1}$ is guarded by frame $(X',Y')=(T_{k},S_{k})$   
then  update $a_t = v_{k}$. Otherwise keep $a_t = a_{t-1}$. 
\end{enumerate}
\item[(d)] The pursuer moved vertically  along  spoke line $S_k$ in order to transition from frame $(X,Y) \sim (T_{k-1}, S_{k-1} )$ to frame $(X',Y') \sim (T_k,S_k)$. If checkpoint $c_{t-1}$ is not yet guarded by  $\cP_t$ in frame $(X', Y')$, then there are no updates (and $\cP$ will continue to move downward). If $c_{t-1}$ is guarded in $(X',Y')$, then  we update our checkpoint depending on the method of guarding, see Definition \ref{def:scallop-guarded}.
\begin{enumerate}
\item[(i)] $c_{t-1} \in X'_{-}(\cP_t).$ Keep $c_t = c_{t-1}$ and reset $a_t = \emptyset$. The  new  frame of reference directly guards the current checkpoint.
\item[(ii)] $y'(c) \leq y'(P)$ and
\begin{enumerate}
\item[($\alpha$)] $\cP_t \in \Pi_L$. Set $c_t = \cP_{t}$ and reset $a_t = \emptyset$. The pursuer is located on her  new checkpoint.
\item[($\beta$)] $\horizlineneg{\cP_t}$ contains the highest local maximum $v \in  \Pi_L$ with $x'(c) < x'(v) < x'(P)$. Set $c_t = v$ and reset $a_t = \emptyset$. This updated checkpoint protects the previous one.
\end{enumerate}
\item[(iii)] $\horizlineneg{\cP_t}$ contains the lowest local minimum $v \in  \Pi_U$ with $x'(c) < x'(v) < x'(P)$. Keep $c_t = c_{t-1}$ and set $a_t = v$. The new auxiliary vertex protects the previous checkpoint.
\end{enumerate}

\end{enumerate}
\end{definition}

\begin{figure}

\begin{tikzpicture}[scale=.9]

\begin{scope}


\draw  (150: 2) -- (70:.5) -- (50: 1.8) ;

\draw (150: 2)  -- (110:2.1) -- (60:1.85) -- (50: 1.8);

\begin{scope}[shift={(110:2.1)}]

\draw[dotted] (0,-1.5) -- (0,.25);

\end{scope}

\draw[fill] (150:2) circle (1.5pt);
\node[below]  at (150:2) {\small $c$};


\begin{scope}[shift={(150:2)}]

\draw[thick, dashed] (0:0) -- (0:1.85);
\draw[fill] (0:1.5) circle (2pt);
\draw[fill] (0:1.85) circle (2pt);
\node[above]  at (0:1.45) {\small $\cP$};
\node[above]  at (0:2) {\small $\cP'$};
\end{scope}

\node at (-.2,-.15) {\small  (a)};

\end{scope}

\begin{scope}[shift={(3.5,0)}]




\draw  (150: 2) -- (130:1) --  (80:1.5) -- (50: 1.5) ;

\draw (150: 2) -- (120:2.5) -- (75:2.5) -- (50: 1.5);

\begin{scope}[shift={(130:1)}]

\draw[dotted] (0,-.25) -- (0,1.75);

\end{scope}


\draw[fill] (-.225,1.1) circle (2pt);
\node[below right]  at (-.25,1) {\small $\cP =c$};
\draw[thick, dashed] (-1.6,1.1) -- (-.25,1.1) ;

\draw[fill] (.1,1.33) circle (2pt);
\node[above]  at (.1,1.4) {\small $\cP' = c'$};
\draw[thick, dashed] (-1.6,1.33) -- (.1,1.33) ;

\node at (-.2,-.15) {\small  (b)};

\end{scope}

\begin{scope}[shift={(7,0)}]



\draw  (150: 2) -- (110:1.75) --  (70:1) -- (50: 1.5) ;

\draw (150: 2) -- (120:2.5) -- (90:2.5) -- (50: 1.5);

\begin{scope}[shift={(90:2.5)}]

\draw[dotted] (0,-1.5) -- (0,.25);

\end{scope}

\node[above]  at (-.6,1.63) {\small $c$};

\draw[fill] (.63,1.63) circle (2pt);
\node[above right]  at (.63,1.63) {\small $\cP$};

\draw[fill] (.8,1.37) circle (2pt);
\node[ right]  at (.8,1.5) {\small $\cP'$};

\draw[fill] (-.6,1.63) circle (1.5pt);
\draw[thick, dashed] (-1.48,1.63) -- (.63,1.63);
\draw[thick, dashed] (-.25,1.37) -- (.8,1.37);

\node at (-.2,-.15) {\small  (c.i)};

\end{scope}

\begin{scope}[shift={(10.5,.5)}]

\draw  (165: 1.25) -- (140:.25) -- (70:.5)  -- (50: 1.5) ;

\draw (165: 1.25) -- (140: 2) -- (130:2.1) -- (120:1)  -- (90:1.5) -- (50: 1.5);

\draw[fill] (140:2) circle (1.5pt);
\node[above left]  at (140:2) {\small $c$};
\node[below]  at (120:1.0) {\small $\cP' = a'$};

\node[above] at (-.65,1.2) {\small $\cP$};
\draw[fill] (-.75,1.1) circle (2pt);
\draw[thick, dashed] (-1.47,1.1) -- (-.75,1.1);

\draw[fill] (-.5,.85) circle (2pt);
\draw[thick, dashed] (-1.37,.85) -- (-.5,.85);

\begin{scope}[shift={(130: 2.1)}]

\draw[dotted] (0,.25) -- (0,-1.5);

\end{scope}

\node at (-.2,-.65) {\small  (c.ii)};

\end{scope}

\begin{scope}[shift={(0,-4.5)}]

\draw  (130: 1.5)  -- (100:3) -- (70: 2.5) --  (50: 1.75);

\draw (130: 1.5) -- (90:2) -- (60:1) -- (50: 1.75);

\foreach \ang in {90,  70} {
\draw[dotted] (\ang:1) -- (\ang:3);
}

\draw[thick,dashed]  (-.75,2) -- (.75,2);

\draw[fill] (.75,2) circle  (2pt);
\node[above] at (0.65, 2) {\small $\cP$};

\begin{scope}[shift={(90:2)},rotate=-20]

\draw[thick,dashed] (-.7,0) -- (.7,0);

\draw[fill] (.7,0) circle  (2pt);
\node[below] at (0.85, 0) {\small $\cP'$};

\end{scope}

\node[above left]  at (0,2) {\small $c$};
\draw[fill] (0,2) circle (1.5pt);

\node at (-.2,-.15) {\small  (d.i)};

\end{scope}

\begin{scope}[shift={(3,-4.5)}]

\draw  (130: 1.5)  -- (100:3.25) -- (55: 3) --  (45: 2.25);

\draw (130: 1.5)  -- (90:1.85) -- (50:1.75) -- (45: 2.25);

\foreach \ang in {90,  55} {
\draw[dotted] (\ang:1) -- (\ang:3);
}

\begin{scope}[shift={(90:1.85)}]
\draw[thick,dashed]  (-.8,0) -- (1.4,0);

\node[ left]  at (0,.15) {\small $c$};
\draw[fill] (0,0) circle (1.5pt);

\node[above left]  at (1.3,0) {\small $\cP$};
\draw[fill] (1.3,0) circle (2pt);

\end{scope}

\node[below right]  at (60:1.55) {\small $\cP'=c'$};
\draw[fill] (55:1.725) circle (2pt);

\begin{scope}[shift={(55:1.725)}, rotate=-35]

\draw[thick, dashed] (0,0) -- (-2.1,0);

\end{scope}

\node at (-.2,-.15) {\small  (d.ii.$\alpha$)};

\end{scope}

\begin{scope}[shift={(6.75,-4.5)}]

\draw  (130: 1.6)  -- (110:3) -- (60: 2.75) --  (45: 2);

\draw (130: 1.6) -- (105:2.25) -- (95:1) -- (90:2) -- (55:1) -- (45: 2);

\foreach \ang in {90,  60} {
\draw[dotted] (\ang:1) -- (\ang:3);
}

\begin{scope}[shift={(105:2.25)}]

\node[ left]  at (0,0) {\small $c$};
\draw[fill] (0,0) circle (1.5pt);

\end{scope}

\begin{scope}[shift={(60:2.3)}]
\draw[thick,dashed]  (-1.6,0) -- (0,0);

\node[above left]  at (0,0) {\small $\cP$};
\draw[fill] (0,0) circle (2pt);

\end{scope}

\begin{scope}[shift={(90:2)},rotate=-30]

\draw[thick,dashed] (-1.19,0) -- (1,0);

\node[ right]  at (1,0) {\small $\cP'$};
\draw[fill] (1.0,0) circle (2pt);

\end{scope}

\node[above right]  at (0,2) {\small $c'$};
\draw[fill] (0,2) circle (1.5pt);

\node at (-.2,-.15) {\small  (d.ii.$\beta$)};

\end{scope}

\begin{scope}[shift={(10,-4.5)}]

\draw  (130: 1.6)  -- (105:3) -- (95:2.75) -- (90:1.5) --  (60: 2.75) --  (45: 2);

\draw (130: 1.6) -- (105:2.25) -- (100:.75) --  (55:1) -- (45: 2);

\foreach \ang in {90,  60} {
\draw[dotted] (\ang:1) -- (\ang:3);
}

\begin{scope}[shift={(105:2.25)}]

\node[ left]  at (0,0) {\small $c$};
\draw[fill] (0,0) circle (1.5pt);

\end{scope}

\begin{scope}[shift={(90:1.5)},rotate=-30]

\draw[thick,dashed] (-.5,0) -- (.78,0);

\node[ left]  at (.78,-.1) {\small $\cP'$};
\draw[fill] (.78,0) circle (2pt);

\end{scope}

\begin{scope}[shift={(90:1.5)}]


\node[ right]  at (.9,0) {\small $\cP$};
\draw[fill] (.9,0) circle (2pt);
\draw[thick,dashed]  (-.375,0) -- (.9,0);

\end{scope}

\node[below]  at (90:1.5) {\small $a$};
\draw[fill] (90:1.5) circle (1.5pt);

\node at (-.2,-.15) {\small  (d.iii)};

\end{scope}

\end{tikzpicture}

\caption{Examples of checkpoint and auxiliary vertex updates when $\cP$ moves to $\cP'$. The subfigure labels match the cases listed in Definition \ref{def:scallop-checkpoint}. An updated checkpoint is denoted by $c'$. An updated auxiliary vertex is denoted by $a'$.}

\label{fig:scallop-update-checkpoint}

\end{figure}

We conclude this descriptive section with a final pair of definitions adapted from the monotone case. The first allows us to talk about regions within the polygon, and the second provides a measure for proving that the pursuit ends in finite time.

\begin{definition}
\label{def:sweepable-terms}
The terminology of Definition \ref{def:search-terms} for search in a monotone polygon also applies to search in a scallop polygon (search frontier, guarded frontier, pursuer territory, evader territory, guarded region).
\end{definition}

\begin{definition}
\label{def:scallop-progress}
The searching pursuer $\cP$ \style{makes progress} in the scallop polygon whenever
\begin{enumerate}
\item she updates her checkpoint on the lower boundary from $c_t=(r:\theta)$ to $c_{t'}=(r':\theta')$ where $t' > t$ and $\theta' < \theta$, or
\item she updates her search frame of reference from $(X,Y)$ to $(X',Y')$ where
$\theta(Y') < \theta(Y)$.
 \end{enumerate}

\end{definition}

\subsection{Overview of Scallop Rook Strategy}

In this brief section,  we describe the adaptations needed for rook strategy in a scallop polygon. The algorithms and proofs can be found in Section  \ref{sec:scallop-rook}.
The goal of searching is to establish rook position in the pursuer's current frame of reference. We start with a straight-forward adaptation of Definition \ref{def:rook} to the scallop setting.

\begin{definition}
Let $Q$ be a scallop polygon and let $(X,Y)$ be the pursuer's current frame of reference. 
The pursuer $\cP$ is in \style{rook position} when $|x(\cP) - x(\cE)| \leq 1/2$ and the line segment $\overline{\cP \cE} \subset Q$, so that $\cP$ sees $\cE$ .  
\end{definition}

Once she is in rook position, the pursuer transitions into Rook Mode. She executes rook's strategy while maintaining her current frame of reference (even when  passing over spoke lines).   There are now three ways for $\cE$ to force $\cP$ back into Search Mode.  As in the monotone case, $\cE$ can make an escape move to the right. An evader in a scallop polygon has two additional gambits: 
he can hide behind a feature or he can use a feature to block  the pursuer. 
Figure \ref{fig:sweep-gambit} shows a hiding move and a blocking move for an evader moving to the left after the pursuer has achieved rook position.

\begin{figure}[ht]

\begin{center}

\begin{tikzpicture}[scale=1.5]

\begin{scope}

\fill[color=gray!20] (140:2.85) -- (133:4.25)  -- (-2.18,3.1) -- cycle; 

\draw  (160:3.15) -- (145: 4) -- (140:2.85) -- (133:4.25)  -- (120:3.6)  -- (110:2.75) -- (100:1);

\draw[thick] (-3,1.25) -- (-.3,1.25);

\draw[fill] (-1.9,1.25) circle (1pt);
\draw[fill] (-1.55,1.25) circle (1pt);

\draw[-latex] (-1.6,1.15) -- (-1.85, 1.15);

\draw[fill] (-1.7,2.5) circle (1pt);
\draw[fill] (-2.05,2.5) circle (1pt);
\draw[fill] (-2.4,2.5) circle (1pt);

\draw[-latex] (-2.05,2.5) -- (-2.35, 2.5);
\draw[-latex] (-1.7,2.5) -- (-2.00, 2.5);


\begin{scope}[shift={(140:2.85)}]
	\draw[dashed] (0,-1) -- (0,1.5);
	\draw[fill] (0,0) circle (1 pt);
	\node[right] at (0,0) {\small $v_{\ell}$};
\end{scope}

\node[above right] at (-2.1,1.25) {\small $\cP_{t-1}$};
\node[below right] at (-1.6,1.25) {\small $\cP_{t-2}$};
\node[below right] at (-1.8,2.5) {\small $\cE_{t-2}$};
\node[above right] at (-2.2,2.5) {\small $\cE_{t-1}$};
\node[above] at (-2.4,2.5) {\small $\cE_t$};

\node[below] at (-1.5,.75) {\scriptsize (a)};

\end{scope}


\begin{scope}[shift={(4.5,0)}]

\fill[color=gray!20]  (147:2.275) -- (145: 4)  -- (130:4)  -- (124:3.475) -- cycle;

\draw[dashed] (-1.93,1) -- (-1.93,3);

\draw  (160:3.5) -- (150:4) -- (148:2) -- (145: 4)  -- (130:4)  -- (110:2.75) -- (105:2.25) -- (90:2.35);

\draw (110:1) -- (85:1.5);

\draw[thick] (-1.95,1.25) -- (-.07,1.25);

\draw[fill] (-1.35,1.25) circle (1pt);
\draw[fill] (-1.75,1.25) circle (1pt);

\draw[fill] (-1.93,1.25) circle (1pt);

\node[below right] at (-1.4,1.25) {\small $\cP_{t-2}$};
\node[above right] at (-1.85,1.25) {\small $\cP_{t-1}$};
\node[above left] at (-1.89,1.3)  {\small $b$};

\draw[-latex] (-1.4,1.15) -- (-1.7, 1.15);

\begin{scope}[shift={(.15,0)}]

\draw[fill] (-1.6,2.1) circle (1pt);
\draw[fill] (-2,2.1) circle (1pt);
\draw[fill] (-2.4,2.1) circle (1pt);

\draw[-latex] (-1.6,2.1) -- (-1.95, 2.1);
\draw[-latex] (-2,2.1) -- (-2.35, 2.1);

\node[below right] at (-1.7,2.1) {\small $\cE_{t-2}$};
\node[above right] at (-2.1,2.1) {\small $\cE_{t-1}$};
\node[above] at (-2.35,2.1) {\small $\cE_t$};

\end{scope}

\node[below] at (-1.6,.75) {\scriptsize (b)};

\end{scope}

\end{tikzpicture}

\end{center}

\caption{While $\cP$ is in Rook Mode, she maintains a single frame of reference until $\cE$ makes an escape move, a hiding move, or a blocking move. (a)  $\cE$ makes a hiding move at time $t$ behind hiding vertex $v_{\ell}$. The hiding pocket is shown in gray.
(b) $\cE$ makes a blocking move at time $t$ in an upper pocket. The blocking point $b$ obstructs the pursuer's movement. The blocking pocket is shown in gray. }
\label{fig:sweep-gambit}

\end{figure}

\begin{definition}
\label{def:hiding}
Suppose that $\cP$ and $\cE$ are in  rook position in a scallop polygon. An evader \style{hiding move} disrupts rook position by using a feature on the  boundary to obstruct the visibility of the pursuer.  The vertex $v_{\ell}$  that obstructs the pursuer's view is the \style{hiding vertex}.
The \style{hiding pocket} is the area obscured by the hiding vertex $v_{\ell}$, and bounded by vertical line $Y(v_{\ell})$ in the rook frame of reference.
\end{definition}

\begin{definition}
\label{def:blocking}
Suppose that $\cP$ and $\cE$ are in upper pocket position or in upper chute position.
An evader \style{blocking move} disrupts rook position by using a feature on the upper boundary to obstruct the movement of the pursuer.  The \style{blocking point} $b$ is the point on the upper feature that blocks the pursuer's move. The \style{blocking pocket} is the area bounded by  the upper boundary and the 
vertical line $Y(b)$ in the rook frame of reference
\end{definition}

A hiding move can use an upper feature (when $\cP$ is below $\cE$) or a lower feature (when $\cP$ is above $\cE$). Furthermore, a hiding move can occur in either pocket position or chute position. On the other hand, a blocking move can only occur in upper pocket position and in upper chute position. In these upper positions, the features radiate inwards, so $\cP$ is blocked before $\cE$ see Figure \ref{fig:sweep-gambit}(b). Meanwhile, the reverse holds for a lower position: $\cE$ will be blocked before $\cP$ because the lower features radiate outwards. 

Next, we compare a blocking move with an escape move. After a blocking move, the evader is also confined by the same boundary that obstructs the pursuer. In other words, $\cE$  finds himself in a pocket. Meanwhile,  after an escape move, $\cE$ can move further down the chute and away from the pursuer. In response to either a hiding move or a blocking move, $\cP$ enters Cautious Search Mode, see Algorithm \ref{alg:cautious-search}. This cautious search strategy reverts back to Rook Mode (using the previous rook frame) if $\cE$ ever moves to the right of $\cP$ with respect to that rook frame.

Finally, we note that after an escape move, the return to Search Mode must also be handled with care. To prevent recontamination,  the pursuer must keep track of two frames of reference; for an example, see Algorithm \ref{alg:cautious-search} below.
In the following subsections, we consider Scallop Search Mode, Scallop Rook Mode, and the delicate transitions between the two.

\subsection{Search  Strategy for Scallop Polygons}
\label{sec:scallop-search}

Algorithms \ref{alg:scallop-path} and  \ref{alg:scallop-search} adapt the monotone searching of Algorithms \ref{alg:monotone-path} and \ref{alg:monotone-search} to our new setting.  Our first search path starts at $v_1$ using frame of reference $(T_1,S_1) = (T(v_1), S(v_1))$.


\begin{algorithm}
\caption{ Create Scallop Search Path}
\label{alg:scallop-path}
\begin{algorithmic}[1]
\Require Starting point $p \in Q$ and starting frame of reference $(X,Y)$
\Require Vertices of $Q$ are $v_1, v_2, \ldots , v_n$ where $\theta(v_{i}) > \theta(v_{i+1})$ for $1 \leq i < n$
\Require $\spokeline{i}$ is the spoke line through vertex $v_i$ and center $C$, $1 \leq i \leq n$
\Require $\transline{i}$ is the transverse line through $v_i$ (perpendicular to $\spokeline{i}$),  $1 \leq i \leq n$
\State Set checkpoint $c$ and auxiliary vertex $a$ as per Definition
\ref{def:scallop-checkpoint}
\State Note: if $p=v_1$, then checkpoint $c \leftarrow v_1$ and $a \leftarrow  \emptyset$.

\While {$p \neq v_n$}
\State Move $p$ along horizontal axis $X$  until reaching either  $\partial Q$ or a spoke line
\If{$p$ is on lower boundary $\Pi_L$}
\While {$X_+(p) = \{ p \}$}
\State Note: moving right in frame $(X,Y)$ still requires  moving along $\Pi_L$
\State Traverse   along $\Pi_L$ until reaching next vertex $v_k \in \Pi_L$
\State Update checkpoint $c \leftarrow v_k$
\EndWhile
\ElsIf  {$p$ is on upper boundary $\Pi_U$}
\State Traverse along $\Pi_U$   until reaching the next spoke line $S_k$.
\Else
\State Note: $p$ in on spoke line  $S_k$, possibly at vertex $v_k$
\State Update frame $(X,Y)  \leftarrow (\transline{k}, \spokeline{k})$
\While {$p$ does not guard $c$ in frame $(X,Y)$ as per Definition \ref{def:scallop-guarded}}
\State Move down spoke line $S_k$
\EndWhile
\State Update checkpoint $c$ and auxiliary vertex $a$ as per Definition
\ref{def:scallop-checkpoint}
\EndIf
\EndWhile

\end{algorithmic}
\end{algorithm}

\begin{definition}
A \style{search path} in scallop polygon $Q$ is a path $\sp$ constructed by the Create Scallop Search Path algorithm.
\end{definition}

\begin{algorithm} 
\caption{Scallop Search Strategy}
\label{alg:scallop-search}
\begin{algorithmic}[1]
\Require Vertices of $Q$ are $v_1, v_2, \ldots , v_n$ where $\theta(v_{i}) > \theta(v_{i+1})$ for $1 \leq i < n$
\Require $\Pi$ is the search path (from Create Scallop Search Path)
\Require Initial frame of reference $(X,Y)$  is the current frame of reference of $\cP_{t-1}$ 
\Require $x(\cP_{t-1}) \leq x(\cE_{t-1})$, but $\cE_{t-1}$ might be invisible to $\cP_{t-1}$ 

\While { {\sc not} $\big(\cE_{t-1}$ is visible and  $0 \leq x(\cE_{t-1})-x(\cP_{t-1}) \leq  1/2 \big)$ }

\State Set checkpoint $c$ and auxiliary vertex $a$ as per Definition
\ref{def:scallop-checkpoint}

\While{$\cP_{t-1}$ is not on a spoke line}
\State Evader moves from $\cE_{t-1}$ to $\cE_t$
\State  $\cP_t \leftarrow$ one monotone search move (Algorithm \ref{alg:monotone-search}) with respect to $(X,Y)$, but stopping if she reaches the next spoke line
\State $t \leftarrow t+1$
\EndWhile

\State Evader moves from $\cE_{t-1}$ to $\cE_t$

\State $\spokeline{k} \leftarrow$  the spoke line that $\cP$ has encountered
\State $(X',Y') \leftarrow$  the frame of reference for $\spokeline{k}$
\If {$\cE_t$ is in the second quadrant of the $(X',Y')$ frame}
\State $(X,Y)  \leftarrow (X'',Y'')$ where $Y'' = \overline{\cP \cE}$ is the new vertical direction 
\State Exit to start a new Scallop Rook's Strategy (Algorithm \ref{alg:upperRook}) with respect to  new reference frame
\EndIf

\State $p \leftarrow$  the lowest point on $\Pi \cap \spokeline{k}$

\While {$\cP$ has not reached $p$}
\If {$\cE_t$ is in rook position with respect to frame $(X,Y)$ or frame $(X',Y')$}
\State Exit to start a new Scallop Rook's Strategy using the appropriate reference frame
\EndIf
\State $\cP_t \leftarrow$ step downwards towards $p$
\State $t \leftarrow t+1$
\If {$\cP_{t-1} \neq p$}
\State Evader moves from $\cE_{t-1}$ to $\cE_t$
\EndIf
\EndWhile

\State Note: $\cP_{t-1} = p$ and $\cP$ can now adopt her new reference frame
\State $(X,Y) \leftarrow (X', Y')$

\EndWhile

\end{algorithmic}
\end{algorithm}

~\begin{lemma} \label{lemma:sweep-search}
If  $\cP$ follows the scallop search strategy then $\cE$ cannot step into the pursuer territory without being caught. Furthermore, $\cP$ either achieves rook position or captures $\cE$ in finite time.
\end{lemma}

\begin{figure}

\begin{center}

\begin{tikzpicture}[scale=1.1]

\begin{scope}

\draw[gray] (-1.25,1.5) -- (1,1.5) -- (1,5.2) -- (-1.25,5.2) --cycle;

\draw[dashed] (0:0) -- (150:4);

\draw[dashed] (0:0) -- (135:5.25);

\draw[dashed] (0:0) -- (120:4);

\draw[dashed] (0:0) -- (90:5.25);

\draw[dashed] (0:0) -- (45:3.5);

\draw[dashed] (0:0) -- (30:3);

\draw (150: 1) -- (150:3) -- (135:5) -- (120:2) -- (90:5) -- (45:3) -- (30:2) -- cycle;

\begin{scope}[shift={(150:3)}, rotate=60]
	\coordinate (a1) at (0,0);
	\coordinate (a2) at (.818,0);
\end{scope}

\begin{scope}[shift={(150:3)}, rotate=45]
	\coordinate (a3) at (.759,0);
	\coordinate (a4) at (1.125, 0);
\end{scope}

\coordinate (a5) at(120:2);

\begin{scope}[shift={(120:2)}, rotate=30]
	\coordinate (a5) at (0,0);
	\coordinate (a6) at (1.163, 0);
\end{scope}

\begin{scope}[shift={(120:2)}]
	\coordinate (a7) at (1,0);
	\coordinate (a8) at (2.725,0);
\end{scope}

\coordinate (a9) at (45:1.15);
\coordinate (a10) at (30:2);

\draw[thick] (a1) -- (a2) -- (a3) -- (a4) -- (a5) -- (a6) -- (a7) -- (a8) -- (a9) -- (a10);

\node[right] at (120:1.9) {$v_k$};
\node[right] at (90:5) {$v_{k+1}$};

\node[right] at (90:3) {$S_{k+1}$};

\draw[fill] (a6) circle (1.5pt);
\draw[fill] (a7) circle (1.5pt);

\node[above left] at (a6) {$Z_1$};
\node[above right] at (a7) {$Z_2$};

\end{scope}

\node at (0,-.5) {\scriptsize (a)};

\end{tikzpicture}
\qquad \qquad
\begin{tikzpicture}[scale=1.5]

\draw[gray] (-1.25,1.5) -- (1,1.5) -- (1,5.2) -- (-1.25,5.2) --cycle;

\begin{scope}

\clip(-1.25,1.5) rectangle (1,5.2);

\draw[dashed] (0:0) -- (150:4);

\draw[dashed] (0:0) -- (135:5.25);

\draw[dashed] (0:0) -- (120:4);

\draw[dashed] (0:0) -- (90:5.25);

\draw[dashed] (0:0) -- (45:4);

\draw[dashed] (0:0) -- (30:4);

\draw (150: 1) -- (150:3) -- (135:5) -- (120:2) -- (90:5) -- (45:3) -- (30:2) -- cycle;

\begin{scope}[shift={(150:3)}, rotate=60]
	\coordinate (a1) at (0,0);
	\coordinate (a2) at (.818,0);
\end{scope}

\begin{scope}[shift={(150:3)}, rotate=45]
	\coordinate (a3) at (.759,0);
	\coordinate (a4) at (1.125, 0);
\end{scope}

\coordinate (a5) at(120:2);

\begin{scope}[shift={(120:2)}, rotate=30]
	\coordinate (a5) at (0,0);
	\coordinate (a6) at (1.163, 0);
\end{scope}

\begin{scope}[shift={(120:2)}]
	\coordinate (a7) at (1,0);
	\coordinate (a8) at (2.725,0);
\end{scope}

\coordinate (a9) at (45:1.15);
\coordinate (a10) at (30:2);

\draw[thick] (a1) -- (a2) -- (a3) -- (a4) -- (a5) -- (a6) -- (a7) -- (a8) -- (a9) -- (a10);

\node[right] at (120:1.9) {$v_k$};

\draw[fill] (a6) circle (1.5pt);
\draw[fill] (a7) circle (1.5pt);

\begin{scope}[shift=(a6), rotate=30]

\draw[dashed] (.15,-2.2) -- (.15,2.75);
\draw[dashed] (0,-2.2) -- (0,2.75);
\draw[dashed] (-.15,-2.2) -- (-.15,2.75);

\end{scope}

\begin{scope}[shift=(a6)]

\draw[dashed] (.15,-2) -- (.15,3);
\draw[dashed] (-.15,-2) -- (-.15,3);

\end{scope}

\end{scope}

\node at (0,5.35) {\small $Y_{k+1}$};
\node at (.5,4.95) {\small $Y_{k+1}''$};
\node at (-.4,4.95) {\small $Y_{k+1}'$};

\node at (-1.5,4.9)  {\small $Y_{k}''$};
\node at (-1.5,4.5) {\small $Y_{k}$};
\node at (-1.5,4.1) {\small $Y_{k}'$};


\fill[gray!20] (-.475,3.45) -- (-.15,4.5) -- (-.15,2.88);

\fill[gray!20] (1,1.74) -- (.52,1.74) -- (.17,2.35) -- (.17,4.76) -- (1,3.63);

\node at (-.3, 3.6) {\small $\mathcal{A}$};
\node at (.6, 3) {\small $\mathcal{B}$};

\node at (0,1) {\scriptsize (b)};

\end{tikzpicture}

\end{center}

\caption{Changing frame during the search. (a) The pursuer $\cP$ starts at $Z_1$ and wants  to move to $Z_2$.
(b) $\cP$ starts with vertical axis $Y_k$ and hopes to switch to vertical axis $Y_{k+1}$.  If $\cE \in \mathcal{A}$, then $\cP$ can pick a vertical axis between $\spokevertline{k}$ and $\spokevertline{k+1}$ and immediately establish rook position. Otherwise,  $\cE \in \mathcal{B}$, and $\cP$ can travel from $z_1$ to $z_2$ and safely transition from vertical axis $\spokevertline{k}$ to vertical axis $\spokevertline{k+1}$ (or establish rook position along the way if $\cE$ makes an unwise move). }

\label{fig:scallop-switch}

\end{figure}

\begin{proof}
As the pursuer travels between spoke lines $\spokeline{k}$ and $\spokeline{k+1}$, the proof  is unchanged from that of Lemma \ref{lemma:search}. Indeed, the polygon region between the spoke lines $\spokeline{k}$ and $\spokeline{k+1}$ is monotone with respect to the horizontal axis $T_k$ (perpendicular to spoke $\spokeline{k}$). Furthermore, the guarded vertex remains guarded while $\cP$ moves through this region. We therefore need only prove that $\cE$ cannot step into the pursuer territory during the change of frame that occurs at spoke lines.

Suppose that $\cP$ reaches spoke line $\spokeline{k+1}$  at point $Z_1$: an example is shown in Figure \ref{fig:scallop-switch}.  With respect to the polar coordinates centered at $C$, the search path takes $\cP$ from her current position $Z_1 = (\rho_1: \theta)$ to $Z_2 = (\rho_2: \theta)$ where $\rho_1 > \rho_2$. The point $Z_2$ is chosen so that $v_k$ is guarded by the  horizontal line through $Z_2$ with respect to vertical $\spokeline{k+1}$.  The simplest case of protecting $v_k$ is when this horizontal line intersects $v_k$. However a feature on the lower boundary might obstruct the line through $v_k$; in this case, $Z_2$ is chosen so that the new horizontal line intersects the vertex on the top of this obstructing feature.

Before proceeding from $Z_1$ to $Z_2$, the pursuer $\cP$ first checks whether a simple change of frame will establish rook position. Let $\spokevertline{k+1} = \spokeline{k+1}$ and let $\spokevertline{k+1}'$ and $\spokevertline{k+1}''$ denote the lines located $\pm 1/2$ from $\spokeline{k+1}$. Let $\spokevertline{k}$ denote the line through $Z_1$ parallel to spoke line $\spokeline{k}$, and define $\spokevertline{k}',  \spokevertline{k}''$ similarly. Since $\cP$ has not established rook position, $\cE$ must be to the right of $\spokevertline{k}''$. However, $\cE$ could be to the left of $\spokevertline{k+1}'$ (in region $\mathcal{A}$ of Figure \ref{fig:scallop-switch}(b)).
In this case, $\cE$ must be visible to $\cP$ since there are no vertices between $x_k$ and $x_{k+1}$. Therefore $\cP$ can change her vertical axis to be the line through $\cP$ and $\cE$. This guards the checkpoint vertex $c$ (since the new horizontal intersects the upper boundary above $c$), and establishes rook position. 

Otherwise, $\cE$ is located to the right of both $\spokevertline{k}''$ and $\spokevertline{k+1}''$ (in region $\cB$ of Figure \ref{fig:scallop-switch}(b)). 
The pursuer moves down the segment $\overline{Z_1 Z_2}$, keeping track of \emph{two} coordinate systems: she uses $\spokevertline{k+1}$ above her position and $\spokevertline{k}$ below her position. If the evader steps within $1/2$ distance to either of these lines, then $\cP$ responds by establishing rook position with respect to the appropriate frame. Once again, this guards the checkpoint vertex $c$. Finally, if $\cP$ reaches $Z_2$ without establishing rook position, she commits to the $Y_{k+1}$-frame (which is also the $\spokeline{k+1}$ frame). If her current search frontier intersects any lower boundary vertices to the left, then $\cP$ updates her checkpoint or auxiliary point, if necessary. (In Figure \ref{fig:scallop-switch}, we continue to guard auxiliary point $v_k$.) Having adopted her new frame of reference, she continues her search.
\end{proof}

%

\subsection{Rook Strategy for Scallop Polygons}
\label{sec:scallop-rook}


The  rook positions in a scallop polygon are the same as those in Figure \ref{fig:rook-position}. 
Algorithm \ref{alg:upperRook} lists the rook's strategy for upper pockets and upper chutes. Lower pockets and lower chutes are handled similarly: the rook strategy for those cases  is listed in Algorithm \ref{alg:lowerRook} below.

Rook's strategy in a scallop polygon is more challenging than the monotone case because the evader can now use hiding moves and blocking moves, see Definitions
\ref{def:hiding} and \ref{def:blocking}.
These moves interrupt the pursuer's rook position, so we need a recovery phase to handle these setbacks. This recovery uses the modified search algorithm listed in Algorithm \ref{alg:cautious-search} which  remains aware of the rook frame of reference, even as the search frame rotates. After recovery in a pocket, $\cP$ will return to using rook moves, perhaps with an updated frame of reference. 
Recovery in an upper (lower) chute must deal with the case where $\cE$ makes an escape move  using the lower (upper) boundary to block the pursuer's rightward movement. In response to such an escape move, the pursuer  will either re-establish a rook position, or initiate a new search phase. Finally, we note once again that the mode transitions are  also more difficult  due to the shifting frame of reference.

\begin{algorithm}
\caption{ Upper Rook's Strategy (for Upper Pocket and Upper Chute)}
\label{alg:upperRook}
\begin{algorithmic}[1]
\Require Initial frame of reference $(X,Y)$  is the current frame of reference of $\cP$ 
\Require $|x(\cP_{t-1})  -x(\cE_{t-1}) | \leq  1/2$ and $y(\cP_{t-1}) < y(\cE_{t-1})$ and $\cP_{t-1}$ sees $\cE_{t-1}$
\State Note: the rook frontier $\frontier(\cP_{t-1})$ lies below $\cE_{t-1}$
\State Note: the leftmost endpoint of $Q_{\cE_{t-1}}$ is in $\Pi_U \cap \frontier(\cP_{t-1})$
\While {$\cE$ is not captured}
\State Evader moves from $\cE_{t-1}$ to $\cE_t$
\If{$\cE$ made a blocking move}
\State Let $Z \in \Pi_U \cap \frontier(\cP_{t-1})$ be the blocking point
\State $\cP_t \leftarrow Z$, which is reachable in one step
\State Enter Cautious Scallop Search Strategy 
\ElsIf{$\cE$ made a hiding move}
\State Let $v_{\ell} \in \Pi_U$ be the hiding vertex. 
\State Note: assume  $x(\cE_t) < x(\cP_{t-1})$; handling $x(\cE_t) > x(\cP_{t-1})$ is similar
\While {$\cE$ is hidden by $v_{\ell}$,}
\State Note:
 $x(\cE_t) < x(v_{\ell})$ and ($x(v_{\ell}) < x(\cP_{t-1})$ or $y(\cP_{t-1}) < y(v_{\ell})$)
\If{$x(v_{\ell}) < x(\cP_{t-1})$}
\State $\cP_t \leftarrow $ move left towards $x(v_{\ell})$ and up with  remaining  budget
\State Note: it takes at most two rounds to reach $x(\cP) = x(v_{\ell})$
\Else
\State $\cP_t \leftarrow$ move upwards towards $v_{\ell}$
\EndIf
\State $t \leftarrow t+1$
\State Evader moves from $\cE_{t-1}$ to $\cE_t$
\EndWhile
\If {$\cE_t$ is still in the hiding pocket and $\cP_{t-1} = v_{\ell}$}
\State Enter Cautious Scallop Search Strategy
\EndIf
\State Note: $\cE_t$ stepped to the right of $\cP_{t-1}$ and is no longer hidden
\State $\cP_t \leftarrow$ one monotone rook move (Algorithm \ref{alg:monotone-rook})
\ElsIf  {$\cE$ made an escape move}
\State Let $Z \in \Pi_L$ be the blocking point
\State $\cP_t \leftarrow Z$, which is reachable in one step
\If {$\theta(\cE_t) \geq \theta(\cP_t)$}
\State Update vertical direction to be $\overline{\cP_t \cE_t}$
\State Start new Upper Rook's Strategy with this updated frame of reference
\Else
\State Exit  (in order to start a new Scallop Search Strategy from $\cP_t=Z$)
\EndIf
\Else
\State $\cP_t \leftarrow $ one monotone rook move (using Monotone Rook's Strategy)
\EndIf
\State $t \leftarrow t+1$
\EndWhile

\end{algorithmic}
\end{algorithm}

\begin{algorithm}
\caption{ Cautious Scallop Search Strategy}
\label{alg:cautious-search}
\begin{algorithmic}[1]
\Require $\cP$ starts at a point $Z$ on the boundary
\Require Previous upper/lower rook frame of reference is $(X_0, Y_0)$ 
\Require Previous checkpoint is $c$ 
\Require $\cE$ is to the left of $\cP$ in frame $(X_0, Y_0)$;  a right move is handled analogously
\State Create a new search path from $Z$ that protects $c$ (using Create Scallop Search Path)
\While {$\cE$ remains to the left of $\cP$ with respect to previous rook frame $(X_0, Y_0)$}
\State $\cP$ makes one Scallop Search Strategy move
\If {$\cP$ is in rook position with respect to the search frame}
\State Enter a new Upper/Lower Rook's Strategy
\EndIf
\EndWhile
\State Note: $\cE$ has stepped to the right of $\cP$ with respect to $(X_0,Y_0)$
\State Exit (back to the previous Upper/Lower Rook's Strategy)
\end{algorithmic}
\end{algorithm}


\begin{figure}[ht]

\begin{center}

\begin{tikzpicture}[scale=1.5]

\fill[color=gray!20] (140:2.85) -- (133:4.25)  -- (-2.18,3.1) -- cycle; 

\draw  (165:3.25) -- (160:3.15) -- (145: 4) -- (140:2.85) -- (133:4.25)  -- (120:3.6)  -- (110:2.75) -- (90:3) -- (80:1) -- (79:3);

\draw[thick] (-3,1.25) -- (.15,1.25);

\draw[fill] (-1.95,1.25) circle (1pt);
\draw[fill] (-2.18,1.55) circle (1pt);

\draw[fill] (-2.05,2.3) circle (1pt);
\draw[fill] (-2.4,2.4) circle (1pt);

\draw[-latex] (-2.05,2.3) -- (-2.37, 2.39);

\draw[-latex] (-1.95,1.25) -- (-2.17,1.53);

\draw[dashed] (90:.75) -- (90:3.25);

\begin{scope}[shift={(140:2.85)}]
	\draw[dashed] (0,-1) -- (0,1.5);
	\draw[fill] (0,0) circle (1 pt);
	\node[right] at (0,0) {\small $v_{\ell}$};
\end{scope}

\node[above right] at (-2,1.25) {\small $\cP_{t-1}$};
\node[left] at (-2.18,1.55) {\small $\cP_t$};
\node[above] at (-1.8,2.3) {\small $\cE_{t-1}$};
\node[above] at (-2.4,2.375) {\small $\cE_t$};

\draw[fill] (90:3) circle (1pt);
\node[above right] at (90:3) {\small $v_k$};
\node[left] at (90:2) {\small $\spokeline{k}$};

\node[below] at (-1.2,.75) {\scriptsize (a)};

\begin{scope}[shift={(4.5,0)}]

\fill[color=gray!20] (129.5:2.83) -- (129.5:4.02) -- (120:3.6) -- cycle;

\draw  (165:3.25) -- (160:3.15) -- (145: 4) -- (140:2.85) -- (133:4.25)  -- (120:3.6)  -- (105:2.5) -- (90:3) -- (80:1) -- (79:3);

\draw[thick] (-3,1.25) -- (.15,1.25);

\draw[fill] (-2,2.8) circle (1pt);
\draw[fill] (-1.8,2.19) circle (1pt);

\node[right] at (-1.8,2.19) {\small $\cP$};
\node[above] at (-2,2.8) {\small $\cE$};

\draw[dashed] (-1.8,.75) -- (-1.8,3.35);

\draw[dashed] (140:1.2) -- (140:4.5);
\draw[dashed] (133:1.2) -- (133:4.5);
\draw[dashed] (129.5:1.2) -- (129.5:4.5);
\draw[dashed] (120:1.2) -- (120:3.85);

\draw[dashed] (90:.75) -- (90:3.25);

\begin{scope}[shift={(140:2.85)}]
	\draw[fill] (0,0) circle (1 pt);
	\node[below left] at (0,0) {\small $v_{\ell}$};
\end{scope}

\begin{scope}[shift={(140:2.85)}, rotate=140]
	\draw[very thick] (0,0) -- (0,-.35);
\end{scope}

\begin{scope}[shift={(133:2.83)}, rotate=133]
	\draw[very thick] (.05,0) -- (0,0) -- (0,-.66);
	\draw[very thick, -latex]  (0,0) -- (0,-.6);
\end{scope}

\begin{scope}[shift={(120:2.68)}, rotate=120]
	\draw[very thick] (.24,0) -- (0,0) -- (0,-.5);
\end{scope}

\draw[fill] (90:3) circle (1pt);
\node[above right] at (90:3) {\small $v_k$};
\node[left] at (90:2) {\small $\spokeline{k}$};

\node[below] at (-1.2,.75) {\scriptsize (b)};

\end{scope}

\end{tikzpicture}

\end{center}

\caption{(a) $\cE$ makes a hiding move into the hiding pocket (shaded) above $v_{\ell}$. $\cP$ responds by moving below $v_{\ell}$ and then  upwards  until reaching $v_{\ell}$.  If $\cE$ moves to the right of $\cP$ then she re-enters Rook Mode. (b) Once $\cP = v_{\ell}$, $\cP$ starts a new cautious search phase. The evader is confined to the gray region:  if $\cE$ moves to the right of $\cP$ in the original rook frame, then $\cP$ reverts  to that Rook Mode.}
\label{fig:hidepocket}

\end{figure}

We will consider upper pocket position first.

\begin{lemma} 
\label{lemma:upperpocket}
If $\cP$ and $\cE$ are in upper  pocket position, then upper rook's strategy captures the evader in $O(\Area(Q_{\cE})+ n(Q_{\cE}))$ turns, where $n(Q_{\cE})$ is the number of vertices in the pocket $Q_{\cE}$.
\end{lemma}

\begin{proof}
As in the proof of Lemma \ref{lemma:endgame}, we show that $\cP$ consistently shrinks the size of the pocket. Furthermore, the pursuer trajectory never visits the same point twice, capturing the evader in $O(\Area(Q_{\cE})+ n(Q_{\cE}))$ turns.

Let $\spokeline{k}= S(v_k)$ be the spoke line perpendicular to the rook frontier.  Suppose that $\cP$ and $\cE$ are in upper pocket position and that $\cP$ is executing the  strategy of Algorithm \ref{alg:upperRook}. If $\cE$ never makes a hiding move or a blocking move, then $\cP$ captures $\cE$ in $O(\Area(Q_{\cE}))$ turns by the same argument used for  Lemma \ref{lemma:endgame}.
So let's suppose that $\cE$ hides behind a feature in the upper pocket. 
We consider the case when $x(\cP_{t-1}) - 1/2 \leq x(\cE_{t-1}) \leq x(\cP_{t-1})$ just before this hiding move; the case $x(\cP_{t-1}) \leq x(\cE_{t-1}) \leq x(\cP_{t-1}) + 1/2$ is argued similarly, swapping the roles of left and right.

Let $v_{\ell}$ be the hiding vertex that obscures the evader. First, we consider the case  $\ell \leq k$, see Figure \ref{fig:hidepocket}(a). The evader cannot have stepped to the right of $\cP$  by an argument similar to the proof of Lemma \ref{lemma:hide-left} (with left and right reversed). 
The pursuer's high-level strategy is to reach the hiding vertex and then traverse a new search path.  However, if $\cE$ ever moves to the right of $\cP$ in the original rook frame, then she  reverts to her previous rook's strategy.
First,  $\cP$ tries to achieve $x(\cP) = x(v_{\ell})$, which takes one or two rounds. In the turn that she achieves $x(\cP) = x(v_{\ell})$, she uses any additional movement budget to move upwards towards the hiding vertex to regain visibility into the hiding pocket, see Figure \ref{fig:hidepocket}(a). 

After that, the pursuer moves upwards until reaching the hiding vertex $v_{\ell}$. In the meantime, if $\cE$ moves to the right of the hiding vertex $v_{\ell}$, then  $x(\cP) \leq x(\cE) \leq x(\cP)+1$ and $\cE$ must be visible to $\cP$ due to the scallop nature of the polygon. The pursuer responds by  returning Rook Mode using vertical $\spokeline{k}$.  In particular, she uses a leftward offset, which means that she will make at least $\sqrt{3}/{2}$ vertical process. This vertical progress is crucial, since it prevents the evader from repeatedly hiding and reappearing  above $v_{\ell}$ ad infinitum. Indeed, once $\cP$ reaches $v_{\ell}$, the evader can no longer hide behind this vertex.

Suppose that $\cE$ remains in the hiding pocket until $\cP$ reaches  $v_{\ell}$. The pursuer switches to the cautious scallop search strategy for Algorithm \ref{alg:cautious-search}, see  Figure \ref{fig:hidepocket}(b). She   draws a new search path starting at $v_{\ell}$, and searches along this path while  keeping track of two frames of reference: her old rook frame (with vertical $\spokeline{k}$) and the current search frame.  If $\cE$ ever steps to the right of $\cP$ with respect to the old rook frame, then $\cP$ reverts back to the original Rook Mode. Note that in this case, $\cP$ has made vertical progress. Meanwhile, if the evader never steps to the right of $\cP$ in the old rook frame, then Search Mode will terminate with rook position in an updated frame of reference. Furthermore, $\cE$ will be trapped in an even smaller pocket than before. Finally, note that $\cP$ must attain rook position before reaching spoke line $\spokeline{k}$: if $\cE$ is to the right of $\spokeline{k}$ then $\cE$ is to the right of $\cP$ in the original rook frame, so $\cP$ will have already reverted to the previous Rook Mode.

\begin{figure}[ht]

\begin{center}

\begin{tabular}{cc}

\begin{tikzpicture}[scale=1.5, yscale=1,xscale=-1]

\fill[color=gray!20] (135:3.1) -- (133:4.25)  -- (-2.18,3.1) -- cycle; 

\draw   (160:3.15) -- (155:3.25) -- (135:3.1) -- (133:4.25)  -- (120:3.6)  -- (105:2.75) -- (80:1) ;

\draw[thick] (-2.95,1.25) -- (.05,1.25);

\draw[fill] (-2.05,1.25) circle (1pt);
\draw[fill] (-2.18,1.6) circle (1pt);
\draw[-latex] (-2.05,1.25) -- (-2.17,1.58);

\draw[fill] (-1.95,2.6) circle (1pt);
\draw[fill] (-2.35,2.6) circle (1pt);

\draw[-latex] (-1.95,2.6) -- (-2.32, 2.6);

\draw[dashed] (-.71,.95) -- (-.71,3);
\draw[fill] (-.71,2.65) circle (.75pt);
\node[above left] at (-.71,2.65) {\small $v_k$};
\node[right] at (-.71,2) {\small $\spokeline{k}$};

\begin{scope}[shift={(135:3.1)}]
	\draw[dashed] (0,-1.25) -- (0,1.15);
	\draw[fill] (0,0) circle (.75 pt);
	\node[right] at (0,0) {\small $v_{\ell}$};
\end{scope}

\node[below] at (-1.9,1.25) {\small $\cP_{t-1}$};
\node[right] at (-2.18,1.55) {\small $\cP_t$};
\node[above] at (-1.9,2.6) {\small $\cE_{t-1}$};
\node[above] at (-2.35,2.6) {\small $\cE_t$};

\node[below] at (-1.2,.75) {\scriptsize (a)};

\end{tikzpicture}

&

\begin{tikzpicture}[scale=3, yscale=1,xscale=-1]

\fill[color=gray!20] (-2.2,1.75) -- (-2.65,2)  -- (-2.2,2.04) -- cycle;

\begin{scope}[shift={(-2.4,1.25)}]
\draw[gray] (0,0) -- (64:.9);
\draw[gray] (0,0) -- (68:.87);
\end{scope}

\draw (-3,1.2) -- (-2.2,1.75) -- (-2.65,2) --   (-1.5, 2.1) -- (-1.08, 1.2);

\draw[thick] (-2.925,1.25) -- (-1.1,1.25);

\draw[fill] (-2.4,1.25) circle (.5pt);
\draw[fill] (-2.2,1.5) circle (.5pt);


\draw[fill] (-2.225,1.6) circle (.5pt);
\draw[fill] (-2.18,1.9) circle (.5pt);



\draw[dashed] (-1.5, 2.25) -- (-1.5,1);
\draw[fill] (-1.5, 2.1) circle (.5pt);
\node[above left] at (-1.5, 2.1) {\small $v_k$};
\node[right] at (-1.5, 1.7) {\small $\spokeline{k}$};

\begin{scope}[shift={(-2.2,1.75)}]
	\draw[dashed] (0,-.7) -- (0,.5);
	\draw[fill] (0,0) circle (.5 pt);
	\node[right] at (0,0) {\small $v_{\ell}$};
\end{scope}

\node[below right] at (-2.25,1.25) {\small $\cP_{t-1}$};
\node[left] at (-2.25,1.45) {\small $\cP_t$};
\node[left] at (-2.225,1.65) {\small $\cE_{t-1}$};
\node[right] at (-2.18,1.95) {\small $\cE_t$};

\node[above] at (-2.3,1.23) {\small $\phi$};

\node[below] at (-2,1) {\scriptsize (b)};

\node[below] at (-.8,1) {\scriptsize \phantom{(b)}};

\end{tikzpicture}

\end{tabular}

\caption{Hiding moves behind a vertex $v_{\ell}$ where $\ell > k$ when $x(\cE_{t-1}) < x(\cP_{t-1})$. The pursuer can reach the hiding vertex in one move. (a) $\cE$ steps to the right into the hiding pocket. (b) $\cE$ steps to the left when hiding which means that  the  angle  from horizontal to $\overline{\cP_{t-1} \cE_{t-1}}$ satisfies $\phi > \pi/3$. }

\label{fig:hide-right}

\end{center}

\end{figure}

A hiding move behind a vertex $v_{\ell}$  where $\ell > k$ is argued similarly.  Recall that we have $x(\cP_{t-1}) -1/2 \leq x(\cE_{t-1}) \leq x(\cP_{t-1}).$ 
We claim that $\cP$ can achieve $x(\cP_{t}) = x(v_{\ell})$ in her first responding move. The argument then proceeds as above, swapping left and right when $x(\cP_{t-1}) < x(v_{\ell})$.

We now prove that $x(\cP_{t}) = x(v_{\ell})$; the two cases are shown  in  Figure \ref{fig:hide-right}. 
First, suppose that the evader steps to the right $x(\cE_{t-1}) < x(\cE_{t}) \leq x(\cP_{t-1}) + 1$. The hiding vertex $v_{\ell}$ must  satisfy $| x(v_{\ell}) - x(\cP_{t-1})| < 1$, so $\cP$ can attain $x(\cP_{t}) = x(v_{\ell})$ and move upward with any remaining movement budget.  
Second, suppose that the evader does not move to the right, so that $x(\cE_{t-1}) \geq x(\cE_{t})$. Recall that $|\cE_{t-1} - \cP_{t-1}| > 1$ (since the evader was not captured at time $t-1$) and that $| x(\cE_{t-1}) - x(\cP_{t-1})| \leq 1/2$. Therefore $|y(\cE_{t-1}) - y(\cP_{t-1})| \geq \sqrt{3}/2$  and the angle $\phi$ of $\overline{\cP_{t-1} \cE_{t-1}}$ with the horizontal is at least $\pi/3$. In order to be obscured by vertex $v_{\ell}$, the evader must increase this angle $\phi$, so  the evader's leftward movement is less than $1/2$. This means that the hiding vertex $v_{\ell}$ satisfies $x(\cP_{t-1}) - x(v_{\ell}) < 1$.  Once again,  $\cP$ can attain $x(\cP_{t}) = x(v_{\ell})$ and move upward with any remaining movement budget.

The pursuer deals with a blocking move in much the same way. Without loss of generality, suppose that $\cP$ is blocked to the left, see Figure \ref{fig:blockpocket}.
 She moves to the blocking point $Z$ and then draws a new search path from $Z$.
She then enters a modified search strategy, keeping track of both the old rook frame and her current search frame. During the search, if $\cE$  steps to the right of $\cP$ with respect to the old rook frame, then $\cP$ immediately re-enters Rook Mode with a leftward offset (making $\sqrt{3}/{2}$ vertical process). Otherwise, $\cP$ establishes Rook Mode with respect to a new frame of reference, and the evader region has been reduced.


\begin{figure}[t]

\begin{center}

\begin{tikzpicture}[scale=1.4]

\draw  (150:4) -- (148:2) -- (145: 4) -- (135:3.4) -- (130:4)  -- (110:2.75) -- (90:3) -- (80:1) -- (79:3);

\draw[thick] (-1.95,1.25) -- (.15,1.25);

\draw[fill] (-1.75,1.25) circle (1pt);
\draw[fill] (-1.925,1.25) circle (1pt);
\draw[fill] (-2,2.1) circle (1pt);
\draw[fill] (-2.4,2.2) circle (1pt);

\draw[-latex] (-2,2.1) -- (-2.35, 2.19);

\draw[-latex] (-1.75,1.35) -- (-1.925,1.35);

\draw[dashed] (90:.5) -- (90:3.25);

\node[above right] at (-1.8,1.25) {\small $\cP_{t-1}$};
\node[below left] at (-1.925,1.25) {\small $\cP_t = Z$};
\node[right] at (-2,2.1) {\small $\cE_{t-1}$};
\node[below] at (-2.45,2.2) {\small $\cE_t$};

\draw[fill] (90:3) circle (1pt);
\node[above right] at (90:3) {\small $v_k$};
\node[left] at (90:2.2) {\small $\spokeline{k}$};

\node[below] at (-1,.5) {\scriptsize (a)};

\begin{scope}[shift={(4.5,0)}]

\fill[color=gray!20]    (130:4)  -- (-1.44,2.74) -- (-1.44,1.75) ;

\draw (150:4) -- (148:2) -- (145: 4) -- (135:3.4) -- (130:4)  -- (110:2.75) -- (90:3) -- (80:1) -- (79:3);

\draw[dashed] (148:1) -- (148:4.25);
\draw[dashed] (145:1) -- (145:4.25);
\draw[dashed] (135:1) -- (135:4.25);
\draw[dashed] (130:1) -- (130:4.25);
\draw[dashed] (110:1) -- (110:3);
\draw[dashed] (90:1) -- (90:3.25);

\draw[thick]  (-1.95,1.25) -- (.15,1.25);

\draw[fill] (130:2.25) circle (1pt);

\draw[fill] (123:3.3) circle (1pt);

\draw[dashed] (-1.44,1.2) -- (-1.44,2.85);

\node at (125:2.3) {\small $\cP$};

\node[below] at (121:3.3) {\small $\cE$};

\draw[fill] (-1.925,1.25) circle (1pt);
\node[below left] at (-1.925,1.25) {\small $Z$};

\draw[fill] (90:3) circle (1pt);
\node[above right] at (90:3) {\small $v_k$};
\node[left] at (90:2.2) {\small $\spokeline{k}$};

\begin{scope}[shift={(-1.93,1.25)}, rotate=148]
	\draw[very thick] (0,0) -- (0,-.09);
\end{scope}

\begin{scope}[shift={(145:2.3)}, rotate=145]
	\draw[very thick] (0,0) -- (0,-.4);
\end{scope}

\begin{scope}[shift={(135:2.25)}, rotate=135]
	\draw[very thick]  (.1,0) -- (0,0) -- (0,-.205);
\end{scope}

\begin{scope}[shift={(130:2)}, rotate=130]
	\draw[very thick] (.255,0) -- (0,0) -- (0,-.735);
	\draw[very thick,-latex] (0,0) -- (0,-.6);	
\end{scope}

\begin{scope}[shift={(110:1.83)}, rotate=110]
	\draw[very thick] (.3, 0) -- (0,0) -- (0,-.66);
\end{scope}

\begin{scope}[shift={(90:1.25)}, rotate=90]
	\draw[very thick] (.7025,0) -- (0,0) -- (0,-.15);
\end{scope}

\node[below] at (-1,.5) {\scriptsize (b)};
\end{scope}

\end{tikzpicture}

\end{center}

\caption{(a) The evader makes a blocking move from $\cE_{t-1}$ to $\cE_t$ in an upper pocket. The upper boundary prevents $\cP$ from re-establishing rook position, so she moves to blocking point $\cP_t$. (b) $\cP$ uses the cautious scallop search strategy, while $\cE$ is in the shaded region. If $\cE$ steps to the right of $\cP$ in the original rook frame, then $\cP$ will revert back to the old rook phase. $\cP$ will achieve rook position before reaching spoke line $S_k$.}

\label{fig:blockpocket}

\end{figure}


During this pursuit, $\cP$ never visits the same point twice. She makes methodical progress in each rook phase. Every time that the evader hides, the pursuer eventually enters another rook phase, and the evader can only block or hide $n(Q_{\cE})$ times. Therefore the  capture time is $O(\Area(Q_{\cE})+ n(Q_{\cE}))$.
\end{proof}

When the evader territory is an an upper (lower) chute, then $\cE$ can make an escape move, using the lower (upper) chain to obstruct $\cP$, see Figure \ref{fig:scallopUpperRecovery}. As in the monotone case, this is the only evader gambit in an upper (lower) chute that involves the lower (upper) boundary.

\begin{lemma}
\label{lemma:upperChute}
Suppose that $\cP$ and $\cE$ are in upper chute position. Using the upper rook's strategy, $\cP$ will either capture $\cE$, or make progress on her search path.
\end{lemma}

\begin{proof}
If $\cE$ never makes an escape move, then $\cP$ will establish pocket position and capture the evader by Lemma \ref{lemma:upperpocket}. Suppose that $\cE$ makes a rightward escape move, meaning that the lower boundary blocks the pursuer's  rook move. Her recovery phase consists of a single round. The pursuer steps rightward to the blocking point $Z$, see Figure \ref{fig:scallopUpperRecovery}.
If spoke line $\sweepline{\cE}$ is between the vertical axis and spoke line $\sweepline{\cP}$, then $\cP$ chooses the line between $\cP$ and $\cE$ as her new vertical axis and enters Rook Mode. This updates the rook frontier and advances her checkpoint.
Otherwise, $\cP$ starts a new search phase from her current location $z \in \Pi_L$. This guards her previous checkpoint: $\theta(\cP) > \theta(\cE)$ means that $\cE$ is to the right of $\cP$ in the new search frame. She updates her checkpoint to $Z$, and begins searching. 
\end{proof}

\def\scalloplowerblockenv
{
\draw  (-.75, 2) -- (-.29, 1.75) -- (0,3.25) -- (55:3.35) -- (50:3.75);

\draw  (123:2.25) -- (90:1) -- (55:2.5) -- (35: 2) -- (25:2.75);

\draw[dotted] (99:.9) -- (99:3.5);
\draw[dotted] (90:.9) -- (90:3.5);
\draw[dotted] (55:1) -- (55:4);
\draw[dotted] (35:1.5) -- (35:3);
}

\def\scalloplowerblockpath
{
\begin{scope}[shift={(90:2.75)}, rotate=0]
\draw (0,0)-- (1.6,0);
\path (1.6,0)  coordinate (P1);
\end{scope}

\begin{scope}[shift={(P1)}, rotate=-30]
\draw (0,0)-- (0,-.8) -- (.3, -.8);
\path (.3, -.8)  coordinate (P2);
\end{scope}

\begin{scope}[shift={(P2)}, rotate=-37]
\draw (0,0)-- (0,-.2) -- (.7, -.2) ;
\path (.7, -.2)  coordinate (P3);
\end{scope}

\draw (P3) -- (35:2.05) -- (25:3);

}

\begin{figure}[h]

\begin{center}

\begin{tikzpicture}[scale=1.0]

\begin{scope}

\fill[gray!20]  (-.29, 1.75)  -- (0,3.25) -- (55:3.35) -- (50:3.75) -- (25:2.75) -- (35: 2) -- (55:2.5)  -- (1, 1.75);

\scalloplowerblockenv

\draw[thick] (-1.05, 1.75)  -- (1, 1.75) ;
\draw[fill] (.8, 1.75) circle (1.5 pt);
\draw[fill] (1, 1.75) circle (1.5 pt);
\draw[fill]  (1.1, 2) circle (1.5 pt);
\draw[fill]  (-.29, 1.75)    circle (1.5 pt);
\draw[fill] (1.3, 2.5) circle (1.5 pt);
\draw[fill] (1.55, 2.25) circle (1.5 pt);

\draw[-latex] (1.1,2) -- (1.32, 2.45);
\draw[-latex] (1.1,2) -- (1.53, 2.23);

\node[above] at (.6, 1.75)  {\scriptsize $\cP_{t-1}$};
\node[below] at (1.1, 1.75)  {\scriptsize $Z$};
\node[above] at (.8,2) {\scriptsize $\cE_{t-1}$};
\node[above left] at (1.3, 2.5) {\scriptsize $\cE_{t}$};
\node[right] at (1.55, 2.25) {\scriptsize $\cE_{t}'$};
\node[below] at  (-.29, 1.75)    {\scriptsize $a$};

\node at (.5, .5) {\scriptsize (a) };

\end{scope}

\begin{scope}[shift={(4.5,0)}]

\begin{scope}[shift={(1,1.75)}, rotate=-22]

\path (-1.3,0) coordinate (QQ1);
\path (90:1.2) coordinate (QQ2);

\end{scope}

\fill[gray!20]    (QQ1) -- (0,3.25) -- (55:3.35) -- (50:3.75) -- (25:2.75) -- (35: 2) -- (55:2.5)  -- (1, 1.75);

\draw[dashed] (60:1) -- (60:4);

\draw (1,1.75) -- (QQ1);

\draw (1,1.75) -- (QQ2);

\scalloplowerblockenv

\draw[thick] (-1.05, 1.75)  -- (1, 1.75) ;
\draw[fill] (1, 1.75) circle (1.5 pt);
\draw[fill]  (-.29, 1.75)    circle (1.5 pt);
\draw[fill] (1.3, 2.5) circle (1.5 pt);

\node[above] at (.9, 1.75)  {\scriptsize $\cP_{t}$};
\node[above left] at (1.3, 2.5) {\scriptsize $\cE_{t}$};
\node[below] at  (-.29, 1.75)    {\scriptsize $a$};

\node at (.5, .5) {\scriptsize (b)};

\end{scope}

\begin{scope}[shift={(9,0)}]

\scalloplowerblockenv

\draw[very thick] (1,1.75) -- (55:2.5);

\begin{scope}[shift={(55:2.5)}, rotate=-35]
\path (.93,0) coordinate (QQ3);
\draw[very thick] (0,0) -- (QQ3);
\end{scope}

\begin{scope}[shift={(QQ3)}, rotate=-55]
\path (.25,-.28) coordinate (QQ4);
\draw[very thick]  (0,0) -- (0,-.3) -- (QQ4);
\end{scope}

\draw[very thick]  (QQ4) -- (25:2.75);

\draw[thick] (-1.05, 1.75)  -- (1, 1.75) ;
\draw[fill] (1, 1.75) circle (1.5 pt);
\draw[fill]  (-.29, 1.75)    circle (1.5 pt);
\draw[fill] (1.55, 2.25) circle (1.5 pt);

\node[above] at (.9, 1.75)  {\scriptsize $\cP_{t}$};
\node[right] at (1.55, 2.25) {\scriptsize $\cE_{t}'$};
\node[below] at  (-.29, 1.75)    {\scriptsize $a$};

\node at (2.1,1.85) {\scriptsize $\Pi$};

\draw[dashed] (60:1) -- (60:4);

\node at (.5, .5) {\scriptsize (c) };

\end{scope}

\end{tikzpicture}

\end{center}

\caption{Recovery from an escape move in an upper chute. (a) The evader makes one of two rightward blocking moves. The pursuer responds by moving to the blocking point $Z$. (b) If $\cE$ is to the left of the radial line through $Z$, then $\cP$ immediately enters Rook Mode using the line between $\cP$ and $\cE$ as the new vertical. (c) If $\cE$ is to the right of this line, then $\cP$ starts a new search phase. }

\label{fig:scallopUpperRecovery}

\end{figure}

Pursuit in a lower pocket is analogous to pursuit in an upper pocket. 
Algorithm  \ref{alg:lowerRook} lists the pursuer's rook strategy for lower pockets and lower chutes. 
We make two observations. First, the outward radial nature of lower features means that the evader cannot make a blocking move. Second, handling an escape move from a lower chute  is slightly different than for an upper chute. Suppose that $\cE$ makes a rightward escape move, meaning that the upper boundary blocks the pursuer's responding rook move, see Figure \ref{fig:scallopLowerRecovery}(a). In response, the  pursuer steps rightward to the blocking point $Z$. 
She then (cautiously) adopts the frame of reference $(T(Z),S(Z))$.
If she is  above her most recent search path $\Pi$ with respect to this new frame, then she cannot start a new search path from her current location, see Figure \ref{fig:scallopLowerRecovery}(b). Instead,
she travels down the sweep line $\sweepline{Z}$ until reaching the unique point $Z' \in \Pi \cap S(Z)$. Note that this point guards $c$ with respect to the frame $(T(v_{\ell}), S(v_{\ell}))$ where $\theta(v_{\ell}) \geq \theta(Z) > \theta(v_{\ell+1})$.  While moving along $\sweepline{Z}$, she keeps track of two frames of reference, just as if she is making a usual frame transition during Search Mode, as shown in Figure \ref{fig:scallopLowerRecovery}(c). This protects the last checkpoint of her previous search phase.

\begin{algorithm}
\caption{ Lower Rook's Strategy (for Lower Pocket and Lower Chute)}
\label{alg:lowerRook}
\begin{algorithmic}[1]
\Require Initial frame of reference $(X,Y)$  is the current frame of reference of $\cP$ 
\Require $|x(\cP_{t-1})  -x(\cE_{t-1}) | \leq  1/2$ and $y(\cP_{t-1}) > y(\cE_{t-1})$ and $\cP_{t-1}$ sees $\cE_{t-1}$
\State Note: $\frontier(\cP_{t-1})$ is above $\cE_{t-1}$ and
 leftmost endpoint of $Q_{\cE_{t-1}}$ is in $\Pi_L \cap \frontier(\cP_{t-1})$
\While {$\cE$ is not captured}
\State Evader moves from $\cE_{t-1}$ to $\cE_t$

\If{$\cE$ made a hiding move}
\State Let $v_{\ell} \in \Pi_U$ be the hiding vertex. 
\State Note: assume  $x(\cE_t) < x(\cP_{t-1})$; handling $x(\cE_t) > x(\cP_{t-1})$ is similar
\While {$\cE$ is hidden by $v_{\ell}$}
 \State Note: $x(\cE_t) < x(v_{\ell})$ and ($x(v_{\ell}) < x(\cP_{t-1})$ or $y(\cP_{t-1}) > y(v_{\ell})$)

\If{$x(v_{\ell}) < x(\cP_{t-1})$}
\State $\cP_t \leftarrow $ move left towards $x(v_{\ell})$ and down with  remaining  budget
\State Note: it takes at most two rounds to reach $x(\cP) = x(v_{\ell})$
\Else
\State $\cP_t \leftarrow$ move downwards towards $v_{\ell}$
\EndIf
\State $t \leftarrow t+1$
\State Evader moves from $\cE_{t-1}$ to $\cE_t$
\EndWhile
\If {$\cE_t$ is still in the hiding pocket and $\cP_{t-1} = v_{\ell}$}
\State Enter Cautious Scallop Search Strategy
\EndIf
\State Note: $\cE_t$ stepped to the right of $\cP_{t-1}$ and is no longer hidden
\State $\cP_t \leftarrow$ one monotone rook move (Algorithm \ref{alg:monotone-rook})
\ElsIf  {$\cE$ made an escape move}
\State Let $Z \in \Pi_U$ be the blocking point where $\theta(v_{\ell}) \leq \theta(Z) < \theta(v_{\ell+1})$
\State $\cP_t \leftarrow Z$, which is reachable in one step
\State Set $(X', Y')  = (T(Z), S(Z))$ to be our provisional frame of reference.
\State Set $Z' \leftarrow$ unique point in $\Pi \cap S(Z)$. 
\While {$y'(\cP) > y'(Z')$ 
and
 not $(\cP$ in $(X,Y)$-rook position with $y'(\cE) \leq y'(\cP)$) 
 and
  not $(\cP$ in $(X',Y')$-rook position with $y'(\cE) \geq y'(\cP)$) 
} 
\State $P_t \leftarrow$ move down $Y' = S(Z)$ towards $Z'$
\State $t \leftarrow t+1$
\State Evader moves from $\cE_{t-1}$ to $\cE_t$
\EndWhile
\If {$\cP$ is  in $(X,Y)$-rook position with $y'(\cE) \leq y'(\cP)$}
\State $\cP_t \leftarrow$ one monotone rook move (Algorithm \ref{alg:monotone-rook})
\ElsIf {$\cP$ is  in $(X',Y')$-rook position with $y'(\cE) \geq y'(\cP)$}
\State Exit (in order to start a new Upper Rook's Strategy)
\Else
\State Exit (in order to start a new Scallop Search Strategy from $\cP_t$)
\EndIf
\EndIf

\EndWhile

\end{algorithmic}
\end{algorithm}


\def\scallopblockenv
{
\draw (-1, 1.75) -- (-.5,3.25) -- (60:3.35) -- (53:3.1) -- (50:4.25);

\draw (-.5, 1) -- (90:2.75) -- (35: 2) -- (25:3);

\draw[fill] (60:3.35) circle (1pt); 
\node[above] at (62:3.35) {\small $v_{\ell}$};

\draw[dotted] (99:.9) -- (99:3.5);
\draw[dotted] (90:.9) -- (90:3.5);
\draw[dotted] (60:1) -- (60:4);
\draw[dotted] (53:1.1) -- (53:4);
\draw[dotted] (35:1.5) -- (35:3.3);
}

\def\scallopblockpath
{
\begin{scope}[shift={(90:2.75)}, rotate=0]
\draw[very thick] (0,0)-- (1.6,0);
\path (1.6,0)  coordinate (P1);
\end{scope}

\begin{scope}[shift={(P1)}, rotate=-30]
\draw[very thick] (0,0)-- (0,-.8) -- (.3, -.8);
\path (.3, -.8)  coordinate (P2);
\end{scope}

\begin{scope}[shift={(P2)}, rotate=-37]
\draw[very thick] (0,0)-- (0,-.2) -- (.7, -.2) ;
\path (.7, -.2)  coordinate (P3);
\end{scope}

\draw[very thick] (P3) -- (35:2.05) -- (25:3);

\node at (.8,2.55) {\scriptsize $\Pi$};

}

\begin{figure}[ht]

\begin{center}

\begin{tikzpicture}[scale=1.0]

\begin{scope}

\scallopblockenv

\draw[thick] (.1,2.65) -- (1.8,2.65);

\draw[fill] (1.6,2.65) circle (1.5 pt);
\draw[fill]  (1.7,2) circle (1.5 pt);
\draw[fill]  (2.2,2.2) circle (1.5 pt);
\draw[fill] (0, 2.75) circle (1.5 pt);

\draw[fill] (1.8,2.65) circle (1.5 pt);

\draw[-latex] (1.7,2) -- (2.15, 2.175);

\node[below] at (1.55,2.65)  {\scriptsize $\cP_{t-1}$};
\node[below] at (1.8,2) {\scriptsize $\cE_{t-1}$};
\node[above] at (2.4,2.2) {\scriptsize $\cE_{t}$};
\node[left] at (0,2.75) {\scriptsize $c$};
\node[above] at (1.9,2.65)  {\scriptsize $Z$};

\node at (.5, .5) {\scriptsize (a)};

\end{scope}

\begin{scope}[shift={(4.5,0)}]

\scallopblockenv

\draw[fill] (1.8,2.65) circle (1.5 pt);

\draw[fill]  (2.2,2.2) circle (1.5 pt);
\draw[fill] (0, 2.75) circle (1.5 pt);

\draw[fill] (56:2.4) circle (1.5 pt);

\draw[dashed] (56:1.2) -- (56:3.85);
\draw[-latex] (56:3.1) -- (56:2.48);

\node[above] at (1.95,2.7)  {\scriptsize $\cP_{t}$};
\node[above] at (2.4,2.2) {\scriptsize $\cE_{t}$};
\node[left] at (0,2.75) {\scriptsize $c$};
\node[left] at (58:2.4) {\scriptsize $Z'$};

\node at (.5, .5) {\scriptsize (b)};

\scallopblockpath

\end{scope}

\begin{scope}[shift={(9,0)}]

\begin{scope}

\clip (-1, 1.75) -- (-.5,3.25) -- (60:3.35) -- (53:3.1) -- (50:4.25) -- (25:3) -- (35: 2) -- (90:2.75) -- (-.5, 1);

\begin{scope}[shift={(56:2.7)}]

	\draw[fill=gray!20] (0,0) -- (0,-3) -- (-34:2);	
	\draw[fill=gray!50]   (-34:2) -- (0,0) -- (56:3);

\end{scope}
\end{scope}

\scallopblockenv

\draw[fill] (56:2.7) circle (1.5 pt);
\draw[fill] (56:2.4) circle (1.5 pt);

\draw[fill] (0, 2.75) circle (1.5 pt);

\draw[dashed] (56:1.2) -- (56:3.85);
\draw (56:3.2) -- (56:2.41);

\node[above] at (51:2.75)   {\scriptsize $\cP$};
\node[left] at (0,2.75) {\scriptsize $c$};
\node[left] at (58:2.4) {\scriptsize $Z'$};

\node at (1.85,1.7) {\scriptsize $\cA$};
\node at (2.3,2.2) {\scriptsize $\cB$};

\node at (.5, .5) {\scriptsize (c)};

\scallopblockpath

\scallopblockpath

\end{scope}

\end{tikzpicture}

\end{center}

\caption{Recovery from  an escape move in a lower chute. (a) The evader makes an escape  move. (b) The pursuer moves to the blocking point $Z$ and then along the sweep line $\sweepline{Z}$ heading to point $Z'$ on the previous search path $\Pi$.  (c) While traversing $\sweepline{Z}$, she uses the previous rook frame in  region $\cA$, and she uses her current frame (where $\sweepline{Z}$ is vertical) in $\cB$.}
\label{fig:scallopLowerRecovery}

\end{figure}

\begin{lemma} 
\label{lemma:lowerpocket}
If $\cP$ and $\cE$ are in lower  pocket position, then lower rook's strategy captures the evader in $O(\Area(Q_{\cE})+ n(Q_{\cE}))$ turns, where $n(Q_{\cE})$ is the number of vertices in the pocket $Q_{\cE}$.
\end{lemma}

\begin{proof}
The proof is analogous to the proof of Lemma \ref{lemma:upperpocket} for upper pocket position. We provide a summary here,  refering the reader to that proof for more detail. If $\cE$ never makes a hiding move then $\cP$ captures $\cE$ in $O(\Area(Q_{\cE}))$ turns by Lemma \ref{lemma:endgame}. If $\cE$ hides behind a feature on the lower boundary, then the argument of Lemma \ref{lemma:upperpocket} still applies, swapping the roles of ``up'' and ``down.'' Meanwhile, the pursuer cannot use a lower feature to block the pursuer: this feature would also block the evader. (So there are only two cases for a lower pocket instead of three.) The pursuer never visits the same point twice, and she makes methodical progress in every rook phase. The evader can hide $O(n(Q_{\cE}))$ times, each of which is followed by another rook phase. Therefore the capture time  is $(\Area(Q_{\cE}) + n(Q_{\cE}))$. 
\end{proof}

\begin{lemma}
\label{lemma:lowerChute}
Suppose that $\cP$ and $\cE$ are in lower chute position. Using the lower rook's strategy, $\cP$ will either capture $\cE$, or make progress on her search path.
\end{lemma}

\begin{proof}
If $\cE$ never makes an escape move, then $\cP$ will establish pocket position and capture the evader by Lemma \ref{lemma:lowerpocket}. Suppose that $\cE$ makes a rightward escape move, meaning that the upper boundary blocks the pursuer's rook move. Let $(X,Y)$ be the rook frame of reference and let $c$ be the current checkpoint. Let $Z \in \Pi_U$ be the blocking point where $\theta(v_{\ell}) \geq \theta(Z) > \theta(v_{\ell+1})$. In response to the blocking move,  the pursuer steps rightward to  $Z$, see Figure \ref{fig:scallopLowerRecovery}(a). Let $(X',Y') = (T(Z), S(Z))$ be the transverse line and sweep line through point $Z$. Let $Z'$ be the unique intersection point $\Pi \cap Y'$ where $\Pi$ is the previous search path. 

Next, the pursuer cautiously moves down $Y'$ towards point $Z'$. At each step, she checks whether $\cP$ and $\cE$ are in rook position with respect to two different frames. If $\cP$ is in $(X,Y)$-rook position and $y'(\cE) \leq y'(\cP)$ (region $\mathcal{A}$ in Figure \ref{fig:scallopLowerRecovery}(c)),
then $\cP$ returns to Rook Mode with frame $(X,Y)$.  If $\cP$ is in $(X',Y')$-rook position and $y'(\cE) \geq y'(\cP)$
(region $\mathcal{B}$ in Figure \ref{fig:scallopLowerRecovery}(c)) then $\cP$  starts a new Rook Mode using frame $(X',Y')$. Otherwise, she reaches the point $Z'$, adopts frame $(X_{\ell}, Y_{\ell}) = (T(v_{\ell}, S(v_{\ell}))$ and enters a new Search Mode starting from $Z'$. In each of these three cases, the pursuer continues to guard checkpoint $c$ in her updated frame of reference. 

Finally, we claim that if the pursuer exits Lower Rook Mode into Search Mode (due to an escape move), then she will make progress (by updating her checkpoint or her frame of reference) as described in Definition \ref{def:scallop-progress}. 
Let  $\phi \leq \theta(c)$ be the angle corresponding to the rook frame $(X,Y)$.
There are two cases to consider: $\phi > \theta(v_{\ell})$ and $\phi = \theta(v_{\ell})$.

\begin{figure}[ht]

\begin{center}

\begin{tikzpicture}[scale=1.25]

\begin{scope}[rotate=30]

\draw (110:1) -- (90:2.25)  -- (20:.5) --  (10:2);

\draw (110:2.5) -- (85:3) -- (75:2) -- (65:3) -- (30:2) -- (25:3);

\draw[dotted] (90:.25) -- (90:3);
\draw[dotted] (85:.25) -- (85:3);
\draw[dotted] (75:.25) -- (75:3);
\draw[dotted] (65:.25) -- (65:3);
\draw[dotted] (30:.25) -- (30:3);
\draw[dotted] (20:.25) -- (20:3);

\begin{scope}[shift={(90:2.25)}]


\coordinate (a1) at (0,0);
\coordinate (a2) at (.2,0);

\end{scope}

\begin{scope}[shift={(90:2.25)}, rotate=-10]


\coordinate (a3) at (.2,0) ;
\coordinate (a4) at (.47,0);

\end{scope}

\begin{scope}[shift={(75:2)}, rotate=-15]


\coordinate (a5) at (0,0) ;
\coordinate (a6) at (.35,0) ;

\end{scope}

\begin{scope}[shift={(75:2)}, rotate=-25]


\coordinate (a7) at (0.35,0) ;
\coordinate (a8) at (1.225,0);

\coordinate (b1label) at (.875,0);

\coordinate (e1) at (1.05, -.4);

\coordinate (b1) at (1, 0);
\coordinate (b2) at (1, -.1);
\coordinate (b3) at (-.2, -.1);
\coordinate (b4) at (-.2, -.2);
\coordinate (b5) at (1.4, -.2);

\coordinate (c1) at (1.24, -.7);
\coordinate (c2) at (.5, -.7);
\coordinate (c3) at (.5, -.8);
\coordinate (c4) at (1.5, -.8);
\coordinate (c5) at (1.5, -.9);
\coordinate (c6) at (.7, -.9);
\coordinate (c7) at (.7, -1.0);
\coordinate (c8) at (1.71, -1.0);

\end{scope}

\begin{scope}[shift={(90:2.25)}, rotate=-60]


\end{scope}

\begin{scope}[shift={(30:2)}, rotate=-60]

\coordinate (a9) at (0,0) ;
\coordinate (a10) at (0,-.88);
\coordinate (a11) at (.2,-.88);

\end{scope}

\begin{scope}[shift={(90:2.25)}, rotate=-70]


\coordinate (a12) at (2.11,0) ;
\coordinate (a13) at (2.18,0);

\end{scope}

\coordinate (a14) at (10:2);

\draw[thick] (a1) -- (a2)  -- (a3) -- (a4)  -- (a5) --  (a6) -- (a7) -- (a8) --  (a9)  -- (a10) -- (a11) -- (a12) -- (a13) -- (a14);



\draw[fill] (a1) circle (.75pt);
\node[left] at (a1)  {\small $c$};

\draw[fill] (a5) circle (.75pt);
\node[ left] at (a5)  {\small $a$};


\draw[fill] (b1) circle (.75pt);
\node[ below] at (b1label)  {\small $\cP$};

\draw[fill] (e1) circle (.75pt);
\node[ below] at (e1)  {\small $\cE$};

\draw[fill] (65:3) circle (.5 pt);
\node[right] at (65:3) {\small $v_{\ell}$};

\node[below] at (0,0) {\small (a)};

\node[above right] at (a6) {$\Pi$};

\end{scope}

\begin{scope}[shift={(5,0)}, rotate=30]

\draw (110:1) -- (90:2.25)  -- (20:.5) --  (10:2);

\draw (110:2.5) -- (85:3) -- (75:2) -- (65:3) -- (30:2) -- (25:3);

\draw[dotted] (90:.25) -- (90:3);
\draw[dotted] (85:.25) -- (85:3);
\draw[dotted] (75:.25) -- (75:3);
\draw[dotted] (65:.25) -- (65:3);
\draw[dotted] (30:.25) -- (30:3);
\draw[dotted] (20:.25) -- (20:3);

\begin{scope}[shift={(90:2.25)}]


\coordinate (a1) at (0,0);
\coordinate (a2) at (.2,0);

\end{scope}

\begin{scope}[shift={(90:2.25)}, rotate=-10]


\coordinate (a3) at (.2,0) ;
\coordinate (a4) at (.47,0);

\end{scope}

\begin{scope}[shift={(75:2)}, rotate=-15]


\coordinate (a5) at (0,0) ;
\coordinate (a6) at (.35,0) ;

\end{scope}

\begin{scope}[shift={(75:2)}, rotate=-25]


\coordinate (a7) at (0.35,0) ;
\coordinate (a8) at (1.225,0);

\coordinate (b1) at (1, 0);
\coordinate (b2) at (1, -.1);
\coordinate (b3) at (-.2, -.1);
\coordinate (b4) at (-.2, -.2);
\coordinate (b5) at (1.4, -.2);

\coordinate (p11) at (1.15, -.2);
\coordinate (p11label) at (1.1, -.2);

\coordinate (e11) at (1.1, -.7);
\coordinate (e12) at (1.4, -.8);
\coordinate (e12a) at (1.375, -.8);
\coordinate (e12label) at (1.5, -.8);

\coordinate (cc1) at (1.3, -.6);
\coordinate (cc2) at (1.8, -.6);
\coordinate (cc3) at (1.8, -.7);

\coordinate (c1) at (1.24, -.7);
\coordinate (c2) at (.5, -.7);
\coordinate (c3) at (.5, -.8);
\coordinate (c4) at (1.3, -.8);
\coordinate (c5) at (1.3, -.9);
\coordinate (c6) at (.4, -.9);
\coordinate (c7) at (.4, -1.0);
\coordinate (c8) at (.7, -1.0);

\coordinate (e2) at (.8, -1.3);

\end{scope}

\begin{scope}[shift={(90:2.25)}, rotate=-60]


\end{scope}

\begin{scope}[shift={(30:2)}, rotate=-60]

\coordinate (a9) at (0,0) ;
\coordinate (a10) at (0,-.88);
\coordinate (a11) at (.2,-.88);

\end{scope}

\begin{scope}[shift={(90:2.25)}, rotate=-70]


\coordinate (a12) at (2.11,0) ;
\coordinate (a13) at (2.18,0);

\end{scope}

\coordinate (a14) at (10:2);

\draw[thick] (a1) -- (a2)  -- (a3) -- (a4)  -- (a5) --  (a6) -- (a7) -- (b1);

\draw[thick]  (b5) --  (a9)  -- (cc1) -- (a10) -- (a11) -- (a12) -- (a13) -- (a14);

\draw[thick, gray] (b1) -- (b2) -- (b3) -- (b4)  -- (b5);


\draw[fill] (a1) circle (.57pt);
\node[left] at (a1)  {\small $c$};

\draw[fill] (a5) circle (.75pt);
\node[ left] at (a5)  {\small $a$};

\draw[fill] (b5) circle (.75pt);
\node[above] at (b5)  {\small $Z$};

\node[below] at (0,0) {\small (b)};

\draw[fill] (p11) circle (.75pt);
\draw[fill] (e11) circle (.75pt);
\draw[fill] (e12) circle (.75pt);

\node[below] at (p11label)  {\small $\cP$};
\node[ left] at (e11)  {\small $\cE$};
\node[above] at (e12label)  {\small $\cE'$};

\draw [-latex] (e11) -- (e12a);

\draw[fill] (65:3) circle (.5 pt);
\node[right] at (65:3) {\small $v_{\ell}$};



\end{scope}

\end{tikzpicture}

\caption{Transition back to search after a lower chute when $\phi = \theta(v_{\ell})$. (a) The pursuer enters Rook Mode while in a lower search on  search path $\Pi$ (black trajectory) while guarding checkpoint $c$ and auxiliary vertex $a$. (b) While in $\cP$ is in Rook Mode (grey trajectory), $\cE$ makes an escape move using blocking point $Z$. The search angle has not been updated, so $Z \in \Pi$, allowing  $\cP$ to immediately re-enter Search Mode. }

\label{fig:lower-chute-search}

\end{center}

\end{figure}

Suppose that $\phi > \theta(v_{\ell})$, see Figure \ref{fig:scallopLowerRecovery}(b). After reaching blocking point  $Z \in  \Pi_U$, the pursuer moves down $S(Z)$ until reaching the search path $\Pi$. At that point, she adopts the frame $(X_{\ell}, Y_{\ell})$ which updates her search angle, and makes progress according to Definition \ref{def:scallop-progress}. (Note that if $\cE$ creates Rook Position while $\cP$ descends to $\Pi$, then $\cP$ does not exit Rook Mode.)

Now suppose that $\phi = \theta(v_{\ell})$. This means that there are no polygon vertices with angles between   $\theta(v_{\ell})$ and $\theta(Z)$. As a consequence, the point $Z$ is actually  on the search path $\Pi$ (so there is no need to descend along spoke line $S(Z)$), see Figure \ref{fig:lower-chute-search}(b). Furthermore, the search path $\Pi$ will continue to coincide with $\Pi_U$ until it reaches vertex $v_{\ell+1}$.  
We claim that the point $Z$ cannot be a blocking point more than once. If the evader immediately moves to recreate rook position, then the pursuer  moves from $Z$ to advance the rook frontier in rook frame $(X,Y)$ by at least $7/22$ as per Algorithm \ref{alg:monotone-rook}. Otherwise, $\cP$ moves one unit along $\Pi$ (which coincides with the upper boundary) or reaches $\theta(v_{\ell+1})$, whichever is closer. 

In a finite number of moves, 
$\cP$ will be at or below the height of the point on $\Pi_U$ at angle $\theta(v_{\ell+1})$. If she is in Search Mode, then she makes progress by updating her angle to $\theta(v_{\ell+1})$. If she is in Rook Mode, then her next blocking  point will have angle less that $\theta(v_{\ell+1})$ and so she will update her angle when she transitions back to Search Mode. 
\end{proof}

\subsection{Catching the Evader}

The pursuit algorithm for a scallop polygon is identical to Algorithm \ref{alg:monotonePursuit}, where we use our scallop algorithms in place of the monotone algorithms. We now prove that this algorithm succeeds in capturing the evader in  a scallop polygon.

\bigskip

\begin{proofof}{Theorem \ref{thm:scallop}}
The proof is similar to that of Theorem \ref{thm:monotone}, and relies on Lemmas \ref{lemma:upperpocket}, \ref{lemma:upperChute},  \ref{lemma:lowerpocket} and \ref{lemma:lowerChute}.
The pursuer alternates between Search Mode and Rook Mode (which includes the Cautious Search Mode). The transition from Rook Mode to Search Mode only occurs after an escape move. Every boundary edge or vertex can only be involved in one such escape transition, so we switch modes $O(n)$ times. The search path is at most twice the diameter (by the Pythagorean theorem the search path from $v_1$ to $v_n$ is at most twice the length of the shortest path between these vertices). During  Search Mode, the pursuer never visits the same point twice. Indeed, when $\cE$ makes an escape move,  $\cP$ starts a new search path from the blocking point. 
The pursuer pauses when passing the spoke lines, accounting for at most $n$ steps.
In summary, accounting for all the Search Modes, the pursuer spends $O(\diam(Q) + n)$ turns  searching.

Next, we argue that $\cP$ visits each point at most twice.  Consider the entire pursuer trajectory from the start until the capture of the evader, and partition the pursuer's path according to each search phase and rook phase. It is clear that none of the search paths intersect one another.
While making rook moves in a fixed rook frame of reference, the pursuer path does not intersect itself. However, it is possible for the pursuer's path to intersect itself when considering the rook phases before and after an escape move. For an escape move in an upper chute, no point is revisited.  When $\cP$ is blocked by the lower boundary, she immediately decides whether she can attain rook position via updating her frame of reference. If she can, then this new frame of reference lies above the previous rook frame. Otherwise, she enters Search Mode, progressing into unexplored territory.

An  escape move for a lower chute may lead to $\cP$ revisiting some points, but no point will be visited more than twice. Suppose that $\cE$ makes an escape move by blocking $\cP$ at point $Z \in \Pi_U$. The pursuer responds by moving to $Z$ and then along spoke line $\sweepline{Z}$ with the goal of reaching the search path $\Pi$ and transitioning her frame of reference.  Consider the turn in which she next attains rook position (either with respect to the old rook frame or with respect to her intended new frame). If $\cE$ lies below $\Pi$, then she enters Rook Mode in previously unexplored territory, so no point is revisited. If $\cP$ attains Rook Mode with $\cE$ above $\Pi$, then the evader territory may include points to the left of $\sweepline{z}$ and above $\Pi$,  see Figure \ref{fig:scallopLowerRecovery}(c). Therefore this region can be revisited during the new rook phase.  However, we are now in upper pocket position or upper chute position, which means that $\cP$ will not revisit any points a third time: the next escape move by the evader cannot involve these points in any way. 

In a given rook phase, the pursuer makes consistent vertical progress, and she never visits the same point twice. Therefore she spends a total of $O(\Area(Q))$ turns in Rook Mode. Eventually, the evader is trapped in a pocket region, where he is caught by the pursuer.
\end{proofof}

%

\section{Conclusion}
\label{sec:conclusion}

We have considered a  pursuit-evasion game in monotone and scallop polygons.  
We have shown that a line-of-sight pursuer has a deterministic winning strategy in these environments. Line-of-sight pursuit must alternate between  searching and chasing. We have taken advantage of the existence a natural traversal path in our environments, which allows the pursuer to guard the polygon from left to right. In particular, we associate a unique frame of reference to each point in a scallop polygon by using the sweep line as the vertical axis. 

We believe that our techniques  can be used to show that a line-of-sight pursuer has a winning strategy in all \style{strictly sweepable polygons}. The general idea is as follows. First, we partition any strictly sweepable  polygon into a series of  monotone and scallop polygons. The  pursuer traverses  the polygon, using  the  search/rook algorithms  that are appropriate for the current sub-environment.  To avoid capture, the evader must escape by moving into the next polygon in the decomposition. However, the most delicate task is to handle the transition between subpolygons: we must  ensure that  the pursuer prevents  recontamination of previous  cleared areas. This is far more complicated than dealing with the sweeping angles of a scallop polygon. For example, the centers for the sub-scallop polygons may be on  opposite sides of the polygon, so ``up'' and ``down'' are no longer global definitions. As a result, preventing recontamination of previous cleared areas becomes much  more difficult.  

More  generally, determining the full class of pursuer-win environments remains a challenging open question. 
In terms of negative results, there are \style{weakly monotone polygons} which are evader-win.\cite{noori+isler} A simply connected polygon is weakly monotone with respect to vertices $s$ and $t$ when two particles can walk along the boundary chain, one from $s$ to $t$ and one from $t$ to $s$ such that the range of the directions of these walks is less than $\pi$ radians.\cite{hefferman}
One milestone for characterizing pursuer-win environments would be resolving whether the family of \style{sweepable polygons} is  pursuer-win. Note that the sweep line is allowed to backtrack as it navigates the sweepable polygon. A potential search path must account for this backtracking, which presents a further challenge to the methods presented herein. We are optimistic that our notion of making headway could generalize to all sweepable polygons, but there are certainly technical challenges to overcome.

Looking beyond sweepable polygons, we are optimistic that our methods could be adapted to the superfamily of  \style{straight walkable polygons}.\cite{icking+klein,tseng} 
The boundary of such  a polygon can be partitioned into two chains between vertices $s$ and $t$, such that we can move two mutually visible points monotonically from $s$ to $t$, one clockwise and the other counterclockwise. For example a \style{spiral polygon}, whose boundary chain consists of one convex chain and one reflex chain, is straight walkable.\cite{orourke+suri+toth} Note that a spiral polygon is not sweepable whenever the polygon spirals  beyond $\pi$ radians, and a similar statement can be formulated for straight walkable polygons. 
Again, we leave these extensions  for future work.

\section*{Acknowledgments}

This work was supported in part by the Institute for Mathematics and its Applications and
in part by NSF Grant DMS-1156701. Volkan Isler was supported in part by NSF Grant
IIS-0917676. We thank Narges Noori for helpful conversations and feedback. 

\bibliography{pursuit}

\end{document}